\documentclass[11pt]{article}
\usepackage{amsmath,amssymb,amsthm}
\usepackage[letterpaper, margin=1in]{geometry}
\usepackage{todonotes}
\usepackage{enumerate}
\usepackage{hyperref}
\usepackage{mathabx}
\usepackage{blkarray}
\usepackage{tikz}
\usepackage{authblk}
\usepackage{algorithmicx,algorithm}
\usepackage{algpseudocode}
\usepackage{thm-restate}
\usepackage{lineno}
\usetikzlibrary{backgrounds}
\usetikzlibrary{shapes}
\usetikzlibrary{positioning}
\hypersetup{
	pdftitle = {2021distancer},
	pdfauthor = {Jungho Ahn, Jinha Kim, O-joung Kwon}
}
\newcommand\abs[1]{\lvert #1\rvert}

\newtheorem{THM}{Theorem}[section]
\newtheorem{LEM}[THM]{Lemma}
\newtheorem{COR}[THM]{Corollary}
\newtheorem{PROP}[THM]{Proposition}

\newtheorem{CLM}{Claim}

\newtheorem{OBS}[THM]{Observation}

\theoremstyle{remark}

\newenvironment{subproof}[1][\proofname]{
	
	\begin{proof}[#1]}{\end{proof}
}
\theoremstyle{definition}

\newcommand{\dist}{\mathrm{dist}}
\newcommand{\wc}{\mathrm{wcol}}
\newcommand{\wre}{\mathrm{WReach}}
\newcommand{\idx}{\operatorname{index}}

\usepackage{boxedminipage}

\begin{document}
\title{Unified almost linear kernels for \\ generalized covering and packing problems \\ on nowhere dense classes}
\author[1,2]{Jungho~Ahn}
\author[2]{Jinha~Kim}
\author[3,2]{O-joung~Kwon}
\affil[1]{Department of Mathematical Sciences, KAIST, Daejeon,~South~Korea}
\affil[2]{Discrete Mathematics Group, Institute for Basic Science (IBS), Daejeon, South~Korea}
\affil[3]{Department of Mathematics, Hanyang University, Seoul,~South~Korea}
\affil[ ]{\small\textit{Email addresses:}
	\texttt{junghoahn@kaist.ac.kr},
	\texttt{jinhakim@ibs.re.kr},
	\texttt{ojoungkwon@hanyang.ac.kr}
}
\footnotetext[1]{Jungho Ahn, Jinha Kim, and O-joung Kwon are supported by the Institute for Basic Science (IBS-R029-C1). O-joung Kwon is also supported by the National Research Foundation of Korea (NRF) grant funded by the Ministry of Education (No. NRF-2018R1D1A1B07050294) and (No. NRF-2021K2A9A2A11101617).}

\maketitle

\begin{abstract}
    Let $\mathcal{F}$ be a family of graphs, and let $p$ and $r$ be nonnegative integers.
    The \textsc{$(p,r,\mathcal{F})$-Covering} problem asks whether for a graph $G$ and an integer $k$, there exists a set $D$ of at most $k$ vertices in $G$ such that $G^p\setminus N_G^r[D]$ has no induced subgraph isomorphic to a graph in $\mathcal{F}$, where $G^p$ is the $p$-th power of $G$ and $N^r_G[D]$ is the set of all vertices in $G$ at distance at most $r$ from $D$ in $G$.
    The \textsc{$(p,r,\mathcal{F})$-Packing} problem asks whether for a graph $G$ and an integer $k$, $G^p$ has $k$ induced subgraphs $H_1,\ldots,H_k$ such that each $H_i$ is isomorphic to a graph in $\mathcal{F}$, and for distinct $i,j\in \{1, \ldots, k\}$, the distance between $V(H_i)$ and $V(H_j)$ in $G$ is larger than $r$.
    The \textsc{$(p,r,\mathcal{F})$-Covering} problem generalizes \textsc{Distance-$r$ Dominating Set} and \textsc{Distance-$r$ Vertex Cover}, and 
    the \textsc{$(p,r,\mathcal{F})$-Packing} problem generalizes \textsc{Distance-$r$ Independent Set} and \textsc{Distance-$r$ Matching}.
    By taking $(p',r',\mathcal{F}')=(pt, rt, \mathcal{F})$, we may formulate the \textsc{$(p,r,\mathcal{F})$-Covering} and \textsc{$(p, r, \mathcal{F})$-Packing} problems on the $t$-th power of a graph.
    Moreover, \textsc{$(1,0,\mathcal{F})$-Covering} is the \textsc{$\mathcal{F}$-Free Vertex Deletion} problem, and \textsc{$(1,0,\mathcal{F})$-Packing} is the \textsc{Induced-$\mathcal{F}$-Packing} problem. 
    
    We show that for every fixed nonnegative integers $p,r$ and every fixed nonempty finite family $\mathcal{F}$ of connected graphs, the \textsc{$(p,r,\mathcal{F})$-Covering} problem with $p\leq2r+1$ and the \textsc{$(p,r,\mathcal{F})$-Packing} problem with $p\leq2\lfloor r/2\rfloor+1$ admit almost linear kernels on every nowhere dense class of graphs, and admit linear kernels on every class of graphs with bounded expansion, parameterized by the solution size $k$. We obtain the same kernels for their annotated variants.
    As corollaries, we prove that \textsc{Distance-$r$ Vertex Cover}, \textsc{Distance-$r$ Matching}, \textsc{$\mathcal{F}$-Free Vertex Deletion}, and \textsc{Induced-$\mathcal{F}$-Packing} for any fixed finite family $\mathcal{F}$ of connected graphs admit almost linear kernels on every nowhere dense class of graphs and linear kernels on every class of graphs with bounded expansion.
    Our results extend the results for \textsc{Distance-$r$ Dominating Set} by Drange et al. (STACS 2016) and Eickmeyer et al. (ICALP 2017), and the result for \textsc{Distance-$r$ Independent Set} by Pilipczuk and Siebertz (EJC 2021).
\end{abstract}

\section{Introduction}\label{sec:intro}

The \textsc{Dominating Set} problem is one of the classical NP-hard problems which asks whether a given graph $G$ contains a set of at most $k$ vertices whose closed neighborhood contains all the vertices of $G$. A natural variant of it is the \textsc{Distance-$r$ Dominating Set} problem which asks whether a given graph $G$ contains a set of at most $k$ vertices such that every vertex of $G$ is at distance at most $r$ from one of these vertices. This problem has many applications in  facility location problems; for instance, when we design a city, we want to afford fire stations so that every building is within $r$ city blocks from one of the fire stations. By solving \textsc{Distance-$r$ Dominating Set} problem, we can test whether at most $k$ fire stations may cover the city. 

\textsc{Dominating Set} has been intensively studied in the context of fixed-parameter algorithms. In a parameterized problem, we are given an instance $(x,k)$ where $k$ is a parameter, and the central question is whether a parameterized problem admits an algorithm that runs in time $f(k)\cdot |x|^c$ for some computable function $f$ and a constant $c$. Such an algorithm is called a \emph{fixed-parameter} algorithm parameterized by $k$. 
It is known that \textsc{Dominating Set} is W[2]-complete parameterized by $k$~\cite{DF1992,DF1995-1,DF1995-2}, meaning that it cannot be solved by a fixed-parameter algorithm, unless an unexpected collapse occurs in the parameterized complexity hierarchy.
So, it is natural to restrict graph classes and see whether a fixed-parameter algorithm exists. \textsc{Dominating Set} admits a fixed-parameter algorithm on planar graphs~\cite{Dorn2006, FominT2006}, and the project of finding larger sparse graph classes on which fixed-parameter algorithms for \textsc{Dominating Set} exist has been studied intensively, see~\cite{DemaineFHT2005, AlonG2009, Gutner2009,  GeevargheseVS2012, GroheKS2017, FominLST2018, TelleV2019}. 

A kernelization algorithm for a parameterized problem takes an instance $(x,k)$ and outputs an equivalent instance $(x', k')$ in time polynomial in $|x|+k$, where $|x'|+k'\le g(k)$ for some function $g$, called the size of the kernel. If $g$ is polynomial (resp. linear), then such an algorithm is called a polynomial (resp. linear) kernel. Observe that if there is a polynomial kernel, then the problem admits a fixed-parameter algorithm. After applying a kernelization, any brute-force search runs in time bounded by a function of~$k$. Therefore, the existence of a polynomial kernel or a linear kernel is an interesting research direction. In particular, the existence of linear kernels for \textsc{Dominating Set} on sparse graph classes have been investigated.

One of the first results is a linear kernel for \textsc{Dominating Set} on planar graphs due to Alber, Fellows, and Niedermeier~\cite{AlberFN2004}.
It has been generalized to classes of bounded genus graphs~\cite{FominT2004},  $H$-minor free graphs~\cite{FominLS2012}, and $H$-topological minor free graphs~\cite{FominLST2018}. 
Drange et al.~\cite{KnS2016} extended the previous results to classes of graphs with bounded expansion for general \textsc{Distance-$r$ Dominating Set}, and Eickmeyer et al.~\cite{NCK2017} obtained almost linear kernels for \textsc{Distance-$r$ Dominating Set} on even larger nowhere dense classes.
Classes of graphs with bounded expansion and nowhere dense classes of graphs were introduced by Ne\v{s}et\v{r}il and Ossona de Mendez~\cite{NesetrilO2008}, which are defined in terms of shallow minors and capture most of well-studied sparse graph classes.  
 
 \textsc{Independent Set} is another classic NP-hard problem which asks to find a set of $k$ vertices in a given graph whose pairwise distance is more than~$1$, and \textsc{Distance-$r$ Independent Set} is the problem obtained by replacing $1$ with $r$.
 It is known that \textsc{Independent Set} is W[1]-complete parameterized by $k$~\cite{DF1995-2}.
 The distance variations of \textsc{Dominating Set} and \textsc{Independent Set} are closely related, in a sense that the size of a distance-$2r$ independent set is a lower bound of the minimum size of a distance-$r$ dominating set. Dvo\v{r}\'{a}k~\cite{Dvorak2013} presented an approximation algorithm for \textsc{Distance-$r$ Dominating Set}, which outputs a set of size bounded by a function of the $2r$-weak coloring number and the maximum size of a distance-$2r$ independent set. Pilipczuk and Siebertz~\cite{PS2021} recently presented an almost linear kernel for \textsc{Distance-$r$ Independent Set} on nowhere dense classes of graphs.
 
 We remark that for a fixed $r$, both of the \textsc{Distance-$r$ Dominating Set} and \textsc{Distance-$r$ Independent Set} can be expressed in first-order logic.
 Thus, by the meta-theorem of Grohe, Kreutzer, and Siebertz~\cite{GroheKS2017}, they are almost-linear-time fixed-parameter tractable on every nowhere dense class of graphs.
 Fabia\'{n}ski, Pilipczuk, Siebertz, and Toru\'{n}czyk~\cite{FPST2019} presented linear-time fixed-parameter algorithms for the \textsc{Distance-$r$ Dominating Set} problem on various graph classes, including powers of nowhere dense classes and map graphs, and a linear-time fixed-parameter algorithm for the \textsc{Distance-$r$ Independent Set} problem on every nowhere dense class.
 Einarson and Reidl~\cite{ER2020} presented linear kernels for generalizations of the \textsc{Distance-$r$ Dominating Set} and \textsc{Distance-$r$ Independent Set} problems on classes with bounded expansion, called the \textsc{$(r,c)$-Dominating Set}, \textsc{$(r,c)$-Scattered Set}, and \textsc{$(r,[\lambda,\mu])$-Domination} problems.

A natural question is whether there are  linear/polynomial kernels for other problems on classes of graphs with bounded expansion and nowhere dense classes of graphs.
Meta-type kernelization results have been studied for graphs on bounded genus~\cite{BodlaenderFLPT2016}, $H$-minor free graphs~\cite{FominLST2020}, $H$-topological minor free graphs~\cite{KimLPRRSS2016}, graph classes with bounded expansion and nowhere dense classes of graphs~\cite{Gajarsky2017}.
Note that the last result by Gajarsky et al.~\cite{Gajarsky2017} is to obtain kernelizations parameterized by the size of a modulator to constant treedepth, and not by the solution size. 
At the moment, no much variations of Drange et al.~\cite{KnS2016}, Eickmeyer et al.~\cite{NCK2017} for \textsc{Distance-$r$ Dominating Set} and Pilipczuk and Siebertz~\cite{PS2021} for \textsc{Distance-$r$ Independent Set} have been studied.
In this paper, we consider generic problems to hit certain finite graphs by the $r$-th neighborhood of a set of vertices, and present almost linear kernels. \textsc{Distance-$r$ Dominating Set} can be seen as a problem of hitting all singleton vertices by the $r$-th neighborhood of a set of vertices.

Let $\mathcal{F}$ be a family of graphs, and let $p$ and $r$ be nonnegative integers.
For a graph $G$, let $G^p$ be the graph with vertex set $V(G)$ such that distinct vertices $v$ and $w$ are adjacent in $G^p$ if and only if the distance between $v$ and $w$ in $G$ is at most $p$, and let $N^r_G[D]$ be the set of all vertices in $G$ at distance at most $r$ from $D$ in $G$.
For a graph $G$ and a set $A$ of vertices in $G$, a set $D$ of vertices in $G$ is a \emph{$(p,r,\mathcal{F})$-cover} of $A$ in $G$ if there is no set $X\subseteq A\setminus N_G^r[D]$ such that $G^p[X]$ is isomorphic to a graph in $\mathcal{F}$.
 A $(p,r,\mathcal{F})$-cover of $A$ in $G$ is called a $(p,r,\mathcal{F})$-cover of $G$ if $A=V(G)$.

The \textsc{$(p,r,\mathcal{F})$-Covering} problem asks whether for a graph $G$ and an integer $k$, there is a $(p,r,\mathcal{F})$-cover of $G$ having size at most $k$.
Its annotated variant, named \textsc{Annotated $(p,r,\mathcal{F})$-Covering}, asks whether for a graph $G$ and $A\subseteq V(G)$ and an integer $k$, there is a $(p,r,\mathcal{F})$-cover of $A$ in $G$ having size at most $k$.
Note that the \textsc{$(p,r,\mathcal{F})$-Covering} problem generalizes the \textsc{Distance-$r$ Dominating Set} problem, because the latter is identical to the \textsc{$(1,r,\{K_1\})$-Covering} problem.
We denote by $\gamma_{p,r}^\mathcal{F}(G,A)$ the minimum size of a $(p,r,\mathcal{F})$-cover of $A$ in $G$, and $\gamma_{p,r}^\mathcal{F}(G):=\gamma_{p,r}^\mathcal{F}(G,V(G))$.

Our main results are the following.
For classes of graphs with bounded expansion, we obtain linear kernels.
We denote by $\mathbb{N}$ the set of nonnegative integers, and by $\mathbb{R}_+$ the set of positive reals.

\begin{restatable}{THM}{domnowhere}\label{thm:ker1 origin}
    For every nowhere dense class $\mathcal{C}$ of graphs, there exists a function $g_{\mathrm{cov}}:\mathbb{N}\times\mathbb{N}\times\mathbb{R}_+\to\mathbb{N}$ satisfying the following.
    For every nonempty family $\mathcal{F}$ of connected graphs with at most $d$ vertices, $p,r\in\mathbb{N}$ with $p\leq2r+1$, and $\varepsilon>0$, there exists a polynomial-time algorithm that given a graph $G\in\mathcal{C}$ and $k\in\mathbb{N}$, either
    \begin{enumerate}
        \item[$\bullet$]    correctly decides that $\gamma_{p,r}^\mathcal{F}(G)>k$, or
        \item[$\bullet$]    outputs a graph $G'$ with $\abs{V(G')}\leq g_{\mathrm{cov}}(r,d,\varepsilon)\cdot k^{1+\varepsilon}$ such that $\gamma_{p,r}^\mathcal{F}(G)\leq k$ if and only if $\gamma_{p,r}^\mathcal{F}(G')\leq k+1$.
    \end{enumerate}
\end{restatable}

A \emph{$(p,r,\mathcal{F})$-packing} of $A$ in $G$ is a family of subsets of $A$, say $A_1,\ldots,A_\ell$, such that each $G^p[A_i]$ is isomorphic to a graph in $\mathcal{F}$, and for all $1\leq i<j\leq\ell$, the distance between $A_i$ and $A_j$ in $G$ is more than $r$.
A \emph{$(p,r,\mathcal{F})$-packing} of $A$ in $G$ is called a $(p,r,\mathcal{F})$-packing of $G$ when $A=V(G)$.

The \textsc{$(p,r,\mathcal{F})$-Packing} problem asks whether for a graph $G$ and an integer $k$, there is a $(p,r,\mathcal{F})$-packing of $G$ having size at least $k$.
Its annotated variant, named \textsc{Annotated $(p,r,\mathcal{F})$-Packing},  asks whether for a graph $G$ and $A\subseteq V(G)$ and an integer $k$, there is a $(p,r,\mathcal{F})$-packing of $A$ in $G$ having size at least $k$.
Note that the \textsc{$(p,r,\mathcal{F})$-Packing} problem generalizes the \textsc{Distance-$r$ Independent Set} problem, because the latter is identical to the \textsc{$(1,r,\{K_1\})$-Packing} problem.
We denote by $\alpha_{p,r}^\mathcal{F}(G,A)$ the maximum size of a $(p,r,\mathcal{F})$-packing of $A$ in $G$, and $\alpha_{p,r}^\mathcal{F}(G):=\alpha_{p,r}^\mathcal{F}(G,V(G))$.

For the \textsc{$(p,r,\mathcal{F})$-Packing} problem, we prove the following.

\begin{restatable}{THM}{pcknowhere}\label{thm:ker2 origin}
    For every nowhere dense class $\mathcal{C}$ of graphs, there exists a function $g_{\mathrm{pck}}:\mathbb{N}\times\mathbb{N}\times\mathbb{R}_+\to\mathbb{N}$ satisfying the following.
    For every nonempty family $\mathcal{F}$ of connected graphs with at most~$d$ vertices, $p,r\in\mathbb{N}$ with $p\leq2\lfloor r/2\rfloor+1$, and $\varepsilon>0$, there exists a polynomial-time algorithm that given a graph $G\in\mathcal{C}$ and $k\in\mathbb{N}$, either
    \begin{enumerate}
        \item[$\bullet$]    correctly decides that $\alpha_{p,r}^\mathcal{F}(G)=0$,
        \item[$\bullet$]    correctly decides that $\alpha_{p,r}^\mathcal{F}(G)>k$, or
        \item[$\bullet$]    outputs a graph $G'$ with $\abs{V(G')}\leq g_{\mathrm{pck}}(r,d,\varepsilon)\cdot k^{1+\varepsilon}$ such that $\alpha_{p,r}^\mathcal{F}(G)\geq k$ if and only if $\alpha_{p,r}^\mathcal{F}(G')\geq k+1$.
    \end{enumerate}
\end{restatable}

   We also obtain almost linear kernels for \textsc{Annotated $(p, r, \mathcal{F})$-Covering} and \textsc{Annotated $(p, r, \mathcal{F})$-Packing} on nowhere dense classes, with the same assumptions on $p$, $r$, and $\mathcal{F}$ as for their original problems. In fact, we first obtain kernels for annotated versions, and to obtain kernels for original problems, we apply the kernelization algorithm for $(G, V(G))$ as an annotated version, and then translate the output $(G', A)$ to an equivalent instance. 

\bigskip
\noindent\textbf{Applications}. Canales, Hern\'{a}ndez, Martins, and Matos~\cite{CHMM2015} introduced a distance-$r$ vertex cover and a distance-$r$ guarding set.
For a graph $G$ and a positive integer $r$, a set $D\subseteq V(G)$ is a \emph{distance-$r$ vertex cover} if $G\setminus N_G^r[D]$ has no edge, and a \emph{distance-$r$ guarding set} if $G\setminus N_G^{r-1}[D]$ has no triangle.
For a positive integer $k$, the \textsc{Distance-$r$ Vertex Cover} problem asks whether a graph $G$ has a distance-$r$ vertex cover of size at most $k$.
Similarly, the \textsc{Distance-$r$ Guarding Set} problem asks whether a graph $G$ has a distance-$r$ guarding set of size at most $k$.
Observe that the \textsc{Distance-$r$ Vertex Cover} problem can be formulated as the \textsc{$(1,r,\{K_2\})$-Covering} problem, the \textsc{Distance-$r$ Guarding Set} problem for $r\geq1$ can be formulated as the \textsc{$(1,r-1,\{K_3\})$-Covering} problem.

The notion of guarding originates from a problem posed by Victor Klee and solved by Chv\'atal~\cite{Chvatal1975}, which concerns the minimum number of guards needed to control an art gallery. In a plane triangulation, assume that a guard on a vertex may control triangles at distance at most $r$ from the vertex. We may ask whether at most $k$ guards are sufficient to control. This problem is formulated by the \textsc{Distance-$r$ Guarding Set} problem. 

Dallard, Krbezlija, and Milani\v{c}~\cite{DKM2021} characterized graphs $H$ such that \textsc{Distance-$r$ Vertex Cover} is polynomial-time solvable on graphs having no induced subgraph isomorphic to $H$.
Canales, Hern\'{a}ndez, Martins, and Matos~\cite{CHMM2015} and Alvardo, Dantas, and Rautenbach~\cite{ADR2015} computed the distance-$r$ vertex cover number and guarding number of maximal outerplanar graphs.
As applications of Theorem~\ref{thm:ker1 origin}, we obtain almost linear kernels for the \textsc{Distance-$r$ Vertex Cover} and \textsc{Distance-$r$ Guarding Set} problems on every nowhere dense class of graphs, and linear kernels for the same problems on every class of graphs with bounded expansion.
In particular, we may apply our algorithm on plane triangulations.

A matching of a graph $G$ is a set $M$ of edges of $G$ such that no two edges in $M$ share an end.
For a positive integer $r$, a \emph{distance-$r$ matching} of $G$ is a matching $M$ of $G$ such that for distinct edges $u_1u_2,v_1v_2\in M$, $\min\{\dist_G(u_i,v_j):i,j\in[2]\}\geq r$.
Note that every distance-$1$ matching is nothing but a matching, and every distance-$2$ matching is an induced matching.
For a positive integer $k$, the \textsc{Distance-$r$ Matching} problem asks whether a graph $G$ has a distance-$r$ matching of size at least $k$.
Observe that the \textsc{Distance-$1$ Matching} is nothing but to find a matching of size at least $k$, so can be solved in polynomial time by the algorithm of~\cite{blossom}.
The \textsc{Distance-$r$ Matching} for $r\geq2$ can be formulated as the \textsc{$(1,r-1,\{K_2\})$-Packing} problem.

There are several structural results about distance-$r$ matching~\cite{WY2005,AP2011,AFS2020,LY2022,KM2012} and hardness results about the \textsc{Distance-$r$ Matching} problem for $r\in\{2,3\}$~\cite{Cameron1989,BM2011}.
Moser and Sikdar~\cite{MS2009} presented linear kernels for the \textsc{Distance-$2$ Matching} problem on planar graphs and graphs of bounded degree, and a cubic kernel for the same problem on graphs of girth at least $6$.
Later, Kanj, Pelsmajer, Schaefer, and Xia~\cite{Kanj2011} presented a kernel of size $40k$ for the \textsc{Distance-$2$ Matching} problem on planar graphs.
As applications of Theorem~\ref{thm:ker2 origin}, we obtain an almost linear kernel for the \textsc{Distance-$r$ Matching} problem for $r\geq2$ on every nowhere dense class of graphs, and a linear kernel for the same problem on every class of graphs with bounded expansion.

For a family $\mathcal{F}$ of graphs, the \textsc{$\mathcal{F}$-Free Vertex Deletion} problem asks whether for a graph $G$ and an integer $k$, there is a set $S$ of at most $k$ vertices in $G$ such that $G\setminus S$ has no induced subgraph isomorphic to $\mathcal{F}$. For hereditary classes $\mathcal{C}$ of graphs that are characterized by a family $\mathcal{F}$ of graphs, the problem asks whether one can remove at most $k$ vertices so that the remaining graph is in $\mathcal{C}$.   If all graphs in $\mathcal{F}$ have at most $d$ vertices, then this problem is related to the \textsc{$d$-Hitting Set} problem, and has a kernel of size $\mathcal{O}(k^{d-1})$~\cite{Abu2010}. Our results imply that if $\mathcal{F}$ is a finite set of connected graphs, then \textsc{$\mathcal{F}$-Free Vertex Deletion} admits an almost linear kernel on nowhere dense classes and a linear kernel on classes of graphs with bounded expansion.
This can be applied to the problems of deleting to cographs (\textsc{Cograph Vertex Deletion})~\cite{Nastos2012}, to the disjoint union of complete graphs (\textsc{Cluster Vertex Deletion})~\cite{AprileDFH2021, Tsur2021}, to claw-free graphs (\textsc{Claw-Free Vertex Deletion})~\cite{BonomoNJO2020}, and so on.

For a graph $H$, the \textsc{$H$-Matching} problem asks for a graph $G$ and an integer $k$, whether $G$ has $k$ vertex-disjoint subgraphs which are isomorphic to $H$.
For every positive integer $d$, let $P_d$ be a path on $d$ vertices.
Dell and Marx~\cite{DM2012} presented a kernel for the \textsc{$P_3$-Matching} problem with $O(k^{2.5})$ edges, and a unified kernel for the \textsc{$P_d$-Matching} problem with $O(d^{d^2}d^7k^3)$ vertices.
They also showed that for every integer $d\geq3$ and every positive real $\varepsilon$, under some complexity hypothesis, \textsc{$K_d$-Matching} does not have kernels with $O(k^{d-1-\varepsilon})$ edges. 
By taking $\mathcal{F}$ as the set of all graphs on $\abs{V(H)}$ vertices that contain $H$ as a subgraph, we can formulate the \textsc{$H$-Matching} problem as the \textsc{$(1,0,\mathcal{F})$-Packing} problem.
Generally, we may consider \textsc{Induced-$\mathcal{F}$-Packing} which asks whether $G$ has $k$ vertex-disjoint induced subgraphs each isomorphic to some graph in $\mathcal{F}$.
As applications of Theorem~\ref{thm:ker2 origin}, we obtain almost linear kernels for \textsc{$K_d$-Matching}, \textsc{$P_d$-Matching} for every fixed positive integer $d$, and \textsc{Induced-$\mathcal{F}$-Packing} for every fixed finite family $\mathcal{F}$ of connected graphs on every nowhere dense class of graphs, and linear kernels for the same problems on every class of graphs with bounded expansion.

We may formulate the \textsc{$(p,r,\mathcal{F})$-Covering} and \textsc{$(p,r,\mathcal{F})$-Packing} problems on fixed powers of a given graph.  
Formally, for a fixed positive integer $t$, 
the \textsc{$(pt, rt, \mathcal{F})$-Covering} problem on a graph $G$ is exactly same as the \textsc{$(p, r, \mathcal{F})$-Covering} problem on its $t$-th power $G^t$. Therefore, our result provides the existence of almost linear kernels for both problems on $t$-th powers $G^t$ of graphs from a nowhere dense class of graphs, assuming that the original graph $G$ is given. However, if the power $G^t$ is only given, then we need to find the graph $G$ to apply our kernelization algorithm. 

The annotated variants provide interesting applications. 
A \emph{map graph} is a graph obtained as the intersection model of finitely many simply connected and internally disjoint regions of a plane.
Chen, Grigni, and Papadimitriou~\cite[Theorem 2.2, Lemma 2.3]{ChenGP2002} showed that $G$ is a map graph if and only if there is a planar bipartite graph $H$ with bipartition $(V(G), P)$ such that $\abs{P}\le 4\abs{V(G)}$ and for $u,v\in V(G)$, $uv\in E(G)$ if and only if $up,vp\in E(H)$ for some $p\in F$.
Observe that for $u,v\in V(G)$, $\dist_G(u,v)\le r$ if and only if $\dist_H(u,v)\le 2r$.
Therefore, the \textsc{$(p, r, \mathcal{F})$-Packing} problem on $G$ is the same as the \textsc{Annotated $(2p, 2r, \mathcal{F})$-Packing} problem on $(H, V(G))$.
As the class of 
planar graphs is of bounded expansion,
the annotated variant of Theorem~\ref{thm:ker2 origin} implies that there is a linear kernel for the \textsc{$(p,r,\mathcal{F})$-Packing} problem on map graphs, assuming that the planar bipartite graph $H$ is given. Note that we cannot adapt to the \textsc{$(p,r,\mathcal{F})$-Covering} problem because in the annotated version, we do not require that the solution has to be in the selected set $A$. It would be interesting if one can obtain linear kernels for \textsc{$(p, r, \mathcal{F})$-Packing} and  \textsc{$(p,r,\mathcal{F})$-Covering}  for map graphs $G$ without having $H$ at hand.

\bigskip
\noindent\textbf{Proof ideas}. We first summarize proof ideas for Theorem~\ref{thm:ker1 origin}.
The proof falls into two major steps.
Given an instance $(G,k)$ for the \textsc{$(p,r,\mathcal{F})$-Covering} problem, we first find a small set $Z$ and its superset $Y$ whose size is $O(\abs{Z}^{1+\varepsilon})$ for any small $\varepsilon>0$ such that $\gamma_{p,r}^\mathcal{F}(G)\leq k$ if and only if $\gamma_{p,r}^\mathcal{F}(G[Y],Z)\leq k$.
So, we actually find an instance for the \textsc{Annotated $(p,r,\mathcal{F})$-Covering} problem, which keeps the information that the covering number is at most the parameter $k$.
After finding these sets, we add a small gadget to make an instance $(G',k)$ which is equivalent to the input instance $(G,k)$ with respect to the \textsc{$(p,r,\mathcal{F})$-Covering} problem.

For the first step, we show that given an instance $(G,A,k)$ for the \textsc{Annotated $(p,r,\mathcal{F})$-Covering} problem, if the size of $A$ is large, then we can either show that the input instance is a no-instance, or find in polynomial time a vertex $z$ in $A$ such that $(G,A\setminus \{v\},k)$ is equivalent to $(G,A,k)$ with respect to the \textsc{Annotated $(p,r,\mathcal{F})$-Covering} problem.
If the size of $A$ is small, we can find a superset $Y$ of $A$ such that the size of $Y$ is $O(\abs{A}^{1+\varepsilon})$ and $\gamma_{p,r}^\mathcal{F}(G,A)=\gamma_{p,r}^\mathcal{F}(G[Y],Z)$.
Starting with $A=V(G)$, by recursively finding the vertex $z$, we lead to desired sets $Z$ and $Y$ at the end.
Finding the superset $Y$ of $Z$ can be easily dealt with known results for nowhere dense classes of graphs, so let us focus on finding a desired vertex $z$.

If the size of $A$ is large, then we can use the \emph{uniformly quasi-wideness} property of nowhere dense class of graphs.
Basically, the uniformly quasi-wideness property tells us that for any nowhere dense class of graphs and a graph $G$ in the class, if we are given a large set $A\subseteq V(G)$, then for any nonnegative $r$, we can find a tiny set $S$ of vertices and a large set $L\subseteq A\setminus S$ which is distance-$r$ independent in $G\setminus S$.
For $r':=2pd+3r$, we find a tiny set $S$ and a large set $L\subseteq A$ which is distance-$r'$ independent in $G\setminus S$.
We are going to find a desired vertex $z$ from $L$, and to do this, in the reasoning, we heavily rely on the fact that all vertices in $L$ are far from each other in $G\setminus S$.

As another tool beside the uniformly quasi-wideness property, we equip with an approximation algorithm for the $(p,r,\mathcal{F})$-covering number.
Based on the polynomial-time $O(\log\mathsf{OPT})$-factor approximation algorithm for the \textsc{$d$-Hitting Set} problem on bounded VC-dimension family of sets, we can easily design a polynomial-time $O(\mathsf{OPT}^\varepsilon)$-factor approximation algorithm for the \textsc{$(p,r,F)$-Covering} problem, so we can find an approximate solution $X$ for the problem in polynomial time.
If the size of $X$ is large, then it immediately tells us that the instance has large covering number, so we can reject the instance.
Thus, we may assume that the size of $X$ is small.
For the technical reasons, we inflate $X$ by adding a small number of vertices, but the resulting set would have size $O(k^{1+\varepsilon})$.

Now, with these tools, we can find a desired vertex $z$.
If some vertex in $L$ is not involved in a set of vertices which induces a graph in $\mathcal{F}$ in $G^p$, then the vertex is irrelevant for our problem, so we can safely choose the vertex as $z$.
Otherwise, every vertex $u\in L$ is contained in some set $B_u$ of vertices such that $G^p[B_u]$ is isomorphic to a graph in $\mathcal{F}$.

Although $B_u$ induces a connected graph in $G^p$, $B_u$ might induce a disconnected graph in $(G\setminus S)^p$.
So for each $u\in L$, we carefully observe some pattern obtained by $(G\setminus S)^p[B_u\setminus S]$.
The pattern consists of the labelled graph on at most $d$ vertices, to which $(G\setminus S)^p[B_u\setminus S]$ is isomorphic, and the $(2r+1)$-projection profile of each vertex $u'\in B_u\setminus S$ toward $S$, which is a function containing distance information from $u'$ to $S$ up to $2r+1$.
Since the size of $L$ is large and the number of graphs on at most $d$ vertices and the number of $r$-projection profiles toward $S$ are constants in $r$ and $d$, we can find many vertices of $L$ seeking the same patterns.

With these in hand, we choose a vertex $z$ appropriately and can show that $z$ is indeed a desired vertex, because otherwise we can derive a contraction by carefully concatenating several paths in $G\setminus S$ to obtain a short walk between two vertices in $L$, which contradicts the assumption that $L$ is distance-$r'$ independent in $G\setminus S$, and this is the end of the first step.

After the first step, we add a small gadget to the resulting graph $G[Y]$ to construct a graph $G'$ having at most $O(\abs{Y})$ vertices such that $\gamma_{p,r}^\mathcal{F}(G)\leq k$ if and only if $\gamma_{p,r}^\mathcal{F}(G')\leq k+1$.
To do this, we define a $(p,\mathcal{F})$-critical graph and show its existence with a polynomial-time algorithm to find it.
A \emph{$(p,\mathcal{F})$-critical graph} $H$ is a graph whose $p$-th power contains an induced subgraph isomorphic to a graph in $\mathcal{F}$ and for every $v\in V(H)$, the $p$-th power of $H\setminus v$ has no such induced subgraph.
Thus, basically, a $(p,\mathcal{F})$-critical graph is a graph where every vertex is involved to form a set of vertices inducing a graph in $\mathcal{F}$ in the $p$-th power.
We show that every vertex of a $(p,\mathcal{F})$-critical graph $H$ forms a $(p,\lfloor p/2\rfloor,\mathcal{F})$-cover of $H$.
Using this property, we can show that for the resulting graph $G'$, $\gamma_{p,r}^\mathcal{F}(G)\leq k$ if and only if $\gamma_{p,r}^\mathcal{F}(G')\leq k+1$, and this will prove Theorem~\ref{thm:ker1 origin}.

\bigskip
Now, we summarize proof ideas for Theorem~\ref{thm:ker2 origin}.
The major steps of the proof is similar to those of Theorem~\ref{thm:ker1 origin}.
Given an instance $(G,k)$ for the \textsc{$(p,r,\mathcal{F})$-Packing} problem, we first find a small set $Z$ and its superset $Y$ whose size is $O(\abs{Z}^{1+\varepsilon})$ for any small $\varepsilon>0$ such that $\alpha_{p,r}^\mathcal{F}(G)\geq k$ if and only if $\alpha_{p,r}^\mathcal{F}(G[Y],Z)\geq k$, so we actually find an instance for the \textsc{Annotated $(p,r,\mathcal{F})$-Packing} problem, which keeps the information that the packing number is at least the parameter $k$.
Similarly as in the previous case, after finding the sets, we add a small gadget to make an instance $(G',k)$ which is equivalent to the input instance $(G,k)$ with respect to \textsc{$(p,r,\mathcal{F})$-Packing} problem.

For the first step, we show that given an instance $(G,A,k)$ for the \textsc{Annotated $(p,r,\mathcal{F})$-Packing} problem, if the size of $A$ is large, then we can either show that the input instance is a no-instance, or find in polynomial time a vertex $z$ in $A$ such that $\alpha_{p,r}^\mathcal{F}(G,A)\geq k$ if and only if $\alpha_{p,r}^\mathcal{F}(G,A\setminus\{z\})\geq k$.
If the size of $A$ is small, then we can find a superset $Y$ of $A$ such that the size of $Y$ is $O(\abs{A}^{1+\varepsilon})$ and that $\alpha_{p,r}^\mathcal{F}(G)\geq k$ if and only if $\alpha_{p,r}^\mathcal{F}(G[Y],Z)\geq k$.
Starting with $A=V(G)$, by recursively finding the vertex $z$, we lead to desired sets $Z$ and $Y$ at the end.
Finding the superset $Y$ of $Z$ can be easily dealt with the known results, so let us focus on finding a desired vertex.

To find a desired vertex $z$, similar to the \textsc{$(p,r,\mathcal{F})$-Covering} problem, we use the uniformly quasi-wideness to find a tiny set $S$ and a large set $L\subseteq V(G)\setminus S$ which is distance-$r'$ independent in $G\setminus S$ for $r':=4pd+3r$.
We also design polynomial-time approximation algorithms and argue by the pattern argument to choose the vertex $z$.

Although all these steps look similar to those of the \textsc{$(p,r,\mathcal{F})$-Covering} problem, the plot of the proof is slightly different, because now the problem we dealt with is a maximization problem.
Note that for every vertex $z\in L$, $\alpha_{p,r}^\mathcal{F}(G,A\setminus\{z\})\geq k$ always implies $\alpha_{p,r}^\mathcal{F}(G,A)\geq k$ by definition.
Thus, we may assume that $G$ has a $(p,r,\mathcal{F})$-packing $I$ of $A$ having size at least $k$.

If we can find a vertex $z\in L$ so that there exists a $(p,r,\mathcal{F})$-packing $I'$ of $A$ which has the same size as $I$ and avoids $z$, then the packing $I'$ is a $(p,r,\mathcal{F})$-packing of $A\setminus\{z\}$, and it implies that $\alpha_{p,r}^\mathcal{F}(G,A\setminus\{z\})\geq k$.
Thus, we may assume that there exists no such $z$.
It tells us the following two things.
First, $I$ contains an element $B_z$ including $z$.
Second, for every $u\in L$, no matter how we find a set $B_u$ of vertices which includes $z$ and induces a subgraph of $G^p$ isomorphic to a graph in $\mathcal{F}$, there exists an element in $I\setminus\{B_z\}$ which is at distance at most $r$ from $B_u$, because otherwise we can substitute $B_z$ with $B_u$ from $I$ to obtain a $(p,r,\mathcal{F})$-packing of $A\setminus\{z\}$ in $G$.
In other words, for each $u\in L$, the set $B_u$ is connected to some element $C_u\in I$ by a path of $G$ having length at most $r$.
However, each element in $I$ has at most $d$ vertices and all distinct vertices in $L$ cannot be connected by paths of $G\setminus S$ having length at most $r'>2r$.
Since the size of $L$ is large, we can show that $L$ contains many vertices such that for distinct such vertices $v$ and $w$, the assigned elements $C_v$ and $C_w$ are distinct.

With these in hand, we can derive a contradiction and show that there exists a $(p,r,\mathcal{F})$-packing of $A$ which has the same size as $I$ and avoids $z$.

After the first step, we add a small gadget to the resulting graph $G[Y]$ to construct a graph $G'$ having at most $O(\abs{Y})$ vertices such that $\alpha_{p,r}^\mathcal{F}(G)\geq k$ if and only if $\alpha_{p,r}^\mathcal{F}(G')\geq k+1$.
The construction is almost same as the previous construction for the \textsc{$(p,r,\mathcal{F})$-Covering} problem, which uses a $(p,\mathcal{F})$-critical graph, but because of some parity issue, we modify the construction a bit.
Using the property of a $(p,\mathcal{F})$-critical graph, we can show that for the resulting graph $G'$, $\alpha_{p,r}^\mathcal{F}(G)\geq k$ if and only if $\alpha_{p,r}^\mathcal{F}(G')\geq k+1$, and this will prove Theorem~\ref{thm:ker2 origin}.

\bigskip
\noindent\textbf{Organization}. We organize this paper as follows.
In Section~\ref{sec:prelim}, we present some terminology from graph theory, especially lemmas on nowhere dense classes of graphs and class of graphs with bounded expansion.
In Section~\ref{sec:ker1}, we present an almost linear kernel for the \textsc{$(p,r,\mathcal{F})$-Covering} problem on nowhere dense classes of graphs, and a linear kernel for the problem on classes of graphs with bounded expansion.
In Section~\ref{sec:ker2}, we present an almost linear kernel for the \textsc{$(p,r,\mathcal{F})$-Packing} problem on nowhere dense classes of graphs, and a linear kernel for the problem on classes of graphs with bounded expansion.

\section{Preliminaries}\label{sec:prelim}

In this paper, all graphs are simple and finite and have at least one vertex.
For an equivalence relation~$\sim$ on a set $X$, we denote by $\idx(\sim)$ the number of equivalence classes of $\sim$ in $X$.
For every integer $n$, let $[n]$ be the set of positive integers at most $n$.

Throughout this section, let $p$ and $r$ be nonnegative integers, $G$ be a graph, and $A$ and $B$ be subsets of $V(G)$.
The \emph{order} of $G$ is $\abs{V(G)}$.
A \emph{walk} of $G$ is a sequence $(v_0,e_1,v_1, \ldots, e_k,v_k)$ of vertices and edges in $G$ such that $\{v_0,\ldots,v_k\}\subseteq V(G)$ and for each $i\in[k]$, $e_i=\{v_{i-1},v_i\}\in E(G)$.
If a walk has no repeated vertices, then we call it a \emph{path} of $G$, and regard it as a subgraph of $G$.
The \emph{length} of a walk or a path $P$ is the number of edges in $P$.
For vertices $v$ and $w$ of $G$, the \emph{distance} between $v$ and $w$ in~$G$, denoted by $\dist_G(v,w)$, is the length of a shortest path in $G$ connecting $v$ and $w$.
If $G$ has no path connecting $v$ and $w$, then we define $\dist_G(v,w)$ as~$\infty$.
The \emph{distance} between $A$ and $B$ in~$G$, denoted by $\dist_G(A,B)$, is the minimum length of paths in~$G$ connecting some $v\in A$ and $w\in B$.
If $G$ has no path having one end in $A$ and the other end in $B$, then we define $\dist_G(A,B)$ as~$\infty$.
If $A=\{v\}$, then we write $\dist_G(v,B)$ instead of $\dist_G(\{v\},B)$.
The \emph{radius} of $G$ is the minimum $t$ such that $G$ has a vertex $x$ with $\dist_G(x,y)\leq t$ for every $y\in V(G)$.
If $G$ has no such $x$, then we define the radius of $G$ as~$\infty$.
The \emph{$p$-th power} of $G$, denoted by $G^p$, is the graph with vertex set $V(G)$ such that distinct vertices $v$ and $w$ are adjacent in $G^p$ if and only if $\dist_G(v,w)\leq p$.
For a positive integer $q$, the \emph{$q$-subdivision} of $G$ is the graph obtained from $G$ by substituting each edge of $G$ with a path of length exactly $q$.

For graphs $G=(V,E)$ and $G'=(V',E')$, an \emph{isomorphism} between $G$ and $G'$ is a bijection $\phi:V\to V'$ such that $vw\in E$ if and only if $\phi(v)\phi(w)\in E'$.
We denote by $G\setminus A$ the graph obtained from $G$ by removing all vertices in $A$ and all edges incident with some vertices in $A$.
If $A=\{v\}$ for some vertex $v$ of $G$, then we write $G\setminus v$ instead of $G\setminus\{v\}$.
Let $G[A]:=G\setminus(V\setminus A)$, called the \emph{subgraph of $G$ induced by $A$}.

For a vertex $v$ of $G$, let $N_G^r[v]$ be the set of vertices of $G$ which are at distance at most $r$ from $v$ in $G$, and $N_G^r(v):=N_G^r[v]\setminus\{v\}$.
Let $N_G^r[A]:=\bigcup_{v\in A}N_G^r[v]$ and $N_G^r(A):=N_G^r[A]\setminus A$.
If $r=1$, then we omit the superscripts in these notations.

A set $X\subseteq V(G)$ is an \emph{independent set} in $G$ if the vertices in $X$ are pairwise nonadjacent in $G$.
A vertex set $X$ is a \emph{distance-$r$ independent set} in $G$ if~$X$ is an independent set in $G^r$.
Let $\alpha_r(G,A)$ be the maximum size of a subset of $A$ that is distance-$r$ independent in $G$.
A set $D\subseteq V(G)$ is a \emph{dominating set} of $A$ in $G$ if $A\subseteq N_G[D]$.
A vertex set $D$ of $G$ is a \emph{distance-$r$ dominating set} of~$A$ in~$G$ if $D$ is a dominating set of $A$ in $G^r$.
Let $\gamma_r(G,A)$ be the minimum size of a distance-$r$ dominating set of $A$ in $G$.
Note that for every integer $r'\geq r$, $\gamma_{r'}(G,A)$ is at most $\gamma_r(G,A)$.

\subsection{Sparse graphs}

A graph $H$ with vertex set $\{v_1,\ldots,v_n\}$ is an \emph{$r$-shallow minor} of $G$ if there exist pairwise disjoint subsets $V_1,\ldots,V_n$ of $V(G)$ satisfying the following.
\begin{enumerate}
    \item[$\bullet$]    For each $i\in[n]$, $G[V_i]$ has radius at most $r$.
    \item[$\bullet$]    If $v_i$ and $v_j$ are adjacent in $H$, then $\dist_G(V_i,V_j)=1$.
\end{enumerate}

A class $\mathcal{C}$ of graphs has \emph{bounded expansion} if there exists a function $f:\mathbb{N}\to\mathbb{N}$ such that for all $r\in\mathbb{N}$, $G\in\mathcal{C}$, and an $r$-shallow minor $H$ of $G$, $\abs{E(H)}/\abs{V(H)}\leq f(r)$.
A class $\mathcal{C}$ of graphs is \emph{nowhere dense} if there exists a function $g:\mathbb{N}\to\mathbb{N}$ such that for all $r\in\mathbb{N}$ and $G\in\mathcal{C}$, $K_{g(r)}$ is not an $r$-shallow minor of $G$.

Let $L$ be a linear ordering of the vertices of $G$.
We denote by $\leq_L$ the total order on $V(G)$ induced by the ordering $L$.
For vertices $v$ and $w$ of $G$, we say $v<_Lw$ if $v\leq_Lw$ and $v\neq w$.
For a nonnegative integer $k$ and vertices $v$ and $w$ of $G$, $v$ is \emph{weakly $k$-accessible} from $w$ in $L$ if $v\leq_Lw$ and there is a path $P$ of length at most $k$ between $v$ and $w$ such that for every vertex $u$ of~$P$, $v\leq_Lu$.
Let $\wre_k[G,L,v]$ be the set of vertices which are weakly $k$-accessible from $v$ in $L$.
The \emph{weak $k$-coloring number} of $L$ is the maximum size of $\wre_k[G,L,v]$ for all vertices $v$ of $G$.
The \emph{weak $k$-coloring number} of $G$, denoted by $\wc_k(G)$, is the minimum of the weak $k$-coloring numbers among all possible linear orderings of the vertices of $G$.

\begin{LEM}[Zhu~\cite{Zhu2009}]\label{lem:zhu}
    A class $\mathcal{C}$ of graphs is nowhere dense if and only if there exists a function $f_{\mathrm{wcol}}:\mathbb{N}\times\mathbb{R}_+\rightarrow\mathbb{N}$ such that for all $r\in\mathbb{N}$, $\varepsilon\in\mathbb{R}_+$, $G\in\mathcal{C}$, and a subgraph $H$ of $G$, $\wc_r(H)\leq f_{\mathrm{wcol}}(r,\varepsilon)\cdot\abs{V(H)}^\varepsilon$.
\end{LEM}

For vertices $v\in A$ and $u\in V(G)\setminus A$, a path $P$ from $u$ to $v$ is \emph{$A$-avoiding} if $V(P)\cap A=\{v\}$.
For a nonnegative integer $r$ and a vertex $u\in V(G)\setminus A$, the \emph{$r$-projection} of $u$ on~$A$, denoted by $M_r^G(u,A)$, is the set of all vertices $v\in A$ connected to $u$ by an $A$-avoiding path of length at most $r$ in $G$.
The \emph{$r$-projection profile} of $u$ on $A$ is a function $\rho_r^G[u,A]:A\rightarrow[r]\cup\{\infty\}$ such that for each vertex $v\in A$, $\rho_r^G[u,A](v)$ is $\infty$ if there is no $A$-avoiding path of length at most $r$ from $u$ to $v$, and otherwise, the shortest length of an $A$-avoiding path from $u$ to $v$.
We say that a function $f:A\rightarrow[r]\cup\{\infty\}$ is \emph{$r$-realized} if there is a vertex $u\in V(G)\setminus A$ with $f=\rho_r^G[u,A]$.
Let $\mu_r(G,A)$ be the number of distinct $r$-realized functions on $A$, that is, $\mu_r(G,A):=\abs{\{\rho_r^G[u,A]:u\in V(G)\setminus A\}}$.

We will use the following lemmas.

\begin{LEM}[Drange et al.~\cite{KnS2016}]\label{lem:proj'}
    For every class $\mathcal{C}$ of graphs with bounded expansion, there exists a function $f_{\mathrm{proj}}:\mathbb{N}\to\mathbb{N}$ such that for all $r\in\mathbb{N}$, $G\in\mathcal{C}$, and $X\subseteq V(G)$, $\mu_r(G,X)\leq f_{\mathrm{proj}}(r)\cdot\abs{X}$.
\end{LEM}

\begin{LEM}[Eickmeyer et al.~\cite{NCK2017}]\label{lem:proj}
    For every nowhere dense class $\mathcal{C}$ of graphs, there exists a function $f_{\mathrm{proj}}:\mathbb{N}\times\mathbb{R}_+\to\mathbb{N}$ such that for all $r\in\mathbb{N}$, $\varepsilon\in\mathbb{R}_+$, $G\in\mathcal{C}$, and $X\subseteq V(G)$, $\mu_r(G,X)\leq f_{\mathrm{proj}}(r,\varepsilon)\cdot\abs{X}^{1+\varepsilon}$.
\end{LEM}

For a nonnegative integer $r$ and a nonnegative real $t$, an \emph{$(r,t)$-close set} is a set $X\subseteq V(G)$ such that for every $u\in V(G)\setminus X$, $\abs{M_r^G(u,X)}\leq t$.

\begin{LEM}[Drange et al.~\cite{KnS2016}]\label{lem:cl'}
    For every class $\mathcal{C}$ of graphs with bounded expansion, there exist a function $f_{\mathrm{cl}}:\mathbb{N}\to\mathbb{N}$ and a polynomial-time algorithm that for all $r\in\mathbb{N}$, $G\in\mathcal{C}$, and $X\subseteq V(G)$, outputs an $(r,f_{\mathrm{cl}}(r))$-close set $X_{\mathrm{cl}}\supseteq X$ of size at most $f_{\mathrm{cl}}(r)\cdot\abs{X}$.
\end{LEM}

\begin{LEM}[Eickmeyer et al.~\cite{NCK2017}, Drange et al.~\cite{KnS2016}]\label{lem:cl}
    For every nowhere dense class $\mathcal{C}$ of graphs, there exist a function $f_{\mathrm{cl}}:\mathbb{N}\times\mathbb{R}_+\to\mathbb{N}$ and a polynomial-time algorithm that for all $r\in\mathbb{N}$, $\varepsilon\in\mathbb{R}_+$, $G\in\mathcal{C}$, and $X\subseteq V(G)$, outputs an $(r,f_{\mathrm{cl}}(r,\varepsilon)\cdot\abs{X}^\varepsilon)$-close set $X_{\mathrm{cl}}\supseteq X$ of size at most $f_{\mathrm{cl}}(r,\varepsilon)\cdot\abs{X}^{1+\varepsilon}$.
\end{LEM}

For $r\in\mathbb{N}$ and a set $X\subseteq V(G)$, an \emph{$r$-path closure} of $X$ is a set $X_{\mathrm{pth}}\supseteq X$ such that for $u,v\in X$, if $\dist_G(u,v)\leq r$, then $\dist_{G[X_{\mathrm{pth}}]}(u,v)=\dist_G(u,v)$.

\begin{LEM}[Drange et al.~\cite{KnS2016}]\label{lem:pth'}
    For every class $\mathcal{C}$ of graphs with bounded expansion, there exist a function $f_{\mathrm{pth}}:\mathbb{N}\to\mathbb{N}$ and a polynomial-time algorithm that for all $r\in\mathbb{N}$, $G\in\mathcal{C}$, and $X\subseteq V(G)$, outputs an $r$-path closure of $X$ having size at most $f_{\mathrm{pth}}(r)\cdot\abs{X}$.
\end{LEM}

\begin{LEM}[Eickmeyer et al.~\cite{NCK2017}, Drange et al.~\cite{KnS2016}]\label{lem:pth}
    For every nowhere dense class $\mathcal{C}$ of graphs, there exist a function $f_{\mathrm{pth}}:\mathbb{N}\times\mathbb{R}_+\to\mathbb{N}$ and a polynomial-time algorithm that for all $r\in\mathbb{N}$, $\varepsilon\in\mathbb{R}_+$, $G\in\mathcal{C}$, and $X\subseteq V(G)$, outputs an $r$-path closure of $X$ having size at most $f_{\mathrm{pth}}(r,\varepsilon)\cdot\abs{X}^{1+\varepsilon}$.
\end{LEM}

A class $\mathcal{C}$ of graphs is \emph{uniformly quasi-wide} if there exist functions $N:\mathbb{N}\times\mathbb{N}\rightarrow\mathbb{N}$ and $s:\mathbb{N}\rightarrow\mathbb{N}$ such that for all $G\in\mathcal{C}$ and $A\subseteq V(G)$ with $\abs{A}\geq N(r,m)$, there exist sets $S\subseteq V(G)$ and $B\subseteq A\setminus S$ such that $\abs{S}\leq s(r)$, $\abs{B}\geq m$, and $B$ is distance-$r$ independent in $G\setminus S$.

Ne\v{s}et\v{r}il and Ossona de Mendez~\cite{NOdM2011} showed that for a class $\mathcal{C}$ of graphs that is closed under taking subgraphs, $\mathcal{C}$ is uniformly quasi-wide if and only if $\mathcal{C}$ is nowhere dense.
They also showed that every nowhere dense class of graphs is uniformly quasi-wide, even if the class is not closed under taking subgraphs.
Kreutzer, Rabinovich, and Siebertz~\cite{KRS2019} showed that the function $N$ in the definition of uniformly quasi-wideness can be chosen as a polynomial depending on $\mathcal{C}$ and $r$.

\begin{THM}[Kreutzer, Rabinovich, and Siebertz~\cite{KRS2019}]\label{thm:uqw}
    Let $\mathcal{C}$ be a nowhere dense class of graphs.
    For every nonnegative integer $r$, there exist constants $p(r)$ and $s(r)$ such that for all $G\in\mathcal{C}$, $m\in\mathbb{N}$, and $A\subseteq V(G)$ with $\abs{A}\geq m^{p(r)}$, there exist sets $S\subseteq V(G)$ and $B\subseteq A\setminus S$ such that $\abs{S}\leq s(r)$, $\abs{B}\geq m$, and $B$ is distance-$r$ independent in $G\setminus S$.
    Moreover, if $K_c$ is not an $r$-shallow minor of $G$, then $s(r)\leq c\cdot r$, and one can find desired sets $S$ and $B$ in $O(r\cdot c\cdot\abs{A}^{c+6}\cdot\abs{V(G)}^2)$ time.
\end{THM}

\subsection{VC-dimension}

A \emph{set-system} is a family of subsets of a set, called the \emph{ground set}.
Let $\mathcal{S}$ be a set-system with the ground set $S$.
A set $S'\subseteq S$ is \emph{shattered} by $\mathcal{S}$ if $\abs{\{S'\cap T:T\in\mathcal{S}\}}=2^{\abs{S'}}$.
The \emph{Vapnik-Chervonenkis dimension}, or \emph{VC-dimension} for short, of $\mathcal{S}$ is the largest cardinality of a shattered subset of $S$ by $\mathcal{S}$.
Observe that if a set $S'\subseteq S$ is shattered by $\mathcal{S}$, then every subset of $S'$ is also shattered by $\mathcal{S}$.
In addition, for every $\mathcal{S}'\subseteq\mathcal{S}$, the VC-dimension of $\mathcal{S}'$ is at most that of $\mathcal{S}$.

\begin{PROP}[See~{\cite[Proposition~10.3.3]{Matousek2002}}]\label{prop:Matousek2002}
    Let $F(X_1,\ldots,X_d)$ be a set-theoretic expression using set variables $X_1,\ldots,X_d$ and the operations of union, intersection, and difference.
    Let $\mathcal{S}$ be a set-system with the ground set $S$, and $\mathcal{T}:=\{F(S_1,\ldots,S_d):S_1,\ldots,S_d\in\mathcal{S}\}$.
    If $\mathcal{S}$ has VC-dimension $c<\infty$, then $\mathcal{T}$ has VC-dimension $O(cd\log d)$.
\end{PROP}

A \emph{hitting set} of $\mathcal{S}$ is a set $X\subseteq S$ such that for every $T\in\mathcal{S}$, $T\cap X\neq\emptyset$.
Let $\tau(\mathcal{S})$ be the minimum size of a hitting set of $\mathcal{S}$.

Br\"{o}nnimann and Goodrich~\cite{BG1995} and Even, Rawitz, and Shahar~\cite{ERS2005} presented polynomial-time algorithms finding a hitting set $X$ of a nonempty set-system $\mathcal{S}$ having VC-dimension at most $c$ with $\abs{X}=O(c\cdot\tau(\mathcal{S})\cdot\ln\tau(\mathcal{S}))$.

\begin{THM}[\cite{BG1995,ERS2005}]\label{thm:approx hitting}
    There exist a constant $C_\tau$ and a polynomial-time algorithm that for every nonempty set-system $\mathcal{S}$ having VC-dimension at most $c$, outputs a hitting set of $\mathcal{S}$ having size at most $C_\tau\cdot c\cdot\tau(\mathcal{S})\cdot\ln\tau(\mathcal{S})+1$.
\end{THM}

The VC-dimension of $G$ is defined by the VC-dimension of $\{N_G[v]:v\in V(G)\}$.

\begin{THM}[Adler and Adler~\cite{AA2014}]\label{thm:AA2014} 
	Let $\mathcal{C}$ be a nowhere dense class of graphs and $\phi(x,y)$ be a first-order formula such that for all $G\in\mathcal{C}$ and vertices $v$ and $w$ of $G$, $G\models\phi(v,w)$ if and only if $G\models\phi(w,v)$.
	For a graph $G\in\mathcal{C}$, let $G_\phi:=(V(G),\{vw:G\models\phi(v,w)\})$.
    Then there exists a nonnegative integer $c$ depending on $\mathcal{C}$ and $\phi$ such that every graph in $\{G_\phi:G\in\mathcal{C}\}$ has VC-dimension at most $c$.
\end{THM}

For every $p\in\mathbb{N}$, since the property that the distance between two vertices is at most $p$ can be expressed in a first-order formula, Theorem~\ref{thm:AA2014} has the following corollary.

\begin{COR}\label{cor:AA2014}
    For every nowhere dense class $\mathcal{C}$ of graphs, there exists a function $f_{\mathrm{vc}}:\mathbb{N}\to\mathbb{N}$ such that for all $p\in\mathbb{N}$ and $G\in\mathcal{C}$, $G^p$ has VC-dimension at most $f_{\mathrm{vc}}(p)$.\qed
\end{COR}

\section{Kernels for the \textsc{$(p,r,\mathcal{F})$-Covering} problems}\label{sec:ker1}

Let $p$ and $r$ be nonnegative integers with $p\leq2r+1$, and $\mathcal{F}$ be a nonempty finite family of connected graphs.
We present an almost linear kernel for the \textsc{$(p,r,\mathcal{F})$-Covering} problem on every nowhere dense class of graphs, which generalizes the almost linear kernel of~\cite{NCK2017} for the \textsc{Distance-$r$ Dominating Set} problem.

We first present an almost linear kernel for the \textsc{Annotated $(p,r,\mathcal{F})$-Covering} problem.

\begin{THM}\label{thm:ker1}
    For every nowhere dense class $\mathcal{C}$ of graphs, there exists a function $f_{\mathrm{cov}}:\mathbb{N}\times\mathbb{N}\times\mathbb{R}_+\to\mathbb{N}$ satisfying the following.
    For every nonempty family $\mathcal{F}$ of connected graphs with at most~$d$ vertices, $p,r\in\mathbb{N}$ with $p\leq2r+1$, and $\varepsilon>0$, there exists a polynomial-time algorithm that given a graph $G\in\mathcal{C}$, $A\subseteq V(G)$, and $k\in\mathbb{N}$, either
     \begin{enumerate}
         \item[$\bullet$]   correctly decides that $\gamma_{p,r}^\mathcal{F}(G,A)>k$, or
         \item[$\bullet$]   outputs sets $Y\subseteq V(G)$ of size at most $f_{\mathrm{cov}}(r,d,\varepsilon)\cdot k^{1+\varepsilon}$ and $Z\subseteq A\cap Y$ such that $\gamma_{p,r}^\mathcal{F}(G[Y],Z)=\gamma_{p,r}^\mathcal{F}(G,A)$.
     \end{enumerate}
\end{THM}

For a graph $G$ and a set $A\subseteq V(G)$, a \emph{$(p,r,\mathcal{F})$-core} of $A$ in~$G$ is a set $Z\subseteq A$ such that every minimum-size $(p,r,\mathcal{F})$-cover of $Z$ in $G$ is a $(p,r,\mathcal{F})$-cover of $A$ in $G$.
Observe that $\gamma_{p,r}^\mathcal{F}(G,A)=\gamma_{p,r}^\mathcal{F}(G,Z)$ and $A$ is a $(p,r,\mathcal{F})$-core of $A$ in $G$.

The following lemma is a key lemma for proving Theorem~\ref{thm:ker1}.

\begin{LEM}\label{lem:core}
    For every nowhere dense class $\mathcal{C}$ of graphs, there exists a function $f_{\mathrm{core}}:\mathbb{N}\times\mathbb{N}\times\mathbb{R}_+\to\mathbb{N}$ satisfying the following.
    For every nonempty family $\mathcal{F}$ of connected graphs with at most~$d$ vertices, $p,r\in\mathbb{N}$ with $p\leq2r+1$, and $\varepsilon>0$, there exists a polynomial-time algorithm that given a graph $G\in\mathcal{C}$, $A\subseteq V(G)$, and $k\in\mathbb{N}$, either
    \begin{enumerate}
        \item[$\bullet$]    correctly decides that $\gamma_{p,r}^\mathcal{F}(G,A)>k$, or
        \item[$\bullet$]    outputs a $(p,r,\mathcal{F})$-core of $A$ in $G$ having size at most $f_{\mathrm{core}}(r,d,\varepsilon)\cdot k^{1+\varepsilon}$.
    \end{enumerate}
\end{LEM}

We first prove Theorem~\ref{thm:ker1} by using Lemma~\ref{lem:core}.

\begin{proof}[Proof of Theorem~\ref{thm:ker1}]
    We apply the algorithm of Lemma~\ref{lem:core}.
    We may assume that the algorithm finds a $(p,r,\mathcal{F})$-core $Z$ of $A$ in $G$ having size at most $f_{\mathrm{core}}(r,d,\varepsilon)\cdot k^{1+\varepsilon}$.
    We define an equivalence relation $\sim$ on $V(G)\setminus Z$ such that for vertices $u,v\in V(G)\setminus Z$, $u\sim v$ if and only if $\rho_r^G[u,Z]=\rho_r^G[v,Z]$.
    By Lemma~\ref{lem:proj}, $\idx{(\sim)}\leq f_{\mathrm{proj}}(r,\varepsilon)\cdot\abs{Z}^{1+\varepsilon}$.
    
    Let $\mathcal{H}$ be the set of equivalence classes of $\sim$.
    For each $\kappa\in\mathcal{H}$, we choose a vertex $v_\kappa\in\kappa$.
    Let $Z':=Z\cup\{v_\kappa:\kappa\in\mathcal{H}\}$.
    Note that $\abs{Z'}\leq\abs{Z}+\idx{(\sim)}\leq(1+f_{\mathrm{proj}}(r,\varepsilon))\cdot\abs{Z}^{1+\varepsilon}$.
    By Lemma~\ref{lem:pth}, one can find in polynomial time a $(2r+1)$-path closure $Y$ of $Z'$ in $G$ with
    \begin{align*}
        \abs{Y}
        &\leq f_{\mathrm{pth}}(2r+1,\varepsilon)\cdot\abs{Z'}^{1+\varepsilon}\\
        &\leq f_{\mathrm{pth}}(2r+1,\varepsilon)\cdot(1+f_{\mathrm{proj}}(r,\varepsilon))^{1+\varepsilon}\cdot\abs{Z}^{1+3\varepsilon}\\
        &\leq f_{\mathrm{pth}}(2r+1,\varepsilon)\cdot(1+f_{\mathrm{proj}}(r,\varepsilon))^{1+\varepsilon}\cdot f_{\mathrm{core}}(r,d,\varepsilon)^{1+3\varepsilon}\cdot k^{1+7\varepsilon}.
    \end{align*}
    Note that $Z\subseteq A\cap Y$.
    
    \begin{CLM}\label{clm:dom}
        $\gamma_{p,r}^\mathcal{F}(G,A)=\gamma_{p,r}^\mathcal{F}(G[Y],Z)$.
    \end{CLM}
    \begin{subproof}
        Firstly, we show that $\gamma_{p,r}^\mathcal{F}(G,A)\geq\gamma_{p,r}^\mathcal{F}(G[Y],Z)$.
        Let $D$ be a $(p,r,\mathcal{F})$-cover of $A$ in $G$ having size $\ell:=\gamma_{p,r}^\mathcal{F}(G,A)$ and
        \begin{equation*}
            D':=(D\cap Z)\cup\{v_\kappa:\kappa\in\mathcal{H},\ \kappa\cap D\neq\emptyset\}.
        \end{equation*}
        Note that $D'\subseteq Z'$ and $\abs{D'}\leq\abs{D}$.
        Since $D$ is a $(p,r,\mathcal{F})$-cover of $Z$ in $G$ and $v_\kappa$ and each vertex in $\kappa\cap D$ have the same $r$-projection profiles on $Z$, $D'$ is a $(p,r,\mathcal{F})$-cover of $Z$ in $G$.
        
        To show that $\ell\geq\gamma_{p,r}^\mathcal{F}(G[Y],Z)$, it suffices to show that $D'$ is a $(p,r,\mathcal{F})$-cover of $Z$ in $G[Y]$.
        Since $Z\cup D'\subseteq Z'$ and $Y$ is a $(2r+1)$-path closure of $Z'$, for vertices $u\in D'$ and $v\in Z$, if $\dist_G(u,v)\leq r$, then $\dist_{G[Y]}(u,v)=\dist_G(u,v)$.
        It implies that $N^r_{G[Y]}[D']\cap Z=N^r_G[D']\cap Z$, so that $Z\setminus N^r_{G[Y]}[D']=Z\setminus N^r_G[D']$.
        Since $D'$ is a $(p,r,\mathcal{F})$-cover of $Z$ in $G$, $Z\setminus N^r_G[D']$ has no subset $B$ such that $G^p[B]$ is isomorphic to a graph in $\mathcal{F}$.
        Since $p\leq2r+1$, for every set $B\subseteq Z\setminus N^r_{G[Y]}[D']$, $G[Y]^p[B]$ is isomorphic to $G^p[B]$.
        Since $Z\setminus N^r_{G[Y]}[D']=Z\setminus N^r_G[D']$, $D'$ is a $(p,r,\mathcal{F})$-cover of $Z$ in $G[Y]$.
        Therefore, $\gamma_{p,r}^\mathcal{F}(G[Y],Z)\leq\abs{D'}\leq\abs{D}=\ell$.
        
        Conversely, we show that $\gamma_{p,r}^\mathcal{F}(G,A)\leq\gamma_{p,r}^\mathcal{F}(G[Y],Z)$.
        Let $D\subseteq Y$ be a $(p,r,\mathcal{F})$-cover of $Z$ in $G[Y]$ having size $\ell:=\gamma_{p,r}^\mathcal{F}(G[Y],Z)$ and
        \begin{equation*}
            D':=(D\cap Z)\cup\{v_\kappa:\kappa\in\mathcal{H},\ \kappa\cap D\neq\emptyset\}.
        \end{equation*}
        Note that $D'\subseteq Z'$ and $\abs{D'}\leq\abs{D}$.

        Since $\gamma_{p,r}^\mathcal{F}(G,A)=\gamma_{p,r}^\mathcal{F}(G,Z)$, to show that $\gamma_{p,r}^\mathcal{F}(G,A)\leq\ell$, it suffices to show that $D'$ is a $(p,r,\mathcal{F})$-cover of $Z$ in $G$.
        Since $p\leq2r+1$, $Z\subseteq Z'$, and $Y$ is a $(2r+1)$-path closure of $Z'$, for every set $B\subseteq Z\setminus N^r_{G[Y]}[D]$, $G[Y]^p[B]$ is isomorphic to $G^p[B]$.
        Since $D$ is a $(p,r,\mathcal{F})$-cover of $Z$ in $G[Y]$, $Z\setminus N_{G[Y]}^r[D]$ has no subset $B$ such that $G^p[B]$ is isomorphic to a graph in $\mathcal{F}$.
        Thus, to show that $D'$ is a $(p,r,\mathcal{F})$-cover of $Z$ in $G$, it suffices to show that $N^r_{G[Y]}[D]\cap Z\subseteq N^r_G[D']\cap Z$.
        
        Let $y\in N^r_{G[Y]}[D]\cap Z$.
        There is a vertex $x\in D$ with $\dist_{G[Y]}(x,y)\leq r$.
        If $x\in Z'$, then $x\in D'$, so that $y\in N^r_{G[Y]}(x)\cap Z\subseteq N^r_{G}[D']\cap Z$.
        Thus, we may assume that $x\in D\setminus Z'$.
        Since $Z\subseteq Z'$, $x$ is not contained in $Z$, and therefore there is an equivalence class $\tau$ of $\sim$ containing $x$.
        Note that $v_\tau\in D'$.
        
        We verify that $y\in N^r_G(v_\tau)\cap Z$.
        We take an arbitrary path $P$ of $G[Y]$ between $x$ and $y$ having length at most $r$.
        Let $y'$ be the vertex in $V(P)\cap Z$ with $\dist_P(x,y')$ is minimum.
        Let $m$ be the shortest length of a $Z$-avoiding path of $G$ between $x$ and $y'$.
        Since $P$ has length at most $r$ and the subpath $P_1$ of $P$ between $x$ and $y'$ is $Z$-avoiding, $m$ is at most $r$.
        Since $\{x,v_\tau\}\subseteq\tau$ which is an equivalence class of $\sim$, $G$ also has a $Z$-avoiding $P_2$ between $v_\tau$ and $y'$ having length $m$.
        Note that $m$ is at most the length of $P_1$.
        By substituting $P_1$ with $P_2$ from $P$, we obtain a walk of $G$ between $v_\tau$ and $y$ having length at most $r$.
        Thus, $y\in N^r_G(v_\tau)\cap Z\subseteq N^r_G[D']\cap Z$.
        
        By the arbitrary choice of $y$, we have that $N^r_{G[Y]}[D]\cap Z\subseteq N^r_G[D']\cap Z$.
        Hence, $D'$ is a $(p,r,\mathcal{F})$-cover of $Z$ in~$G$.
        Since $Z$ is a $(p,r,\mathcal{F})$-core of $A$ in~$G$, we have that $\gamma_{p,r}^\mathcal{F}(G,A)=\gamma_{p,r}^\mathcal{F}(G,Z)\leq\abs{D'}\leq\abs{D}=\ell$.
    \end{subproof}
    
    By scaling $\varepsilon$ accordingly, one can choose the function $f_{\mathrm{cov}}(r,d,\varepsilon)$ with $\abs{Y}\leq f_{\mathrm{cov}}(r,d,\varepsilon)\cdot k^{1+\varepsilon}$, and this completes the proof.
\end{proof}

The following observation easily follows from the proof of Theorem~\ref{thm:ker1}, since $Y$ is a $(2r+1)$-path closure of $Z'\supseteq Z$ in $G$.
This observation will be used to prove Theorem~\ref{thm:ker2} in Section~\ref{sec:ker2}.

\begin{OBS}\label{obs:obs}
    For every nonnegative integer $r_0\leq2r+1$, $\alpha_{p,r_0}^\mathcal{F}(G[Y],Z)$ and $\alpha_{p,r_0}^\mathcal{F}(G,Z)$ are equal.\qed
\end{OBS}

To prove Lemma~\ref{lem:core}, we will use the following lemma.

\begin{LEM}\label{lem:core recursion}
    For every nowhere dense class $\mathcal{C}$ of graphs, there exists a function $f_{\mathrm{core}}:\mathbb{N}\times\mathbb{N}\times\mathbb{R}_+\to\mathbb{N}$ satisfying the following.
    For every nonempty family $\mathcal{F}$ of connected graphs with at most $d$ vertices, $p,r\in\mathbb{N}$ with $p\leq2r+1$, and $\varepsilon>0$, there exists a polynomial-time algorithm that given a graph $G\in\mathcal{C}$, $A\subseteq V(G)$, $k\in\mathbb{N}$, and a $(p,r,\mathcal{F})$-core $Z$ of $A$ in $G$ with $\abs{Z}>f_{\mathrm{core}}(r,d,\varepsilon)\cdot k^{1+\varepsilon}$, either
    \begin{enumerate}
        \item[$\bullet$]    correctly decides that $\gamma_{p,r}^\mathcal{F}(G,A)>k$, or
        \item[$\bullet$]    outputs a vertex $z\in Z$ such that $Z\setminus\{z\}$ is a $(p,r,\mathcal{F})$-core of $A$ in $G$.
    \end{enumerate}
\end{LEM}

Since $A$ is a $(p,r,\mathcal{F})$-core of $A$ in $G$, we can prove Lemma~\ref{lem:core} by iteratively applying Lemma~\ref{lem:core recursion} to $G$ with $A$ at most $\abs{A}$ times.
To prove Lemma~\ref{lem:core recursion}, we will use the following approximation algorithm.

\begin{PROP}\label{prop:apx}
    For every nowhere dense class $\mathcal{C}$ of graphs, there exists a function $f_{\mathrm{apx}}:\mathbb{N}\times\mathbb{N}\times\mathbb{R}_+\to\mathbb{N}$ satisfying the following.
    For every nonempty family $\mathcal{F}$ of connected graphs with at most~$d$ vertices, $p,r\in\mathbb{N}$, and $\varepsilon>0$, there exists a polynomial-time algorithm that given a graph $G\in\mathcal{C}$ and $A\subseteq V(G)$, outputs a $(p,r,\mathcal{F})$-cover of $A$ in $G$ having size at most $f_{\mathrm{apx}}(r,d,\varepsilon)\cdot\gamma_{p,r}^\mathcal{F}(G,A)^{1+\varepsilon}$.
\end{PROP}
\begin{proof}
    Let $\mathcal{N}:=\{N_G^r[v]:v\in V(G)\}$ and $\mathcal{N}_A:=\{N_G^r[v]:v\in A\}$.
    By Corollary~\ref{cor:AA2014}, $\mathcal{N}$ has VC-dimension at most $f_{\mathrm{vc}}(r)$.
    Since $\mathcal{N}_A\subseteq\mathcal{N}$, $\mathcal{N}_A$ has VC-dimension at most $f_{\mathrm{vc}}(r)$.
    Let $\mathcal{H}_0:=\{N_G^r[B]:B\subseteq A,\ \abs{B}\leq d\}$.
    Let $\mathcal{H}_1$ be the family of sets $B\subseteq A$ such that $G^p[B]$ is isomorphic to a graph in $\mathcal{F}$, and $\mathcal{H}_2:=\{N_G^r[B]:B\in\mathcal{H}_1\}$.
    Since $\mathcal{N}$ has VC-dimension at most $f_{\mathrm{vc}}(r)$, by Proposition~\ref{prop:Matousek2002}, $\mathcal{H}_0$ has VC-dimension at most $O(f_{\mathrm{vc}}(r)\cdot d\log d)$.
    Since $\mathcal{H}_2\subseteq\mathcal{H}_0$, $\mathcal{H}_2$ has VC-dimension at most $O(f_{\mathrm{vc}}(r)\cdot d\log d)$. 
    
    Let $\gamma:=\gamma_{p,r}^\mathcal{F}(G,A)$ and $\delta$ be the VC-dimension of $\mathcal{H}_2$.
    Observe that $(p,r,\mathcal{F})$-covers of $A$ in $G$ correspond to hitting sets of $\mathcal{H}_2$, and vice versa.
	By Theorem~\ref{thm:approx hitting}, one can find in polynomial time a hitting set $X$ of $\mathcal{H}_2$ having size at most $C_\tau\cdot\delta\cdot\gamma\cdot\ln\gamma+1$.
	Thus, one can choose the function $f_{\mathrm{apx}}(r,d,\varepsilon)$ with $\abs{X}\leq f_{\mathrm{apx}}(r,d,\varepsilon)\cdot\gamma^{1+\varepsilon}$, and this completes the proof.
\end{proof}

We now prove Lemma~\ref{lem:core recursion}.

\begin{proof}[Proof of Lemma~\ref{lem:core recursion}]
    The function $f_{\mathrm{core}}(r,d,\varepsilon)$ will be defined later.
    At the beginning, we assume that $\abs{Z}>f_{\mathrm{core}}(r,d,\varepsilon)\cdot k^{1+C\varepsilon}$ for some constant $C$, and at the end, we scale $\varepsilon$ accordingly.

    If $Z$ contains a vertex $v$ such that for every set $B\subseteq Z\setminus\{v\}$ with $\abs{B}\leq d-1$, $G^p[B\cup\{v\}]$ is isomorphic to no graph in $\mathcal{F}$, then the statement holds by taking $v$ as $z$.
    Thus, we may assume that for every $v\in Z$, $Z\setminus\{v\}$ has a subset $B$ such that $G^p[B\cup\{v\}]$ is isomorphic to a graph in $\mathcal{F}$.
    
    By Proposition~\ref{prop:apx}, one can find in polynomial time a $(p,r,\mathcal{F})$-cover $X$ of $Z$ in $G$ having size at most $f_{\mathrm{apx}}(r,d,\varepsilon)\cdot\gamma_{p,r}^\mathcal{F}(G,Z)^{1+\varepsilon}$.
    If $\abs{X}>f_{\mathrm{apx}}(r,d,\varepsilon)\cdot k^{1+\varepsilon}$, then $\gamma_{p,r}^\mathcal{F}(G,A)=\gamma_{p,r}^\mathcal{F}(G,Z)>k$.
    Thus, we may assume that $\abs{X}\leq f_{\mathrm{apx}}(r,d,\varepsilon)\cdot k^{1+\varepsilon}$.
    Let
    \begin{equation*}
        r':=2pd+3r.
    \end{equation*}
    By Lemma~\ref{lem:cl}, one can find in polynomial time an $(r',f_{\mathrm{cl}}(r',\varepsilon)\cdot\abs{X}^\varepsilon)$-close set $X_{\mathrm{cl}}\supseteq X$ of size at most $f_{\mathrm{cl}}(r',\varepsilon)\cdot\abs{X}^{1+\varepsilon}\leq f_{\mathrm{cl}}(r',\varepsilon)\cdot f_{\mathrm{apx}}(r,d,\varepsilon)^{1+\varepsilon}\cdot k^{1+3\varepsilon}$.
    
    We define an equivalence relation $\sim$ on $Z\setminus X_{\mathrm{cl}}$ such that for vertices $u,v\in Z\setminus X_{\mathrm{cl}}$, $u\sim v$ if and only if $\rho_{r'}^G[u,X_{\mathrm{cl}}]=\rho_{r'}^G[v,X_{\mathrm{cl}}]$.
    By Lemma~\ref{lem:proj}, 
    \begin{align*}
        \idx(\sim)
        &\leq f_{\mathrm{proj}}(r',\varepsilon)\cdot\abs{X_{\mathrm{cl}}}^{1+\varepsilon}\\
        &\leq f_{\mathrm{proj}}(r',\varepsilon)\cdot f_{\mathrm{cl}}(r',\varepsilon)^{1+\varepsilon}\cdot f_{\mathrm{apx}}(r,d,\varepsilon)^{1+3\varepsilon}\cdot k^{1+7\varepsilon}.
    \end{align*}
    Let $p(r')$ and $s:=s(r')$ be the constants in Theorem~\ref{thm:uqw}.
    Let
    \begin{align*}
        \xi&:=2\cdot f_{\mathrm{cl}}(r',\varepsilon)\cdot f_{\mathrm{apx}}(r,d,\varepsilon)^\varepsilon\cdot k^{2\varepsilon}+d^2/4+s+1,\\
        m&:=2^{2^{d^2/2+sd}\cdot(r+1)^{sd}}\cdot\xi+1.
    \end{align*}
    By setting $C=7+2\cdot p(r')$, one can choose the function $f_{\mathrm{core}}(r,d,\varepsilon)$ with
    \begin{equation*}
        f_{\mathrm{core}}(r,d,\varepsilon)\cdot k^{1+C\varepsilon}\geq\abs{X_{\mathrm{cl}}}+\idx(\sim)\cdot m^{p(r')}.
    \end{equation*}
    Since $\abs{Z}>f_{\mathrm{core}}(r,d,\varepsilon)\cdot k^{1+C\varepsilon}$, we have that $\abs{Z\setminus X_{\mathrm{cl}}}>\idx(\sim)\cdot m^{p(r')}$.
    Thus, by the pigeonhole principle, there is an equivalence class $\lambda$ of $\sim$ with $\abs{\lambda}>m^{p(r')}$.
    By Theorem~\ref{thm:uqw}, one can find in polynomial time sets $S\subseteq V(G)$ and $L\subseteq\lambda\setminus S$ such that $\abs{S}\leq s$, $\abs{L}\geq m$, and $L$ is distance-$r'$ independent in $G\setminus S$.
    
    We are going to find a desired vertex $z$ from $L$.
    To do this, we define the following.
    For each $i\in[d]$, let $\mathcal{G}_i$ be the set of all graphs whose vertex sets are $[i]$.
    Note that $\abs{\mathcal{G}_i}=2^{i(i-1)/2}$ for each $i\in[d]$.
    Let $\mathcal{H}$ be the set of functions $\rho:S\to[2r+1]\cup\{\infty\}$.
    Since $\abs{S}\leq s$, we have that $\abs{\mathcal{H}}\leq(2r+2)^s$.
    For each $i\in[d]$, let $\mathcal{H}_i$ be the set of all vectors $(h_1,\ldots,h_i,g)$ of length $i+1$ where $h_j\in\mathcal{H}$ for each $j\in[i]$ and $g\in\mathcal{G}_i$.
    Let $\overline{\mathcal{H}}:=\bigcup_{i=1}^d\mathcal{H}_i$.
    Note that
    \begin{align*}
        \abs{\overline{\mathcal{H}}}&=\sum_{i=1}^d\abs{\mathcal{H}_i}=\sum_{i=1}^d(\abs{\mathcal{H}}^i\cdot\abs{\mathcal{G}_i})\leq\sum_{i=1}^d((2r+2)^{si}\cdot2^{i(i-1)/2})\\
        &\leq2^{d(d-1)/2}\cdot\sum_{i=1}^d(2r+2)^{si}\leq2^{d(d-1)/2}\cdot2(2r+2)^{sd}\leq2^{d^2/2+sd}\cdot(r+1)^{sd}.
    \end{align*}
    
    Let $\ell:=\abs{\overline{\mathcal{H}}}$.
    We take an arbitrary ordering $\sigma_1,\ldots,\sigma_\ell$ of $\overline{\mathcal{H}}$.
    For each $v\in L$, let $\mathcal{A}_v:=\emptyset$ and $\mathbf{x}(v)$ be a zero vector of length $\ell$.
    One can enumerate in polynomial time the sets $B\subseteq Z\setminus\{v\}$ of size at most $d-1$ such that $G^p[B\cup\{v\}]$ is isomorphic to a graph in $\mathcal{F}$.
    For each such $B$, we do the following.
    If there is an index $i\in[\ell]$ such that the $i$-th entry of $\mathbf{x}(v)$ is $0$ and for $\sigma_i=(h_1^i,\ldots,h_t^i,g_i)\in\overline{\mathcal{H}}$, there exists
    \begin{enumerate}
        \item[$\bullet$]    an isomorphism $\phi_i:(B\setminus S)\cup\{v\}\to[t]$ between $(G\setminus S)^p[(B\setminus S)\cup\{v\}]$ and $g_i$ where $\phi_i(v)=1$ and for each $j\in[t]$, $\rho_{2r+1}^G[\phi_i^{-1}(j),S]=h^i_j$,
    \end{enumerate}
    then we put $B$ into $\mathcal{A}_v$ and convert the $i$-th entry of $\mathbf{x}(v)$ to $1$.
    Otherwise, we do nothing for the chosen $B$.
    Since $\abs{B}\leq d-1$, one can check in polynomial time whether~$B$ satisfies the conditions.
    Thus, the resulting $\mathcal{A}_v$ and $\mathbf{x}(v)$ can be computed in polynomial time.
    
    For each $v\in L$, since $Z\setminus\{v\}$ has a subset $B$ such that $G^p[B\cup\{v\}]$ is isomorphic to a graph in $\mathcal{F}$, $\mathcal{A}_v\neq\emptyset$ and $\mathbf{x}(v)$ has a nonzero entry.
    For each set $B\in\mathcal{A}_v$, let $B^*$ be the vertex set of the component of $(G\setminus S)^p[(B\setminus S)\cup\{v\}]$ having $v$, and $\mathcal{B}_v:=\bigcup_{B\in\mathcal{A}_v}B^*$.
    
    Since $\abs{L}\geq m=2^{2^{d^2/2+sd}\cdot(r+1)^{sd}}\cdot\xi+1$ and $\ell\leq2^{d^2/2+sd}\cdot(r+1)^{sd}$, by the pigeonhole principle, $L$ has a subset $\kappa_1$ such that $\abs{\kappa_1}\geq\xi+1$ and $\mathbf{x}(v)=\mathbf{x}(w)$ for all $v,w\in\kappa_1$.
    Let $z$ be a vertex in $\kappa_1$ such that $\dist_{G\setminus S}(\mathcal{B}_z,X_{\mathrm{cl}})\geq\dist_{G\setminus S}(\mathcal{B}_v,X_{\mathrm{cl}})$ for every $v\in\kappa_1$.
    
    We show that $Z\setminus\{z\}$ is a $(p,r,\mathcal{F})$-core of $A$ in $G$.
    Let $D$ be a minimum-size $(p,r,\mathcal{F})$-cover of $Z\setminus\{z\}$ in $G$.
    To show that $Z\setminus\{z\}$ is a $(p,r,\mathcal{F})$-core of $A$ in $G$, we need to show that $D$ is a $(p,r,\mathcal{F})$-cover of $A$ in $G$.
    Since $Z$ is a $(p,r,\mathcal{F})$-core of $A$ in $G$, it suffices to show that $D$ is a $(p,r,\mathcal{F})$-cover of $Z$ in~$G$.
    
    Suppose for contradiction that $D$ is not a $(p,r,\mathcal{F})$-cover of $Z$ in $G$.
    Since $D$ is a $(p,r,\mathcal{F})$-cover of $Z\setminus\{z\}$ in $G$, there is a set $B_z\subseteq Z\setminus(N^r_G[D]\cup\{z\})$ such that $G^p[B_z\cup\{z\}]$ is isomorphic to a graph in~$\mathcal{F}$.
    In particular, there exist a graph $H\in\mathcal{G}_t$ for some $t\leq d$ and an isomorphism $\psi_z:(B_z\setminus S)\cup\{z\}\to[t]$ between $(G\setminus S)^p[(B_z\setminus S)\cup\{z\}]$ and $H$ where $\psi_z(z)=1$.
    For each vertex $v\in\kappa_1\setminus\{z\}$, there exist $B_v\in\mathcal{A}_v$ and
    \begin{enumerate}
        \item[$\bullet$]    an isomorphism $\psi_v:(B_v\setminus S)\cup\{v\}\to[t]$ between $(G\setminus S)^p[(B_v\setminus S)\cup\{v\}]$ and $H$ where $\psi_v(v)=1$ and for each $j\in[t]$, $\rho_{2r+1}^G[\psi_v^{-1}(j),S]=\rho_{2r+1}^G[\psi_z^{-1}(j),S]$.
    \end{enumerate}
    
    To derive a contradiction, we do the following steps.
    \begin{enumerate}
        \item[(1)]  Find a set $\kappa_3\subseteq\kappa_1\setminus\{z\}$ such that for each $u\in\kappa_3$, $\dist_{G\setminus S}(\mathcal{B}_u,X_{\mathrm{cl}})>r$ and $G^p[B_u^*\cup(B_z\setminus B_z^*)]$ is isomorphic to $G^p[B_z\cup\{z\}]$.
        \item[(2)]  Show that $\abs{D}\geq\abs{\kappa_3}$.
        \item[(3)]  Construct a $(p,r,\mathcal{F})$-cover of $Z\setminus\{z\}$ in $G$ having size less than $\abs{D}$.
    \end{enumerate}
    Since $D$ is a minimum-size $(p,r,\mathcal{F})$-cover of $Z\setminus\{z\}$ in $G$, these steps derive a contradiction.
    
    For the first step, we will use the following three claims.
    Let $\kappa'_1$ be the set of vertices $v\in\kappa_1$ with $\dist_{G\setminus S}(\mathcal{B}_v,X_{\mathrm{cl}})\leq r$.

    \begin{CLM}\label{clm:new1}
        $\abs{\kappa'_1}\leq\abs{M_{r'}^G(z,X_{\mathrm{cl}})}$.
    \end{CLM}
    \begin{subproof}
        We first show that for each $v\in\kappa'_1$, $G\setminus S$ has a path of length at most $r'/2$ between $v$ and $X_{\mathrm{cl}}$.
        Since $v\in\kappa'_1$, for some $B\in\mathcal{A}_v$, $G\setminus S$ has a path $R$ of length at most $r$ between some $b\in B^*$ and $x\in X_{\mathrm{cl}}$.
        Since $(G\setminus S)^p[B^*]$ is connected and $\abs{B^*}\leq d$, $G\setminus S$ has a path $R'$ between $v$ and $b$ having length at most $p(d-1)$.
        By concatenating $R'$ and $R$, we obtain a walk of $G\setminus S$ between $v$ and $x$ having length at most $p(d-1)+r\leq r'/2$.
        Therefore, $G\setminus S$ has a path of length at most $r'/2$ between $v$ and $X_{\mathrm{cl}}$.
        
        We choose a shortest path $P_v$ among such paths. Let $y_v$ be the vertex in $V(P_v)\cap X_{\textrm{cl}}$.
        Note that $P_v$ is an $X_{\mathrm{cl}}$-avoiding path of $G\setminus S$ having length at most $r'/2$ and $y_v$ is contained in $M_{r'}^G(v,X_{\mathrm{cl}})$.
        
        We now show that $\abs{\kappa'_1}\leq\abs{M_{r'}^G(z,X_{\mathrm{cl}})}$.
        Suppose not.
        Since $\{v,z\}\subseteq\kappa_1\subseteq\lambda$ where $\lambda$ is an equivalence class of $\sim$, $M_{r'}^G(v,X_{\mathrm{cl}})$ and $M_{r'}^G(z,X_{\mathrm{cl}})$ are same, and therefore $y_v$ is contained in $M_{r'}^G(z,X_{\mathrm{cl}})$.
        By the pigeonhole principle, there are distinct $v,v'\in\kappa'_1$ with $y_v=y_{v'}$.
        By concatenating $P_v$ and $P_{v'}$, we obtain a walk of $G\setminus S$ having length at most $r'$ between $v$ and $v'$, contradicting the assumption that $L$ is distance-$r'$ independent in $G\setminus S$.
    \end{subproof}
    
    Let $\kappa_2:=\kappa_1\setminus\kappa'_1$.
    Note that by Claim~\ref{clm:new1},
    \begin{align*}
        \abs{\kappa_2}
        &\geq\xi+1-\abs{M_{r'}^G(z,X_{\mathrm{cl}})}\\
        &\geq\xi+1-f_{\mathrm{cl}}(r',\varepsilon)\cdot\abs{X}^\varepsilon\\
        &\geq\xi+1-f_{\mathrm{cl}}(r',\varepsilon)\cdot f_{\mathrm{apx}}(r,d,\varepsilon)^\varepsilon\cdot k^{2\varepsilon}\\
        &=f_{\mathrm{cl}}(r',\varepsilon)\cdot f_{\mathrm{apx}}(r,d,\varepsilon)^\varepsilon\cdot k^{2\varepsilon}+d^2/4+s+2,
    \end{align*}
    and $z\in\kappa_2$.
    Let $B_z^*$ be the vertex set of the component of $(G\setminus S)^p[(B_z\setminus S)\cup\{z\}]$ having $z$.
    Note that for vertices $v,w\in\kappa_2$, $\psi_v^{-1}\circ\psi_w$ is an isomorphism between $(G\setminus S)^p[(B_w\setminus S)\cup\{w\}]$ and $(G\setminus S)^p[(B_v\setminus S)\cup\{v\}]$ assigning $w$ to $v$.
    Thus, $\psi_v^{-1}\circ\psi_z(B_z^*)=B_v^*$.
    
    \begin{CLM}\label{clm:new2-1}
        For vertices $v,w\in\kappa_2$, $\psi_w^{-1}\circ\psi_v$ is an isomorphism between $G^p[(B_v\setminus S)\cup\{v\}]$ and $G^p[(B_w\setminus S)\cup\{w\}]$.
    \end{CLM}
    \begin{subproof}
        It suffices to show that for $i,j\in[t]$, $\psi_v^{-1}(i)$ is adjacent to $\psi_v^{-1}(j)$ in $G^p$ if and only if $\psi_w^{-1}(i)$ is adjacent to $\psi_w^{-1}(j)$ in $G^p$.
        Suppose that $\psi_v^{-1}(i)$ is adjacent to $\psi_v^{-1}(j)$ in $G^p$.
        Since $\psi_w^{-1}\circ\psi_v$ is an isomorphism between $(G\setminus S)^p[(B_v\setminus S)\cup\{v\}]$ and $(G\setminus S)^p[(B_w\setminus S)\cup\{w\}]$, we may assume that $\psi_v^{-1}(i)$ and $\psi_v^{-1}(j)$ are nonadjacent in $(G\setminus S)^p[(B_v\setminus S)\cup\{v\}]$.
        Thus, every path of length at most $p$ in $G$ between $\psi_v^{-1}(i)$ and $\psi_v^{-1}(j)$ has a vertex in $S$.
        
        We take an arbitrary path $Q$ of $G$ between $\psi_v^{-1}(i)$ and $\psi_v^{-1}(j)$ having length at most $p$.
        Let $q_i$ and $q_j$ be the vertices in $V(Q)\cap S$ such that each of $\dist_Q(\psi_v^{-1}(i),q_i)$ and $\dist_Q(\psi_v^{-1}(j),q_j)$ is minimum.
        Such $q_i$ and $q_j$ exist, because $Q$ has a vertex in $S$.
        Let $Q_i$ be the subpath of $Q$ between $\psi_v^{-1}(i)$ and $q_i$, and $Q_j$ be the subpath of $Q$ between $\psi_v^{-1}(j)$ and $q_j$.
        Note that both $Q_i$ and $Q_j$ are $S$-avoiding paths of length at most $p\leq2r+1$.
        
        Since $\{v,w\}\subseteq\kappa_2\subseteq\kappa_1$, $\rho_{2r+1}^G[\psi_v^{-1}(i),S]$ and $\rho_{2r+1}^G[\psi_w^{-1}(i),S]$ are same, and therefore $G$ has an $S$-avoiding path $Q'_i$ between $\psi_w^{-1}(i)$ and $q_i$ whose length is at most that of $Q_i$.
        Similarly, $G$ has an $S$-avoiding path $Q'_j$ between $\psi_w^{-1}(j)$ and $q_j$ whose length is at most that of $Q_j$.
        By substituting $Q_i$ and $Q_j$ with $Q'_i$ and $Q'_j$ from $Q$, respectively, we obtain a walk of $G$ between $\psi_w^{-1}(i)$ and $\psi_w^{-1}(j)$ whose length is at most $p$.
        Therefore, $\psi_w^{-1}(i)$ is adjacent to $\psi_w^{-1}(j)$ in $G^p$, and this proves the claim.
    \end{subproof}
    
    \begin{CLM}\label{clm:new2}
        $\kappa_2$ contains at most $d^2/4$ vertices $v$ such that $G^p[B_v^*\cup(B_z\setminus B_z^*)]$ is not isomorphic to $G^p[B_z\cup\{z\}]$.
    \end{CLM}
    \begin{subproof}
        For vertices $u\in\kappa_2\setminus\{z\}$ and $i\in\psi_z(B_z^*)$, since $\{u,z\}\subseteq\kappa_2\subseteq\kappa_1$, $\rho_{2r+1}^G[\psi_u^{-1}(i),S]$ and $\rho_{2r+1}^G[\psi_z^{-1}(i),S]$ are same.
        Therefore, for each $w\in S$, $\psi_u^{-1}(i)$ is adjacent to $w$ in $G^p$ if and only if $\psi_z^{-1}(i)$ is adjacent to $w$ in $G^p$.
        By Claim~\ref{clm:new2-1}, the restriction of $\psi_u^{-1}\circ\psi_z$ on $B_z^*$ is an isomorphism between $G^p[B_z^*]$ and $G^p[B_u^*]$.
        
        We first show that for all vertices $v\in\kappa_2$, $i\in\psi_z(B_z^*)$, and $w\in B_z\setminus(B_z^*\cup S)$, if $\psi_z^{-1}(i)$ is adjacent to $w$ in $G^p$, then $\psi_v^{-1}(i)$ is adjacent to $w$ in $G^p$.
        Suppose that $\psi_z^{-1}(i)$ is adjacent to $w$ in $G^p$.
        We take an arbitrary path $Q'$ of $G$ between $\psi_z^{-1}(i)$ and $w$ having length at most $p$.
        Since $(G\setminus S)^p[B_z^*]$ is a component of $(G\setminus S)^p[(B_z\setminus S)\cup\{z\}]$ having $z$ and $w\notin B_z^*$, $Q'$ must have a vertex in $S$.
        
        Let $q$ be the the vertex in $V(Q')\cap S$ such that $\dist_{Q'}(\psi_z^{-1}(i),q)$ is minimum.
        Such $q$ exists, because $Q'$ has a vertex in $S$.
        Let $Q'_1$ be the subpath of $Q'$ between $\psi_z^{-1}(i)$ and $q$.
        Note that $Q'_1$ is an $S$-avoiding path of length at most $p\leq2r+1$.
        Since $\rho_{2r+1}^G[\psi_v^{-1}(i),S]=\rho_{2r+1}^G[\psi_z^{-1}(i),S]$, $G$ has an $S$-avoiding path $Q'_2$ between $\psi_v^{-1}(i)$ and $q$ having length at most that of $Q'_1$.
        By substituting $Q'_1$ with $Q'_2$ from $Q'$, we obtain a walk of $G$ between $\psi_v^{-1}(i)$ and $w$ having length at most $p$.
        Hence, $\psi_v^{-1}(i)$ is adjacent to~$w$ in $G^p$.
        
        Thus, there is no pair of vertices $i\in\psi_z(B_z^*)$ and $w\in B_z\setminus(B_z^*\cup S)$ such that $\psi_z^{-1}(i)$ is adjacent to $w$ in $G^p$ and $\psi_u^{-1}(i)$ is nonadjacent to $w$ in $G^p$.
        
        We now show that if there exist vertices $i\in\psi_z(B_z^*)$ and $w\in B_z\setminus(B_z^*\cup S)$ such that $\psi_z^{-1}(i)$ is nonadjacent to $w$ in $G^p$, then $\kappa_2$ contains at most one vertex $x$ such that $\psi_x^{-1}(i)$ is adjacent to $w$ in $G^p$.
        To prove the claim, it suffices to show this statement, because $\abs{B_z^*}\cdot \abs{B_z\setminus (B_z^*\cup S)}\le d^2/4$.
        
        Suppose for contradiction that there exist vertices $i\in\psi_z(B_z^*)$, $w\in B_z\setminus(B_z^*\cup S)$, and distinct $x,x'\in\kappa_2$ such that $\psi_z^{-1}(i)$ is nonadjacent to $w$ in $G^p$ and both $\psi_x^{-1}(i)$ and $\psi_{x'}^{-1}(i)$ are adjacent to $w$ in $G^p$.
        Then $G$ has paths $R$ and $R'$ of length at most $p$ from $w$ to $\psi_x^{-1}(i)$ and $\psi_{x'}^{-1}(i)$, respectively.
        
        We first verify that $R$ or $R'$ has a vertex in $S$.
        Suppose not.
        Since $\abs{B_x^*}\leq d$, $G\setminus S$ has a path $R_1$ of length at most $p(d-1)$ between $x$ and $\psi_x^{-1}(i)$.
        Similarly, $G\setminus S$ has a path $R'_1$ of length at most $p(d-1)$ between $x'$ and $\psi_{x'}^{-1}(i)$.
        Since neither $R$ nor $R'$ has a vertex in $S$, by concatenating $R_1$, $R$, $R'$, and $R'_1$, we obtain a walk of $G\setminus S$ of length at most $2pd\leq r'$ between $x$ and $x'$, contradicting the assumption that $L$ is distance-$r'$ independent in $G\setminus S$.
        Hence, $R$ or $R'$ has a vertex in $S$.
        By symmetry, we may assume that $R$ has a vertex in $S$.
        
        Let $t$ be the vertex in $V(R)\cap S$ such that $\dist_R(\psi_x^{-1}(i),t)$ is minimum.
        Let $R_0$ be the subpath of $R$ between $\psi_x^{-1}(i)$ and $t$.
        Note that $R_0$ is an $S$-avoiding path of length at most $p\leq2r+1$.
        Since $\rho_{2r+1}^G[\psi_x^{-1}(i),S]=\rho_{2r+1}^G[\psi_z^{-1}(i),S]$, $G$ has an $S$-avoiding path $R'_0$ between $\psi_z^{-1}(i)$ and $t$ having length at most that of $R_0$.
        By substituting $R_0$ with $R'_0$ from $R$, we obtain a walk of $G$ between $\psi_z^{-1}(i)$ and $w$ having length at most $p$, contradicting the assumption that $\psi_z^{-1}(i)$ is nonadjacent to $w$ in $G^p$, and this proves the claim.
    \end{subproof}
    
    Since $\abs{\kappa_2}\geq f_{\mathrm{cl}}(r',\varepsilon)\cdot f_{\mathrm{apx}}(r,d,\varepsilon)^\varepsilon\cdot k^{2\varepsilon}+d^2/4+s+2$, by Claim~\ref{clm:new2}, $\kappa_2\setminus\{z\}$ has a subset $\kappa_3$ of size at least $f_{\mathrm{cl}}(r',\varepsilon)\cdot f_{\mathrm{apx}}(r,d,\varepsilon)^\varepsilon\cdot k^{2\varepsilon}+s+1$ such that for each vertex $u\in\kappa_3$, $G^p[B_u^*\cup(B_z\setminus B_z^*)]$ is isomorphic to $G^p[B_z\cup\{z\}]$, which is isomorphic to a graph in $\mathcal{F}$.
    
    We now show that $\abs{D}\geq\abs{\kappa_3}$.
    Since $B_u^*\cup(B_z\setminus B_z^*)\subseteq Z\setminus\{z\}$ and $D$ is a $(p,r,\mathcal{F})$-cover of $Z\setminus\{z\}$ in $G$, for each vertex $u\in\kappa_3$, there exist vertices $x_u\in B_u^*\cup(B_z\setminus B_z^*)$ and $d_u\in D$ with $\dist_G(x_u,d_u)\leq r$.
    Observe that $x_u\in\psi_u^{-1}\circ\psi_z(B_z^*)$, because $B_z\setminus B_z^*\subseteq B_z\subseteq Z\setminus(N_G^r[D]\cup\{z\})$.
    We take an arbitrary path $P_u$ of $G$ between $x_u$ and $d_u$ having length at most~$r$.
    
    To show that $\abs{D}\geq\abs{\kappa_3}$, it suffices to show that for distinct $u,u'\in\kappa_3$, $d_u\neq d_{u'}$, because $\{d_u,d_{u'}\}\subseteq D$.
    We show this by the following two claims.
    
    \begin{CLM}\label{clm:new3}
        For each $u\in\kappa_3$, $V(P_u)\cap(S\cup X_{\mathrm{cl}})=\emptyset$.
    \end{CLM}
    \begin{subproof}
        Let $u$ be a vertex in $\kappa_3$.
        Suppose for contradiction that $V(P_u)\cap S\neq\emptyset$.
        Let $q$ be the vertex in $V(P_u)\cap S$ such that $\dist_{P_u}(x_u,q)$ is minimum.
        Let $P_1$ be the subpath of $P_u$ between $x_u$ and $q$.
        Note that $P_1$ is an $S$-avoiding path of length at most $r$.
        Since $\{u,z\}\subseteq\kappa_2\subseteq\kappa_1$, $G$ has an $S$-avoiding path $P_2$ between $\psi_z^{-1}\circ\psi_u(x_u)$ and $q$ having length at most that of $P_1$.
        By substituting $P_1$ with $P_2$ from $P_u$, we obtain a walk of $G$ between $\psi_z^{-1}\circ\psi_u(x_u)\in B_z^*\subseteq B_z\cup\{z\}$ and $d_u$ having length at most $r$, contradicting the assumption that $B_z\cap N^r_G[D]=\emptyset$.
        Hence, $V(P_u)\cap S=\emptyset$.
        
        Since $u\notin\kappa'_1$ and $B_u\in\mathcal{A}_u$, we have that
        \begin{equation*}
            \dist_{G\setminus S}(x_u,X_{\mathrm{cl}})\geq\dist_{G\setminus S}(B_u^*,X_{\mathrm{cl}})\geq\dist_{G\setminus S}(\mathcal{B}_u,X_{\mathrm{cl}})>r.
        \end{equation*}
        Since $P_u$ is a path of $G\setminus S$ having length at most $r$, $V(P_u)\cap X_{\mathrm{cl}}=\emptyset$.
    \end{subproof}
    
    \begin{CLM}\label{clm:new4}
        For distinct $u,u'\in\kappa_3$, $d_u\neq d_{u'}$.
    \end{CLM}
    \begin{subproof}
        Suppose not.
        Since $(G\setminus S)^p[B_u^*]$ is connected and $\abs{B_u^*}\leq d$, $G\setminus S$ has a path $Q$ between $u$ and $x_u$ having length at most $p(d-1)$.
        Similarly, $G\setminus S$ has a path $Q'$ between $u'$ and $x_{u'}$ having length at most $p(d-1)$.
        By Claim~\ref{clm:new3}, neither $P_u$ nor $P_{u'}$ has a vertex in $S$.
        By concatenating $Q$, $P_u$, $P_{u'}$, and $Q'$, we obtain a walk of $G\setminus S$ between $u$ and $u'$ having length at most $2p(d-1)+2r\leq r'$, contradicting the assumption that $L$ is distance-$r'$ independent in $G\setminus S$.
    \end{subproof}
    
    We now construct a $(p,r,\mathcal{F})$-cover of $Z\setminus\{z\}$ in $G$ having size less than~$\abs{D}$.
    Let $D_{\mathrm{sell}}:=\{d_u:u\in\kappa_3\}$, $D_{\mathrm{buy}}:=M_{r'}^G(z,X_{\mathrm{cl}})\cup S$, and $D':=(D\setminus D_{\mathrm{sell}})\cup D_{\mathrm{buy}}$.
    By Claim~\ref{clm:new4},
    \begin{align*}
        \abs{D_{\mathrm{sell}}}&=\abs{\kappa_3}\geq f_{\mathrm{cl}}(r',\varepsilon)\cdot f_{\mathrm{apx}}(r,d,\varepsilon)^\varepsilon\cdot k^{2\varepsilon}+s+1,\\
        \abs{D_{\mathrm{buy}}}&\leq\abs{M_{r'}^G(z,X_{\mathrm{cl}})}+\abs{S}\leq f_{\mathrm{cl}}(r',\varepsilon)\cdot f_{\mathrm{apx}}(r,d,\varepsilon)^\varepsilon\cdot k^{2\varepsilon}+s.
    \end{align*}
    Since $D_{\mathrm{sell}}\subseteq D$, we have that $\abs{D'}<\abs{D}$.
    
    To derive a contradiction, we show the following claim.
    
    \begin{CLM}\label{clm:new5}
        $D'$ is a $(p,r,\mathcal{F})$-cover of $Z\setminus\{z\}$ in $G$.
    \end{CLM}
    \begin{subproof}
        Suppose not.
        Then there is a set $B'\subseteq Z\setminus(N^r_G[D']\cup\{z\})$ such that $G^p[B']$ is isomorphic to a graph in $\mathcal{F}$.
        Since $D$ is a $(p,r,\mathcal{F})$-cover of $Z\setminus\{z\}$ in $G$ and $D\setminus D'\subseteq D_{\mathrm{sell}}$, $D_{\mathrm{sell}}$ contains a vertex $d_u$ for some $u\in\kappa_3$ with $\dist_G(d_u,B')\leq r$.
        
        Since $(G\setminus S)^p[B_u^*]$ is connected and $\abs{B_u^*}\leq d$, $G\setminus S$ has a path $Q_0$ of length at most $p(d-1)$ between $u$ and $x_u$.
        More specifically, $Q_0$ is a concatenation of paths $Q_0^1,\ldots,Q_0^{t_1}$ for $t_1\leq d-1$ such that for each $i\in[t_1]$, the length of $Q_0^i$ is at most $p$ and the ends of $Q_0^i$ are in $B_u^*$.
        Since $\dist_G(d_u,B')\leq r$, $G$ has a path $Q_1$ of length at most $r$ between $d_u$ and $w_1\in B'$.
        Since $X_{\mathrm{cl}}$ is a $(p,r,\mathcal{F})$-cover of $Z$ in $G$, $G$ has a path $Q_2$ of length at most $r$ between $w_2\in B'$ and $x\in X_{\mathrm{cl}}$.
        Since $G^p[B']$ is isomorphic to a connected graph in $\mathcal{F}$ and $\abs{B'}\leq d$, $G$ has a path $R$ of length at most $p(d-1)$ between $w_1$ and $w_2$.
        More specifically, $R$ is a concatenation of paths $R_1,\ldots,R_{t_2}$ for $t_2\leq d-1$ such that for each $i\in[t_2]$, the length of $R_i$ is at most $p$ and the ends of $R_i$ are in $B'$.
        By concatenating $Q_0$, $P_u$, $Q_1$, $R$, and $Q_2$, we obtain a walk of $G$ between $u$ and $x$ having length at most
        \begin{align*}
            &\abs{E(Q_0)}+\abs{E(P_u)}+\abs{E(Q_1)}+\abs{E(R)}+\abs{E(Q_2)}\\
            &\leq p(d-1)+r+r+p(d-1)+r=2p(d-1)+3r\leq r'.
        \end{align*}
        Let $P$ be a path of $G$ between $u$ and $x$ consisting of edges of the walk.
        Let~$b$ be the vertex in $V(P)\cap(S\cup X_{\mathrm{cl}})$ such that $\dist_P(u,b)$ is minimum.
        Such~$b$ exists, because $x\in X_{\mathrm{cl}}$.
        
        We first show that $\dist_G(b,B')\leq r$.
        Note that $Q_0$ has no vertex in $S$.
        Since $u\notin\kappa'_1$, $\dist_{G\setminus S}(\mathcal{B}_u,X_{\mathrm{cl}})>r$.
        Since $p\leq2r+1$, for some $j\in[t_1]$, if $Q_0^j$ has a vertex in $X_{\mathrm{cl}}$, then $\dist_{G\setminus S}(\mathcal{B}_u,X_{\mathrm{cl}})\leq\dist_{G\setminus S}(B_u^*,X_{\mathrm{cl}})\leq r$, a contradiction.
        Therefore, $Q_0$ has no vertex in $X_{\mathrm{cl}}$.
        By Claim~\ref{clm:new3}, $V(P_u)\cap(S\cup X_{\mathrm{cl}})=\emptyset$.
        These imply that $b\in V(Q_1)\cup V(R)\cup V(Q_2)$.
        If $b\in V(Q_1)\cup V(Q_2)$, then $\dist_G(b,B')\leq r$ clearly.
        Since $p\leq2r+1$, for some $j\in[t_2]$, if $b\in R_j$, then $\dist_G(b,B')\leq r$.
        Therefore, $\dist_G(b,B')\leq r$.
        
        Since $B'\subseteq Z\setminus(N^r_G[D']\cup\{z\})$, $b$ is not contained in $D'$.
        Since $S\subseteq D_{\mathrm{buy}}\subseteq D'$, $b$ is contained in $X_{\mathrm{cl}}\setminus S$.
        Since the subpath of $P$ between $u$ and~$b$ is an $X_{\mathrm{cl}}$-avoiding path of length at most $r'$, $b$ is contained in $M_{r'}^G(u,X_{\mathrm{cl}})$.
        Since $\{u,z\}\subseteq\kappa_2\subseteq\lambda$ where $\lambda$ is an equivalence class of $\sim$, $M_{r'}^G(u,X_{\mathrm{cl}})$ and $M_{r'}^G(z,X_{\mathrm{cl}})$ are same.
        Therefore, $b\in M_{r'}^G(z,X_{\mathrm{cl}})\subseteq D'$, a contradiction.
    \end{subproof}
    
    Claim~\ref{clm:new5} contradicts the assumption that $D$ is a minimum-size $(p,r,\mathcal{F})$-cover of $Z\setminus\{z\}$ in $G$.
    Therefore, $Z\setminus\{z\}$ is a $(p,r,\mathcal{F})$-core of $A$ in $G$.
    We conclude the proof by scaling $\varepsilon$ to $\varepsilon/C$ throughout the reasoning.
\end{proof}

We present a linear kernel for the \textsc{Annotated $(p,r,\mathcal{F})$-Covering} problem on every class of graphs with bounded expansion, which generalizes the linear kernel of~\cite{KnS2016}.

\begin{THM}\label{thm:ker1'}
    For every class $\mathcal{C}$ of graphs with bounded expansion, there exists a function $f_{\mathrm{cov}}:\mathbb{N}\times\mathbb{N}\to\mathbb{N}$ satisfying the following.
    For every nonempty family $\mathcal{F}$ of connected graphs with at most~$d$ vertices and $p,r,r_0\in\mathbb{N}$ with $\max\{p,r_0\}\leq2r+1$, there exists a polynomial-time algorithm that given a graph $G\in\mathcal{C}$, $A\subseteq V(G)$, and $k\in\mathbb{N}$, either
    \begin{enumerate}
        \item[$\bullet$]    correctly decides that $\gamma_{p,r}^\mathcal{F}(G,A)>k$, or
        \item[$\bullet$]    outputs sets $Y\subseteq V(G)$ of size at most $f_{\mathrm{cov}}(r,d)\cdot k$ and $Z\subseteq A\cap Y$ such that $\gamma_{p,r}^\mathcal{F}(G[Y],Z)=\gamma_{p,r}^\mathcal{F}(G,A)$.
    \end{enumerate}
\end{THM}

To prove Theorem~\ref{thm:ker1'}, we will use the following lemma.

\begin{LEM}\label{lem:core'}
    For every class $\mathcal{C}$ of graphs with bounded expansion, there exists a function $f_{\mathrm{core}}:\mathbb{N}\times\mathbb{N}\to\mathbb{N}$ satisfying the following.
    For every nonempty family $\mathcal{F}$ of connected graphs with at most $d$ vertices and $p,r\in\mathbb{N}$ with $p\leq2r+1$, there exists a polynomial-time algorithm that given a graph $G\in\mathcal{C}$, $A\subseteq V(G)$, and $k\in\mathbb{N}$, either
    \begin{enumerate}
        \item[$\bullet$]    correctly decides that $\gamma_{p,r}^\mathcal{F}(G,A)>k$, or
        \item[$\bullet$]    outputs a $(p,r,\mathcal{F})$-core $Z$ of $A$ with $\abs{Z}\leq f_{\mathrm{core}}(r,d)\cdot k$.
    \end{enumerate}
\end{LEM}
\begin{proof}
    It easily follows from the proofs of Lemmas~\ref{lem:core} and~\ref{lem:core recursion} by setting $\xi:=2\cdot f_{\mathrm{cl}}(r')+d^2/4+s+1$ and replacing Lemmas~\ref{lem:proj} and~\ref{lem:cl} and Proposition~\ref{prop:apx} with Lemmas~\ref{lem:proj'} and ~\ref{lem:cl'} and Proposition~\ref{prop:approx1} for $r_0=2r$, respectively.
\end{proof}

We can now easily prove Theorem~\ref{thm:ker1'} from the proof of Theorem~\ref{thm:ker1} by replacing Lemmas~\ref{lem:core},~\ref{lem:proj}, and~\ref{lem:pth} with Lemmas~\ref{lem:core'},~\ref{lem:proj'}, and~\ref{lem:pth'}, respectively.

\subsection{Proof of the first main theorem}

The kernel of Theorem~\ref{thm:ker1} always returns an instance $(G[Y],Z,k)$ for the \textsc{Annotated $(p,r,\mathcal{F})$-Covering} problem, even if the kernel gets an input $(G,V(G),k)$, which is basically an input for the \textsc{$(p,r,\mathcal{F})$-Covering}.
From the resulting instance $(G[Y],Z,k)$, however, we can prove Theorem~\ref{thm:ker1 origin} by constructing an instance $(G',V(G'),k+1)$ of size $O(k^{1+\varepsilon})$ which is equivalent to $(G,V(G),k)$ in polynomial time.

To do this, we will use the following two lemmas.
For a nonnegative integer $q$ and a nonempty family $\mathcal{G}$ of graphs, a graph $H$ is \emph{$(q,\mathcal{G})$-critical} if one of the following holds.
\begin{enumerate}
    \item[$\bullet$]    $H$ is an $1$-vertex graph and $\mathcal{G}$ contains an $1$-vertex graph.
    \item[$\bullet$]    $H$ has at least two vertices, $H^q$ has an induced subgraph isomorphic to a graph in $\mathcal{G}$, and for every vertex $v$ of $H$, $(H\setminus v)^q$ has no induced subgraph isomorphic to a graph in $\mathcal{G}$.
\end{enumerate}

The following lemma shows the existence of a $(q,\mathcal{G})$-critical graph for every positive integer $q$.

\begin{LEM}\label{lem:critical}
    Let $q$ be a positive integer and $\mathcal{G}$ be a nonempty family of graphs.
    Let $F$ be an arbitrary graph in $\mathcal{G}$ and $d:=\abs{V(F)}$.
    There exists a $(q,\mathcal{G})$-critical graph having at most $d(dq+1)/2$ vertices.
    Moreover, if every graph in $\mathcal{G}$ has at most $d$ vertices, then one can construct in polynomial time the $(q,\mathcal{G})$-critical graph.
\end{LEM}
\begin{proof}
    Let $F_0$ be the $q$-subdivision of $F$.
    Since $F$ has at most $d(d-1)/2$ edges,
    \begin{equation*}
        \abs{V(F_0)}\leq d+\frac{d(d-1)(q-1)}{2}=d\cdot\frac{dq-d-q+3}{2}\leq\frac{d(dq+1)}{2}.
    \end{equation*}
    
    Let $H$ be a graph which is initially set as $F_0$.
    Note that $H^q$ has an induced subgraph isomorphic to $F\in\mathcal{G}$.
    If $\abs{V(H)}=1$, then $H$ is $(q,\mathcal{G})$-critical.
    Otherwise, for each vertex $v$ of $H$, we check whether $(H\setminus v)^q$ has an induced subgraph isomorphic to a graph in $\mathcal{G}$.
    If $H$ has no such vertex, then $H$ is $(q,\mathcal{G})$-critical.
    Otherwise, we set $H$ by $H\setminus v$ and do the above process until either $\abs{V(H)}=1$ or $H$ has no such a vertex.
    It is readily seen that the resulting graph is $(q,\mathcal{G})$-critical graph and has at most $d(dq+1)/2$ vertices.
    Whole these processes work in polynomial time when every graph in $\mathcal{G}$ has at most $d$ vertices.
\end{proof}

We remark that the positive condition of $q$ in Lemma~\ref{lem:critical} is necessary, because if $q=0$ and $\mathcal{G}$ contains no $1$-vertex graph, then there exists no $(q,\mathcal{G})$-critical graph.

The following lemma shows that for a $(q,\mathcal{G})$-critical graph $H$ and every vertex $x$ of $H$, $\{x\}$ is a $(q,\lfloor q/2\rfloor,\mathcal{G})$-cover of $H$.

\begin{LEM}\label{lem:guard}
    Let $q$ be a positive integer and $\mathcal{G}$ be a nonempty family of graphs.
    If $H$ is a $(q,\mathcal{G})$-critical graph and $V(H)$ has a subset $B$ such that $H^q[B]$ is isomorphic to a graph in $\mathcal{G}$, then for every vertex $x$ of~$H$, $B$ contains a vertex in $N^{\lfloor q/2\rfloor}_H[x]$.
\end{LEM}
\begin{proof}
    Suppose for contradiction that $B$ contains no vertex in $N^{\lfloor q/2\rfloor}_H[x]$.
    Since $H$ is $(q,\mathcal{G})$-critical, $(H\setminus x)^q[B]$ is isomorphic to no graph in $\mathcal{G}$.
    Since $H^q[B]$ is isomorphic to a graph in $\mathcal{G}$, $B$ contains distinct vertices $v$ and $w$ such that $v$ and $w$ are adjacent in $H^q$ and every path of $H$ between $v$ and $w$ having length at most $q$ should contain $x$.
    However, since neither $v$ nor $w$ is in $N^{\lfloor q/2\rfloor}_H[x]$, if $H$ has a path $P$ between $v$ and $w$ having $x$ as an internal vertex, then the length of $P$ is at least $2\lfloor q/2\rfloor+2>q$, contradicting that $v$ and $w$ are adjacent in $H^q$.
\end{proof}

We now prove Theorem~\ref{thm:ker1 origin}.

\domnowhere*

\begin{proof}
    Let $d$ be the maximum order of a graph in $\mathcal{F}$.
    We apply Theorem~\ref{thm:ker1} for $(G,V(G),k)$.
    We may assume that the first outcome of Theorem~\ref{thm:ker1} does not arise.
    Let $(G[Y],Z,k)$ be the resulting instance, and $F$ be a graph in $\mathcal{F}$ such that $\abs{V(F)}$ is minimum.
    Note that $\gamma_{p,r}^\mathcal{F}(G)=\gamma_{p,r}^\mathcal{F}(G[Y],Z)$.
    
    Suppose that $r=0$.
    Note that $p\leq1$.
    Then for every set $B\subseteq Z$, $G[Y]^p[B]$ is isomorphic to $G[Z]^p[B]$.
    Since $r=0$, every minimal $(p,0,\mathcal{F})$-cover of $Z$ in $G[Y]$ is a subset of $Z$.
    Therefore, for a set $D\subseteq Z$, $D$ is a $(p,0,\mathcal{F})$-cover of $Z$ in $G[Y]$ if and only if $D$ is a $(p,0,\mathcal{F})$-cover of $G[Z]$.
    Hence, $\gamma_{p,0}^\mathcal{F}(G[Y],Z)=\gamma_{p,0}^\mathcal{F}(G[Z])$.
    Let $G'$ be the disjoint union of $G[Z]$ and $F$.
    Note that $G'$ has at most $\abs{Z}+d\leq\abs{Y}+d$ vertices.
    Since $\abs{Y}=O(k^{1+\varepsilon})$, one can choose the function $g_{\mathrm{cov}}(0,d,\varepsilon)$ such that $\abs{V(G')}\leq g_{\mathrm{cov}}(0,d,\varepsilon)\cdot k^{1+\varepsilon}$.
    
    If $p=0$ and $\abs{V(F)}\geq2$, then $\gamma_{0,0}^\mathcal{F}(G[Y],Z)=0=\gamma_{0,0}^\mathcal{F}(G')$, so that the statement holds.
    Otherwise, either $p=0$ and $\abs{V(F)}=1$, or $p=1$ holds.
    In both cases, the following hold.
    \begin{enumerate}
        \item[$\bullet$]    For every $(p,0,\mathcal{F})$-cover $D$ of $Z$ in $G[Y]$ and every $v \in V(F)$, $D\cup \{v\}$ is a $(p,0,\mathcal{F})$-cover of $G'$.
        \item[$\bullet$]    For every $(p,0,\mathcal{F})$-cover $D$ of $G'$, $D$ should contain a vertex of $F$, and $D\setminus V(F)$ is a $(p,0,\mathcal{F})$-cover of $Z$ in $G[Y]$.
    \end{enumerate}
    Therefore, $\gamma_{p,0}^\mathcal{F}(G)=\gamma_{p,0}^\mathcal{F}(G[Y],Z)\leq k$ if and only if $\gamma_{p,0}^\mathcal{F}(G')\leq k+1$.
    Hence, the statement holds for $r=0$.
    
    Thus, we may assume that $r>0$.
    Suppose that $p=0$.
    Since $G^0$ is edgeless, if $\abs{V(F)}\geq2$, then $\gamma_{0,r}^\mathcal{F}(G)=0$, so the statement holds by taking $G'$ as an $1$-vertex graph.
    Thus, we may assume that $\abs{V(F)}=1$.
    We construct a graph $G'$ from $G[Y]$ as follows.
    \begin{enumerate}
        \item[$\bullet$]    Add two new vertices $h$ and $h'$.
        \item[$\bullet$]    For each vertex $v\in(Y\setminus Z)\cup\{h'\}$, connect $h$ and $v$ by a path $P_v$ of length $r$.
    \end{enumerate}
    Let $G'$ be the resulting graph.
    Note that $G'$ can be constructed in polynomial time, and
    \begin{equation*}
        \abs{V(G')}\leq1+\abs{Z}+r(\abs{Y\setminus Z}+1)\leq1+r(\abs{Y}+1).
    \end{equation*}
    Since $\abs{Y}=O(k^{1+\varepsilon})$, one can choose the function $g_{\mathrm{cov}}(r,d,\varepsilon)$ such that $\abs{V(G')}\leq g_{\mathrm{cov}}(r,d,\varepsilon)\cdot k^{1+\varepsilon}$.
    
    Since $\gamma_{0,r}^\mathcal{F}(G)=\gamma_{0,r}^\mathcal{F}(G[Y],Z)$, to show that $\gamma_{0,r}^\mathcal{F}(G)\leq k$ if and only if $\gamma_{0,r}^\mathcal{F}(G')\leq k+1$, it suffices to show that $\gamma_{0,r}^\mathcal{F}(G[Y],Z)\leq k$ if and only if $\gamma_{0,r}^\mathcal{F}(G')\leq k+1$.
    
    Firstly, suppose that $\gamma_{0,r}^\mathcal{F}(G[Y],Z)\leq k$.
    Then $G[Y]$ has a $(0,r,\mathcal{F})$-cover $D$ of $Z$ having size at most $k$.
    Since $V(G')\setminus Z\subseteq N^r_{G'}[h]$ and $Z\setminus N^r_{G'}[D]\subseteq Z\setminus N^r_{G[Y]}[D]$, we have that
    \begin{equation*}
        V(G')\setminus N^r_{G'}[D\cup\{h\}]\subseteq Z\setminus N^r_{G'}[D]\subseteq Z\setminus N^r_{G[Y]}[D].
    \end{equation*}
    Since $D$ is a $(0,r,\mathcal{F})$-cover of $Z$ in $G[Y]$ and $\abs{V(F)}=1$, $Z\setminus N^r_{G[Y]}[D]$ is empty, and so is $V(G')\setminus N^r_{G'}[D\cup\{h\}]$.
    Thus, $D\cup\{h\}$ is a $(0,r,\mathcal{F})$-cover of $G'$.
    Therefore, $\gamma_{0,r}^\mathcal{F}(G')\leq k+1$.
    
    Conversely, suppose that $\gamma_{0,r}^\mathcal{F}(G')\leq k+1$.
    Then $G'$ has a $(0,r,\mathcal{F})$-cover $D$ having size at most $k+1$.
    Note that $D$ contains a vertex in $V(P_{h'})$.
    Let $D'$ be the set obtained from $D$ by substituting each vertex $v\in D\setminus(Y\cup\{h\})$ with the end of $P_v$ distinct to $h$, and then removing $V(P_{h'})$.
    Note that $D'\subseteq Y$ and $\abs{D'}<\abs{D}\leq k+1$.
    
    Observe that $h$ is the only vertex in $V(P_{h'})$ at distance at most $r$ from $Y$.
    Since $N^r_{G'}[h]\cap Y=Y\setminus Z$, no vertex in $V(P_{h'})$ is at distance at most $r$ from $Z$, and therefore
    \begin{equation*}
        Z\setminus N^r_{G[Y]}[D']\subseteq Z\setminus N^r_{G'}[D\setminus V(P_{h'})]\subseteq Z\setminus N^r_{G'}[D].
    \end{equation*}
    Since $D$ is a $(0,r,\mathcal{F})$-cover of $G'$ and $\abs{V(F)}=1$, $Z\setminus N^r_{G'}[D']$ is empty, and so is $Z\setminus N^r_{G[Y]}[D']$.
    Thus, $D'$ is a $(0,r,\mathcal{F})$-cover of $Z$ in $G[Y]$.
    Therefore, $\gamma_{p,r}^\mathcal{F}(G[Y],Z)\leq k$.
    Hence, the statement holds for $r>0$ and $p=0$.
    
    Thus, we may further assume that $p>0$.
    By Lemma~\ref{lem:critical}, one can find in polynomial time a $(p,\mathcal{F})$-critical graph $H$ having at most $d(dp+1)/2$ vertices.
    Let $p':=\lfloor p/2\rfloor$ and $x$ be a vertex of $H$.
    We construct a graph $G'$ as follows.
    \begin{enumerate}
        \item[$\bullet$]    Take the disjoint union of $G[Y]$ and $H$, and add a new vertex $h$.
        \item[$\bullet$]    For each vertex $v\in(Y\setminus Z)\cup N^{p'}_H[x]$, connect $h$ and $v$ by a path $P_v$ of length $r$.
    \end{enumerate}
    Let $G'$ be the resulting graph.
    Note that $G'$ can be constructed in polynomial time, and
    \begin{equation*}
        \abs{V(G')}\leq1+\abs{Z}+r(\abs{Y\setminus Z}+\abs{V(H)})\leq1+r(\abs{Y}+d(dp+1)/2).
    \end{equation*}
    Since $\abs{Y}=O(k^{1+\varepsilon})$, one can choose the function $g_{\mathrm{cov}}(r,d,\varepsilon)$ such that $\abs{V(G')}\leq g_{\mathrm{cov}}(r,d,\varepsilon)\cdot k^{1+\varepsilon}$.
    
    We verify that for every set $B\subseteq Z$, $G'^p[B]$ is isomorphic to $G[Y]^p[B]$.
    If $G'$ has a path $P$ between two vertices in $Z$ such that $V(P)\setminus Y$ is nonempty, then $P$ should have $h$, and therefore the length of $P$ is at least $2r+2$.
    Thus, every path of $G'$ between two vertices in $Z$ having length at most $p$ is also a path of $G[Y]$.
    Therefore, for every set $B\subseteq Z$, $G'^p[B]$ is isomorphic to $G[Y]^p[B]$.
    
    We will use the following claim.
    
    \begin{CLM}\label{clm:preserve}
        For vertices $v,w\in V(H)$, $v$ and $w$ are adjacent in $H^p$ if and only if $v$ and $w$ are adjacent in $G'^p$.
    \end{CLM}
    \begin{subproof}
        Since $H$ is an induced subgraph of $G'$, the forward direction is obvious.
        We show the backward direction.
        Suppose for contradiction that $v$ and $w$ are nonadjacent in $H^p$ and $G'$ has a path $P$ of length at most $p$ between $v$ and $w$.
        Since $v$ and $w$ are nonadjacent in $H^p$, $P$ has a vertex not in $V(H)$.
        Thus, $P$ should have $h$.
        Since $h$ is a cut-vertex of $G'$ and the ends of $P$ are in $V(H)$, $P$ has vertices $v',w'\in N^{p'}_H[x]$ such that the subpath $Q_1$ of $P$ between $v'$ and $w'$ has $h$ as an internal vertex.
        Note that every vertex of $P$ which is not an internal vertex of $Q_1$ is contained in $V(H)$.
        Since both $v'$ and $w'$ are in $N^{p'}_H[x]$, $H$ has a path $Q_2$ of length at most $2p'=2\lfloor p/2\rfloor$ between $v'$ and~$w'$.
        
        Since $p\leq2r+1$, we have that $2\lfloor p/2\rfloor\leq2r$.
        Thus, the length of $Q_2$ is at most that of $Q_1$.
        Therefore, by substituting $Q_1$ with $Q_2$ from $P$, we obtain a path of $H$ between $v$ and $w$ having length at most~$p$, contradicting the assumption that $v$ and $w$ are nonadjacent in $H^p$.
    \end{subproof}
    
    Since $\gamma_{p,r}^\mathcal{F}(G)=\gamma_{p,r}^\mathcal{F}(G[Y],Z)$, to show that $\gamma_{p,r}^\mathcal{F}(G)\leq k$ if and only if $\gamma_{p,r}^\mathcal{F}(G')\leq k+1$, it suffices to show that $\gamma_{p,r}^\mathcal{F}(G[Y],Z)\leq k$ if and only if $\gamma_{p,r}^\mathcal{F}(G')\leq k+1$.
    
    Firstly, suppose that $\gamma_{p,r}^\mathcal{F}(G[Y],Z)\leq k$.
    Then $G[Y]$ has a $(p,r,\mathcal{F})$-cover $D$ of $Z$ having size at most $k$.
    Since $Y\setminus N^r_{G'}[h]=Z$, we have that
    \begin{equation*}
        V(G')\setminus N^r_{G'}[D\cup\{h\}]\subseteq (Z\setminus N^r_{G'}[D])\cup(V(H)\setminus N^{p'}_H[x]).
    \end{equation*}
    Note that $G'\setminus N^r_{G'}[D\cup\{h\}]$ has no path between $Z\setminus N^r_{G'}[D]$ and $V(H)\setminus N^{p'}_H[x]$.
    Since $D$ is a $(p,r,\mathcal{F})$-cover of $Z$ in $G[Y]$ and $H$ is $(p,\mathcal{F})$-critical, by Lemma~\ref{lem:guard}, $V(G')\setminus N^r_{G'}[D\cup\{h\}]$ has no subset $B$ such that $G'^p[B]$ is isomorphic to a graph in $\mathcal{F}$.
    Thus, $D\cup\{h\}$ is a $(p,r,\mathcal{F})$-cover of $G'$.
    Therefore, $\gamma_{p,r}^\mathcal{F}(G')\leq k+1$.
    
    Conversely, suppose that $\gamma_{p,r}^\mathcal{F}(G')\leq k+1$.
    Then $G'$ has a $(p,r,\mathcal{F})$-cover $D$ having size at most $k+1$.
    Since $H$ is $(p,\mathcal{F})$-critical, $V(H)$ has a subset $X$ such that $H^p[X]$ is isomorphic to a graph in $\mathcal{F}$.
    By Claim~\ref{clm:preserve}, $H^p[X]$ is isomorphic to $G'^p[X]$.
    Let
    \begin{equation*}
        W:=V(H)\cup\bigcup_{v\in N^{p'}_H[x]}V(P_v).
    \end{equation*}
    Note that $N^r_{G'}[X]\subseteq N^r_{G'}[V(H)]\subseteq W$.
    Thus, $D$ should contain at least one vertex in $W$.
    Let $D'$ be the set obtained from $D$ by substituting each vertex $v\in D\setminus(W\cup Y)$ with the end of $P_v$ distinct to $h$, and then removing $W$.
    Note that $D'\subseteq Y$ and $\abs{D'}<\abs{D}\leq k+1$.
    
    Observe that $h$ is the only vertex in $W$ at distance at most $r$ from $Y$.
    In addition, $N^r_{G'}[h]\cap Y=Y\setminus Z$.
    Thus, no vertex in $W$ is at distance at most $r$ from $Z$, and therefore
    \begin{equation*}
        Z\setminus N^r_{G[Y]}[D']\subseteq Z\setminus N^r_{G'}[D\setminus W]\subseteq Z\setminus N^r_{G'}[D].
    \end{equation*}
    Recall that for every set $B\subseteq Z$, $G'^p[B]$ is isomorphic to $G[Y]^p[B]$.
    Since $D$ is a $(p,r,\mathcal{F})$-cover of $G'$, $Z\setminus N^r_{G[Y]}[D']$ has no subset $B$ such that $G[Y]^p[B]$ is isomorphic to a graph in $\mathcal{F}$.
    Thus, $D'$ is a $(p,r,\mathcal{F})$-cover of $Z$ in $G[Y]$.
    Therefore, $\gamma_{p,r}^\mathcal{F}(G[Y],Z)\leq k$.
    Hence, the statement holds for $r>0$ and $p>0$, and this completes the proof.
\end{proof}

Similarly, we can show the following.

\begin{THM}\label{thm:ker1' origin}
    For every class $\mathcal{C}$ of graphs with bounded expansion, there exists a function $g_{\mathrm{cov}}:\mathbb{N}\times\mathbb{N}\to\mathbb{N}$ satisfying the following.
    For every nonempty family $\mathcal{F}$ of connected graphs with at most~$d$ vertices and $p,r\in\mathbb{N}$ with $p\leq2r+1$, there exists a polynomial-time algorithm that given a graph $G\in\mathcal{C}$ and $k\in\mathbb{N}$, either
    \begin{enumerate}
        \item[$\bullet$]    correctly decides that $\gamma_{p,r}^\mathcal{F}(G)>k$, or
        \item[$\bullet$]    outputs a graph $G'$ such that $\abs{V(G')}=g_{\mathrm{cov}}(r,d)\cdot k$, and $\gamma_{p,r}^\mathcal{F}(G)\leq k$ if and only if $\gamma_{p,r}^\mathcal{F}(G')\leq k+1$.
    \end{enumerate}
\end{THM}

\section{Kernels for the \textsc{$(p,r,\mathcal{F})$-Packing} problems}\label{sec:ker2}

Let $p$ and $r$ be nonnegative integers with $p\leq r+1$, and $\mathcal{F}$ be a nonempty finite family of connected graphs.
We present an almost linear kernel for the \textsc{$(p,r,\mathcal{F})$-Packing} problem on every nowhere dense class of graphs, which generalizes the almost linear kernel of~\cite{PS2021} for the \textsc{Distance-$r$ Independent Set} problem.

We first present an almost linear kernel for the \textsc{Annotated $(p,r,\mathcal{F})$-Packing} problem.

\begin{THM}\label{thm:ker2}
    For every nowhere dense class $\mathcal{C}$ of graphs, there exists a function $f_{\mathrm{pck}}:\mathbb{N}\times\mathbb{N}\times\mathbb{R}_+\to\mathbb{N}$ satisfying the following.
    For every nonempty family $\mathcal{F}$ of connected graphs with at most~$d$ vertices, $p,r\in\mathbb{N}$ with $p\leq2\lfloor r/2\rfloor+1$, and $\varepsilon>0$, there exists a polynomial-time algorithm that given a graph $G\in\mathcal{C}$, $A\subseteq V(G)$, and $k\in\mathbb{N}$, either 
    \begin{enumerate}
        \item[$\bullet$]    correctly decides that $\alpha_{p,r}^\mathcal{F}(G,A)>k$, or
        \item[$\bullet$]    outputs sets $Y\subseteq V(G)$ of size at most $f_{\mathrm{pck}}(r,d,\varepsilon)\cdot k^{1+\varepsilon}$ and $Z\subseteq A\cap Y$ such that $\alpha_{p,r}^\mathcal{F}(G,A)\geq k$ if and only if $\alpha_{p,r}^\mathcal{F}(G[Y],Z)\geq k$.
    \end{enumerate} 
\end{THM}

The following lemma is a key lemma for proving Theorem~\ref{thm:ker2}.

\begin{LEM}\label{lem:rd}
    For every nowhere dense class $\mathcal{C}$ of graphs, there exists a function $f_{\mathrm{rd}}:\mathbb{N}\times\mathbb{N}\times\mathbb{R}_+\to\mathbb{N}$ satisfying the following.
    For every nonempty family $\mathcal{F}$ of connected graphs with at most~$d$ vertices, $p,r\in\mathbb{N}$ with $p\leq2\lfloor r/2\rfloor+1$, and $\varepsilon\in\mathbb{R}_+$, there exists a polynomial-time algorithm that given a graph $G\in\mathcal{C}$, $A\subseteq V(G)$, and $k\in\mathbb{N}$, the algorithm either
    \begin{enumerate}
        \item[$\bullet$]    correctly decides that $\alpha_{p,r}^\mathcal{F}(G,A)>k$, or
        \item[$\bullet$]    outputs a set $Z\subseteq A$ of size at most $f_{\mathrm{rd}}(r,d,\varepsilon)\cdot k^{1+\varepsilon}$ such that $\alpha_{p,r}^\mathcal{F}(G,A)\geq k$ if and only if $\alpha_{p,r}^\mathcal{F}(G,Z)\geq k$.
    \end{enumerate}
\end{LEM}

We first prove Theorem~\ref{thm:ker2} by using Lemma~\ref{lem:rd}.

\begin{proof}[Proof of Theorem~\ref{thm:ker2}]
    We apply the algorithm of Lemma~\ref{lem:rd}.
    We may assume that the algorithm finds a set $Z\subseteq A$ of size at most $f_{\mathrm{rd}}(r,d,\varepsilon)\cdot k^{1+\varepsilon}$ such that $\alpha_{p,r}^\mathcal{F}(G,A)\geq k$ if and only if $\alpha_{p,r}^\mathcal{F}(G,Z)\geq k$.
    
    By Lemma~\ref{lem:pth}, one can find in polynomial time an $(r+1)$-path closure $Y$ of $Z$ in $G$ with
    \begin{equation*}
        \abs{Y}\leq f_{\mathrm{pth}}(r+1,\varepsilon)\cdot\abs{Z}^{1+\varepsilon}\leq f_{\mathrm{pth}}(r+1,\varepsilon)\cdot f_{\mathrm{rd}}(r,d,\varepsilon)^{1+\varepsilon}\cdot k^{1+3\varepsilon}.
    \end{equation*}
    Note that $Z\subseteq A\cap Y$.
    Since $Y$ is an $(r+1)$-path closure of $Z$ in $G$ and $p\leq2\lfloor r/2\rfloor+1\leq r+1$, we deduce that $\alpha_{p,r}^\mathcal{F}(G,A)\geq k$ if and only if $\alpha_{p,r}^\mathcal{F}(G,Z)\geq k$ if and only if $\alpha_{p,r}^\mathcal{F}(G[Y],Z)\geq k$.
    
    By scaling $\varepsilon$ accordingly, one can choose the function $f_{\mathrm{pck}}(r,d,\varepsilon)$ with $\abs{Y}\leq f_{\mathrm{pck}}(r,d,\varepsilon)\cdot k^{1+\varepsilon}$, and this completes the proof.
\end{proof}

To prove Lemma~\ref{lem:rd}, we will use the following lemma.

\begin{LEM}\label{lem:rd recursion}
    For every nowhere dense class $\mathcal{C}$ of graphs, there exist a function $f_{\mathrm{rd}}:\mathbb{N}\times\mathbb{N}\times\mathbb{R}_+\to\mathbb{N}$ satisfying the following.
    For every nonempty family $\mathcal{F}$ of connected graphs with at most $d$ vertices, $p,r\in\mathbb{N}$ with $p\leq2\lfloor r/2\rfloor+1$, and $\varepsilon>0$, there exists a polynomial-time algorithm that given a graph $G\in\mathcal{C}$, $A\subseteq V(G)$ with $\abs{A}>f_{\mathrm{rd}}(r,d,\varepsilon)\cdot k^{1+\varepsilon}$, and $k\in\mathbb{N}$, either
    \begin{enumerate}
        \item[$\bullet$]    correctly decides that $\alpha_{p,r}^\mathcal{F}(G,A)>k$, or
        \item[$\bullet$]    outputs a vertex $z\in A$ such that $\alpha_{p,r}^\mathcal{F}(G,A)\geq k$ if and only if $\alpha_{p,r}^\mathcal{F}(G,A\setminus\{z\})\geq k$.
    \end{enumerate}
\end{LEM}

We can prove Lemma~\ref{lem:rd} by iteratively applying Lemma~\ref{lem:rd recursion} to $G$ with $A$ at most $\abs{A}$ times.
To prove Lemma~\ref{lem:rd recursion}, we will use the approximation algorithm in Proposition~\ref{prop:approx2}.

To design the algorithm in Proposition~\ref{prop:approx2}, we will use the following approximation algorithm, which generalizes the approximation algorithm of~\cite{Dvorak2013} for the \textsc{Distance-$r$ Dominating Set}.

\begin{algorithm}[t]
\caption{Finding the $(p,r,\mathcal{F})$-cover of $A$ in $G$.}
\label{alg:cover}
\begin{algorithmic}[1]
	\State{Fix a linear ordering $L$ of $V(G)$.}
	\State{Initialize $s:=0$, $D_s:=\emptyset$, $M_s:=\emptyset$, and $R_s:=A$.}
	\While{there is a set $R\subseteq R_s$ such that $G^p[R]$ is isomorphic to a graph in $\mathcal{F}$}\label{line:begin if}
		\State{Set $s:=s+1$.}
		\State{Find the set $X_s:=\{x^s_1,\ldots,x^s_{\abs{X_s}}\}\subseteq R_s$ with $x^s_1<_L\cdots<_Lx^s_{\abs{X_s}}$ such that $(x^s_1,\ldots,x^s_{\abs{X_s}})$ is the lexicographically smallest tuple satisfying that $G^p[X_s]$ is isomorphic to a graph in $\mathcal{F}$.}
		\State{Set $M_s:=M_{s-1}\cup\{X_s\}$.}
		\State{Set $W_s:=\bigcup_{i\in[\abs{X_s}]}\wre_{r_0}[G,L,x^s_i]$.}
		\State{Set $D_s:=D_{s-1}\cup W_s$.}\label{line:D}
        \State{Set $R_s:=R_{s-1}\setminus N_G^r[W_s]$.}\label{line:end if}
    \EndWhile\\
	\Return $s$, $M_s$, $D_s$, and $R_s$.
\end{algorithmic}
\end{algorithm}

\begin{PROP}\label{prop:approx1}
     For every nonempty family $\mathcal{F}$ of connected graphs with at most $d$ vertices and $p,r,r_0\in\mathbb{N}$ with $r_0\leq2r+1$, there exists a polynomial-time algorithm that given a graph $G$ and $A\subseteq V(G)$, outputs a $(p,r,\mathcal{F})$-cover of $A$ in $G$ having size at most $(d\cdot\wc_{r_0}(G)+1)^2\cdot\alpha_{p,r_0}^\mathcal{F}(G,A)$.
\end{PROP}
\begin{proof}
    Firstly, we show that Algorithm~\ref{alg:cover} terminates in polynomial time.
	For each $i\in[s]$, by investigating all subsets of $R_i$ having size at most~$d$, we can enumerate in polynomial time all tuples $(a_1,\ldots,a_\ell)$ of distinct vertices such that $\{a_1,\ldots,a_\ell\}\subseteq R_i$ and $G^p[\{a_1,\ldots,a_\ell\}]$ is isomorphic to a graph in~$\mathcal{F}$.
	Thus, by lexicographically ordering all such tuples, we can find the set $X_s$ in polynomial time.
	If $R_i$ has a subset $R$ such that $G^p[R]$ is isomorphic to a graph in $\mathcal{F}$, then $\abs{R_i}$ is decreased at the end of the iteration.
	Thus, unless $R_i$ has no such subset, the loop runs, and it terminates in at most $\abs{V(G)}$ iterations.
	Therefore, Algorithm~\ref{alg:cover} terminates in polynomial time.
	
	Secondly, we show that $D_s$ is a $(p,r,\mathcal{F})$-cover of $A$ in $G$.
	We may assume that $s>0$.
	Since $R_s$ has no subset $R$ such that $G^p[R]$ is isomorphic to a graph in $\mathcal{F}$, it suffices to show that for each $i\in[s]$, $R_i=A\setminus N_G^r[D_i]$.
	We proceed by induction on $i$ to show that $R_i=A\setminus N_G^r[D_i]$.
	Since $R_0=A$ and $D_1=W_1$, the equality holds for $i=1$.
	Thus, we may assume that $i>1$.
	By the inductive hypothesis, $R_{i-1}=A\setminus N_G^r[D_{i-1}]$.
	Note that $N_G^r[D_{i-1}\cup W_i]=N_G^r[D_{i-1}]\cup N_G^r[W_i]$.
	Since $R_i:=R_{i-1}\setminus N_G^r[W_i]$ and $D_i:=D_{i-1}\cup W_i$,
	\begin{align*}
	    R_i&=R_{i-1}\setminus N_G^r[W_i]\\
	    &=(A\setminus N_G^r[D_{i-1}])\setminus N_G^r[W_i]\\
	    &=A\setminus N_G^r[D_{i-1}\cup W_i]=A\setminus N_G^r[D_i].
	\end{align*}
	Therefore, by induction, $D_s$ is a $(p,r,\mathcal{F})$-cover of $A$ in $G$.
	
	We now show that $\abs{D_s}\leq(d\cdot\wc_{r_0}(G)+1)^2\cdot\alpha_{p,r_0}^\mathcal{F}(G,A)$.
	If $s=0$, then $\gamma_{p,r}^\mathcal{F}(G,A)=0$, so there is nothing to prove.
	Thus, we may assume that $s>0$.
    Since $\abs{D_s}\leq s\cdot d\cdot\wc_{r_0}(G)$, it suffices to find a $(p,r_0,\mathcal{F})$-packing $M\subseteq M_s$ of $A$ in $G$ with $s\leq(d\cdot\wc_{r_0}(G)+1)\cdot\abs{M}$.
	To find such $M$, we construct an auxiliary graph $H$ on the vertex set $[s]$ such that for distinct $i,j\in[s]$, $i$ and $j$ are adjacent in $H$ if and only if $\dist_G(X_i,X_j)\leq r_0$.
	For every $i\in[s]$, since $G^p[X_i]$ is isomorphic to a graph in $\mathcal{F}$, if $H$ has an independent set $I$, then $\{X_i:i\in I\}$ is a $(p,r_0,\mathcal{F})$-packing of $A$ in $G$.
	Therefore, it suffices to find an independent set in $H$ having size at least $s/(d\cdot\wc_{r_0}(G)+1)$.
	
	To do this, we show that the degeneracy of $H$ is at most $d\cdot\wc_{r_0}(G)$.
    Since $\dist_G(X_i,X_j)>r$ for all distinct $i,j\in[s]$, if $r_0\leq r$, then $H$ has no edge.
    Thus, we may assume that $r_0>r$.
	For each $w\in D_s$, let $T_w$ be the set of all integers $s'\in[s]$ such that
	\begin{equation*}
	    w\in\bigcup_{\ell\in[\abs{X_{s'}}]}\wre_r[G,L,x^{s'}_\ell].
	\end{equation*}
	
	We will use the following two claims.
	
	\begin{CLM}\label{clm:new11}
	    For $1\leq j<i\leq s$, if $i$ and $j$ are adjacent in $H$, then $j\in\bigcup_{w\in W_i}T_w$.
	\end{CLM}
	\begin{subproof}
	    Since $i$ and $j$ are adjacent in $H$, $G$ has a path $P$ of length at most $r_0$ between $c_j\in X_j$ and $c_i\in X_i$.
	    By taking a shorter path if necessary, we may assume that $c_i$ and $c_j$ are the only vertices in $V(P)\cap(X_i\cup X_j)$.
	    Let $z$ be the least vertex of $P$ in the ordering $L$.
	    Then $z$ is weakly $r_0$-accessible in $G$ from each of $c_i$ and $c_j$, and therefore $z\in W_j\cap W_i$.
	    Since $X_i\subseteq R_j=R_{j-1}\setminus N_G^r[W_j]$, we have that $\dist_G(X_i,z)\geq\dist_G(X_i,W_j)>r$.
	    Since $r_0\leq2r+1$, the length of the subpath of $P$ between $z$ and~$c_j$ is at most $r$.
        Since $z$ is the least vertex of $P$, $z$ is weakly $r$-accessible in~$G$ from $c_j\in X_j$, and therefore $j\in T_z$.
	    Since $z\in W_i$, we have that $j\in T_z\subseteq\bigcup_{w\in W_i}T_w$, and this proves the claim.
	\end{subproof}
	
	\begin{CLM}\label{clm:new12}
	    For each vertex $w\in D_s$, $\abs{T_w}\leq1$.
	\end{CLM}
	\begin{subproof}
	    Suppose that $D_s$ contains a vertex $w$ with $\abs{T_w}\geq2$.
	    Let $q$ be the smallest integer in $T_w$.
	    Since
	    \begin{equation*}
	        w\in\bigcup_{\ell\in[\abs{X_q}]}\wre_r[G,L,x^q_\ell]\subseteq W_q,
	    \end{equation*}
	    no vertex in $N_G^r[w]$ is in $R_q$.
	    For every $q'\in T_w\setminus\{q\}$, by the definition of $T_w$, $X_{q'}$ has a vertex at distance at most $r$ from $w$.
	    However, since $q<q'$, at least one vertex in $X_{q'}$ should be removed, a contradiction.
	\end{subproof}
	
	For each $s'\in[s]$, by Claim~\ref{clm:new12}, $\abs{\bigcup_{w\in W_{s'}}T_w}\leq\abs{W_{s'}}\leq d\cdot\wc_{r_0}(G)$.
	By Claim~\ref{clm:new11}, $i$ is adjacent in $H$ to at most $d\cdot\wc_{r_0}(G)$ vertices in $[i]$.
	Therefore, the degeneracy of $H$ is at most $d\cdot\wc_{r_0}(G)$, and this completes the proof.
\end{proof}

As mentioned before, we will use the following approximation algorithm to prove Lemma~\ref{lem:rd recursion}.
This algorithm generalizes the approximation algorithm of~\cite{PS2021} for the \textsc{Distance-$r$ Dominating Set} problem.

\begin{PROP}\label{prop:approx2}
    For every nowhere dense class $\mathcal{C}$ of graphs, there exists a function $f_{\mathrm{dual}}:\mathbb{N}\times\mathbb{N}\times\mathbb{R}_+\to\mathbb{N}$ satisfying the following.
    For every nonempty family $\mathcal{F}$ of connected graphs with at most~$d$ vertices, $p,r,r_0\in\mathbb{N}$ with $\max\{p,r_0\}\leq2r+1$, and $\varepsilon>0$, there exists a polynomial-time algorithm that given a graph $G\in\mathcal{C}$ and $A\subseteq V(G)$, outputs a $(p,r,\mathcal{F})$-cover of $A$ in $G$ having size at most $f_{\mathrm{dual}}(r,d,\varepsilon)\cdot\alpha_{p,r_0}^\mathcal{F}(G,A)^{1+\varepsilon}$.
\end{PROP}
\begin{proof}
    Let $\delta>0$ be a constant depending on $\varepsilon$ which will be defined later, and $\gamma:=\gamma_{p,r}^\mathcal{F}(G,A)$.
    If $\gamma=0$, then the statement holds by taking an empty set as a $(p,r,\mathcal{F})$-cover of $A$ in $G$.
    Note that $\gamma=0$ if and only if $\alpha_{p,r}^\mathcal{F}(G,A)=0$, and one can check in polynomial time whether $\alpha_{p,r}^\mathcal{F}(G,A)>0$.
    Thus, we may assume that $\gamma>0$.
    
    By Proposition~\ref{prop:apx}, one can find in polynomial time a $(p,r,\mathcal{F})$-cover $X$ of $A$ in $G$ having size at most $q:=f_{\mathrm{apx}}(r,d,\delta)\cdot\gamma^{1+\delta}$.
    Note that $q\geq\gamma\geq1$.
    
    We apply Theorem~\ref{thm:ker1} with $k=q$.
    Since $q\geq\gamma$, the first outcome of Theorem~\ref{thm:ker1} does not arise.
    Thus, the algorithm finds sets $Y\subseteq V(G)$ of size at most $f_{\mathrm{cov}}(r,d,\delta)\cdot q^{1+\delta}$ and $Z\subseteq A\cap Y$ such that $\gamma=\gamma_{p,r}^\mathcal{F}(G[Y],Z)$.
    By Observation~\ref{obs:obs}, $\alpha_{p,r_0}^\mathcal{F}(G[Y],Z)=\alpha_{p,r_0}^\mathcal{F}(G,Z)$.
    By Lemma~\ref{lem:zhu},
    \begin{equation*}
        1\leq\wc_{r_0}(G[Y])\leq f_{\mathrm{wcol}}(r_0,\delta)\cdot\abs{Y}^\delta\leq f_{\mathrm{wcol}}(r_0,\delta)\cdot f_{\mathrm{cov}}(r,d,\delta)^\delta\cdot q^{2\delta}.
    \end{equation*}
    
    By Proposition~\ref{prop:approx1},
    \begin{align*}
        \gamma
        &=\gamma_{p,r}^\mathcal{F}(G[Y],Z)\\
        &\leq(d\cdot\wc_{r_0}(G[Y])+1)^2\cdot\alpha_{p,r_0}^\mathcal{F}(G[Y],Z)\\
        &=(d\cdot\wc_{r_0}(G[Y])+1)^2\cdot\alpha_{p,r_0}^\mathcal{F}(G,Z)\\
        &\leq(d\cdot\wc_{r_0}(G[Y])+1)^2\cdot\alpha_{p,r_0}^\mathcal{F}(G,A)\\
        &\leq(d\cdot f_{\mathrm{wcol}}(r_0,\delta)\cdot f_{\mathrm{cov}}(r,d,\delta)^\delta\cdot q^{2\delta}+1)^2\cdot\alpha_{p,r_0}^\mathcal{F}(G,A).
    \end{align*}
    Note that there is a constant $C>0$ with
    \begin{equation*}
        (d\cdot f_{\mathrm{wcol}}(r_0,\delta)\cdot f_{\mathrm{cov}}(r,d,\delta)^\delta\cdot q^{2\delta}+1)^2\leq C\cdot d^2\cdot f_{\mathrm{wcol}}(r_0,\delta)^2\cdot f_{\mathrm{cov}}(r,d,\delta)^{2\delta}\cdot q^{4\delta}.
    \end{equation*}
    Thus, $\gamma\leq C'\cdot\gamma^{8\delta}\cdot\alpha_{p,r_0}^\mathcal{F}(G,A)$ for $C':=C\cdot d^2\cdot f_{\mathrm{wcol}}(r_0,\delta)^2\cdot f_{\mathrm{cov}}(r,d,\delta)^{2\delta}\cdot f_{\mathrm{apx}}(r,d,\varepsilon)^{4\delta}$, and therefore $\gamma^{1-8\delta}\leq C'\cdot\alpha_{p,r_0}^\mathcal{F}(G,A)$.
    
    Let $\delta$ be a real satisfying that $(1-8\delta)(1+9\delta)\geq1$ and $\delta\leq\varepsilon/11$.
    Such $\delta$ exists, because $(1-8x)(1+9x)\geq1$ for a sufficiently small $x$.
    We remark that $(1+9\delta)(1+\delta)\leq1+11\delta$.
    Then
    \begin{equation*}
        \gamma\leq(C'\cdot\alpha_{p,r_0}^\mathcal{F}(G,A))^{1/(1-8\delta)}\leq(C'\cdot\alpha_{p,r_0}^\mathcal{F}(G,A))^{1+9\delta}.
    \end{equation*}
    Since $(1+9\delta)(1+\delta)\leq1+11\delta$ and $\delta\leq\varepsilon/11$,
    \begin{align*}
        \abs{X}
        &\leq q=f_{\mathrm{apx}}(r,d,\delta)\cdot\gamma^{1+\delta}\\
        &\leq f_{\mathrm{apx}}(r,d,\delta)\cdot(C'\cdot\alpha_{p,r_0}^\mathcal{F}(G,A))^{1+11\delta}\\
        &\leq f_{\mathrm{apx}}(r,d,\delta)\cdot(C'\cdot\alpha_{p,r_0}^\mathcal{F}(G,A))^{1+\varepsilon}.
    \end{align*}
    Thus, the statement holds by setting $f_{\mathrm{dual}}(r,d,\varepsilon):=f_{\mathrm{apx}}(r,d,\delta)\cdot C'^{1+\varepsilon}$.
\end{proof}

We now prove Lemma~\ref{lem:rd recursion}.

\begin{proof}[Proof of Lemma~\ref{lem:rd recursion}]
    The function $f_{\mathrm{rd}}(r,d,\varepsilon)$ will be defined later.
    At the beginning, we assume that $\abs{A}>f_{\mathrm{rd}}(r,d,\varepsilon)\cdot k^{1+C\varepsilon}$ for some constant $C$, and at the end, we scale $\varepsilon$ accordingly.
    
    If $A$ contains a vertex $v$ such that for every $B\subseteq A\setminus\{v\}$ with $\abs{B}\leq d-1$, $G^p[B\cup\{v\}]$ is isomorphic to no graphs in $\mathcal{F}$, then the statement holds by taking $v$ as $z$.
    Thus, we may assume that for every $v\in A$, $A\setminus\{v\}$ has a subset $B$ such that $G^p[B\cup\{v\}]$ is isomorphic to a graph in $\mathcal{F}$.
    
    Since $p\leq2\lfloor r/2\rfloor+1$, by Proposition~\ref{prop:approx2}, one can find a $(p,\lfloor r/2\rfloor,\mathcal{F})$-cover $X$ of $A$ in $G$ having size at most $f_{\mathrm{dual}}(\lfloor r/2\rfloor,d,\varepsilon)\cdot\alpha_{p,r}^\mathcal{F}(G,A)^{1+\varepsilon}$.
    If $\abs{X}>f_{\mathrm{dual}}(\lfloor r/2\rfloor,d,\varepsilon)\cdot k^{1+\varepsilon}$, then $\alpha_{p,r}^\mathcal{F}(G,A)>k$.
    Thus, we may assume that $\abs{X}\leq f_{\mathrm{dual}}(\lfloor r/2\rfloor,d,\varepsilon)\cdot k^{1+\varepsilon}$.
    Let
    \begin{equation*}
        r':=4pd+3r.
    \end{equation*}
    By Lemma~\ref{lem:cl}, one can find in polynomial time an $(r',f_{\mathrm{cl}}(r',\varepsilon)\cdot\abs{X}^\varepsilon)$-close set $X_{\mathrm{cl}}\supseteq X$ of size at most $f_{\mathrm{cl}}(r',\varepsilon)\cdot\abs{X}^{1+\varepsilon}\leq f_{\mathrm{cl}}(r',\varepsilon)\cdot f_{\mathrm{dual}}(\lfloor r/2\rfloor,d,\varepsilon)^{1+\varepsilon}\cdot k^{1+3\varepsilon}$.
    
    We define an equivalence relation $\sim$ on $A\setminus X_{\mathrm{cl}}$ such that for $u,v\in A\setminus X_{\mathrm{cl}}$, $u\sim v$ if and only if $\rho_{r'}^G[u,X_{\mathrm{cl}}]=\rho_{r'}^G[v,X_{\mathrm{cl}}]$.
    Then by Lemma~\ref{lem:proj},
    \begin{align*}
        \idx(\sim)
        &\leq f_{\mathrm{proj}}(r',\varepsilon)\cdot\abs{X_{\mathrm{cl}}}^{1+\varepsilon}\\
        &\leq f_{\mathrm{proj}}(r',\varepsilon)\cdot f_{\mathrm{cl}}(r',\varepsilon)^{1+\varepsilon}\cdot f_{\mathrm{dual}}(\lfloor r/2\rfloor,d,\varepsilon)^{1+3\varepsilon}\cdot k^{1+7\varepsilon}.
    \end{align*}
    Let $p(r')$ and $s:=s(r')$ be the constants in Theorem~\ref{thm:uqw}.
    Let
    \begin{align*}
        \xi&:=d\cdot(f_{\mathrm{cl}}(r',\varepsilon)\cdot f_{\mathrm{dual}}(\lfloor r/2\rfloor,d,\varepsilon)^\varepsilon\cdot k^{2\varepsilon}+s+d^2/4+1),\\
        m&:=2^{2^{d^2/2}\cdot(r+2)^{sd}}\cdot\xi+1.
    \end{align*}
    By setting $C=7+2\cdot p(r')$, one can choose the function $f_{\mathrm{rd}}(r,d,\varepsilon)$ with
    \begin{equation*}
        f_{\mathrm{rd}}(r,d,\varepsilon)\cdot k^{1+C\varepsilon}\geq\abs{X_{\mathrm{cl}}}+\idx(\sim)\cdot m^{p(r')}.
    \end{equation*}
    Since $\abs{A}>f_{\mathrm{rd}}(r,d,\varepsilon)\cdot k^{1+C\varepsilon}$, we have that $\abs{A\setminus X_{\mathrm{cl}}}>\idx(\sim)\cdot m^{p(r')}$.
    Thus, by the pigeonhole principle, there is an equivalence class $\lambda$ of $\sim$ with $\abs{\lambda}>m^{p(r')}$.
    By Theorem~\ref{thm:uqw}, one can find in polynomial time sets $S\subseteq V(G)$ and $L\subseteq\lambda\setminus S$ such that $\abs{S}\leq s$, $\abs{L}\geq m$, and $L$ is distance-$r'$ independent in $G\setminus S$.

    We are going to find a desired vertex $z$ from $L$.
    To do this, we define the following.
    For each $i\in[d]$, let $\mathcal{G}_i$ be the set of all graphs whose vertex sets are $[i]$.
    Note that $\abs{\mathcal{G}_i}=2^{i(i-1)/2}$ for each $i\in[d]$.
    Let $\mathcal{H}'$ be the set of functions $\rho:S\to[r+1]\cup\{\infty\}$.
    Since $\abs{S}\leq s$, we have that $\abs{\mathcal{H}'}\leq(r+2)^s$.
    For each $i\in[d]$, let $\mathcal{H}'_i$ be the set of all vectors $(h_1,\ldots,h_i,g)$ of length $i+1$ where $h_j\in\mathcal{H}'$ for each $j\in[i]$ and $g\in\mathcal{G}_i$.
    Let $\overline{\mathcal{H}'}:=\bigcup_{i=1}^d\mathcal{H}'_i$.
    Note that
    \begin{align*}
        \abs{\overline{\mathcal{H}'}}
        &=\sum_{i=1}^d\abs{\mathcal{H}'_i}=\sum_{i=1}^d(\abs{\mathcal{H}'}^i\cdot\abs{\mathcal{G}_i})\leq\sum_{i=1}^d((r+2)^{si}\cdot2^{i(i-1)/2})\\
        &\leq2^{d(d-1)/2}\cdot\sum_{i=1}^d(r+2)^{si}\leq2^{d(d-1)/2}\cdot2(r+2)^{sd}\leq2^{d^2/2}\cdot(r+2)^{sd}.
    \end{align*}
    
    Let $\ell:=\abs{\overline{\mathcal{H}'}}$.
    We take an arbitrary ordering $\sigma_1,\ldots,\sigma_\ell$ of $\overline{\mathcal{H}'}$.
    For each $v\in L$, let $\mathcal{A}_v:=\emptyset$ and $\mathbf{x}(v)$ be a zero vector of length $\ell$.
    One can enumerate in polynomial time the sets $B\subseteq A\setminus\{v\}$ of size at most $d-1$ such that $G^p[B\cup\{v\}]$ is isomorphic to a graph in $\mathcal{F}$ in polynomial time.
    For each such $B$, we do the following.
    If there is an index $i\in[\ell]$ such that the $i$-th entry of $\mathbf{x}(v)$ is $0$ and for $\sigma_i=(h_1^i,\ldots,h_t^i,g_i)\in\overline{\mathcal{H}'}$, there exists
    \begin{enumerate}
        \item[$\bullet$]    an isomorphism $\phi_i:(B\setminus S)\cup\{v\}\to[t]$ between $(G\setminus S)^p[(B\setminus S)\cup\{v\}]$ and $g_i$ where $\phi_i(v)=1$ and for each $j\in[t]$, $\rho_r^G[\phi^{-1}(j),S]=h^i_j$,
    \end{enumerate}
    then we put $B$ into $\mathcal{A}_v$ and convert the $i$-th entry of $\mathbf{x}(v)$ to $1$.
    Otherwise, we do nothing for the chosen $B$.
    Since $\abs{B}\leq d-1$, one can check in polynomial time whether $B$ satisfies the conditions.
    Thus, the resulting $\mathcal{A}_v$ and $\mathbf{x}(v)$ can be computed in polynomial time.
    
    For each $v\in L$, since $A\setminus\{v\}$ has a subset $B$ such that $G^p[B\cup\{v\}]$ is isomorphic to a graph in $\mathcal{F}$, $\mathcal{A}_v\neq\emptyset$ and $\mathbf{x}(v)$ has a nonzero entry.
    For each $B\in\mathcal{A}_v$, let $B^*$ be the vertex set of the component of $(G\setminus S)^p[(B\setminus S)\cup\{v\}]$ having $v$.
    
    Since $\abs{L}\geq m=2^{2^{d^2/2}\cdot(r+2)^{sd}}\cdot\xi+1$ and $\ell\leq2^{d^2/2}\cdot(r+2)^{sd}$, by the pigeonhole principle, $L$ has a subset $\kappa_1$ such that $\abs{\kappa_1}\geq\xi+1$ and $\mathbf{x}(v)=\mathbf{x}(w)$ for all $v,w\in\kappa_1$.
    Let $z$ be an arbitrary vertex in $\kappa_1$.
    
    We show that $\alpha_{p,r}^\mathcal{F}(G,A)\geq k$ if and only if $\alpha_{p,r}^\mathcal{F}(G,A\setminus\{z\})\geq k$.
    The backward direction is obvious.
    Suppose that $G$ has a $(p,r,\mathcal{F})$-packing $I$ of $A$ in $G$ having size at least $k$.
    We may assume that $z$ is contained in some $B_z\in I$, because otherwise $I$ is also a $(p,r,\mathcal{F})$-packing of $A\setminus\{z\}$.
    In particular, there exist a graph $H\in\mathcal{G}_t$ for some $t\leq d$ and an isomorphism $\psi_z:(B_z\setminus S)\cup\{z\}\to[t]$ between $(G\setminus S)^p[(B_z\setminus S)\cup\{z\}]$ and $H$ where $\psi_z(z)=1$.
    To show that $\alpha_{p,r}^\mathcal{F}(G,A\setminus\{z\})\geq k$, it suffices to show that there exist a vertex $z'\in\kappa_1\setminus\{z\}$ and a set $B_{z'}\subseteq A\setminus\{z\}$ such that $z'\in B_{z'}$ and $(I\setminus\{B_z\})\cup\{B_{z'}\}$ is a $(p,r,\mathcal{F})$-packing of $A\setminus\{z\}$ in $G$ having the same size as $I$.
    
    Suppose for contradiction that no such $z'$ exists.
    It means that for each $v\in\kappa_1\setminus\{z\}$, if $A\setminus\{z\}$ has a subset $B$ such that $v\in B$ and $G^p[B]$ is isomorphic to a graph in $\mathcal{F}$, then $I\setminus\{B_z\}$ contains an element $B'$ with $\dist_G(B,B')\leq r$, because otherwise we can substitute $B_z$ with $B$ from~$I$.
    For each $v\in\kappa_1\setminus\{z\}$, there exist $B_v\in\mathcal{A}_v$ and
    \begin{enumerate}
        \item[$\bullet$]    an isomorphism $\psi_v:(B_v\setminus S)\cup\{v\}\to[t]$ between $(G\setminus S)^p[(B_v\setminus S)\cup\{v\}]$ and $H$ where $\psi_v(v)=1$ and for each $j\in[t]$, $\rho_{r+1}^G[\psi_v^{-1}(j),S]=\rho_{r+1}^G[\psi_z^{-1}(j),S]$.
    \end{enumerate}
    For each $v\in\kappa_1\setminus\{z\}$, let $f(v):=B_v^*\cup(B_z\setminus B_z^*)$.
    
    To derive a contradiction, we do the following steps.
    \begin{enumerate}
        \item[(1)]  Find a set $\kappa_4\subseteq\kappa_1\setminus\{z\}$ such that for each $u\in\kappa_4$, $G^p[f(u)]$ is isomorphic to $G^p[B_z\cup\{z\}]$ and $I$ contains an element $C_u$ with $\dist_G(f(u),C_u)\leq r$ and $\dist_G(C_u,S)>\lfloor r/2\rfloor$.
        \item[(2)]  Show that $\kappa_4$ contains distinct vertices $v$ and $v'$ with $\dist_{G\setminus S}(v,v')\leq r'$.
    \end{enumerate}
    Since $\kappa_4\subseteq L$ is distance-$r'$ independent in $G\setminus S$, these steps derive a contradiction.
    
    Let $B_z^*$ be the vertex set of the component of $(G\setminus S)^p[(B_z\setminus S)\cup\{z\}]$ having $z$.
    Note that for vertices $v,w\in\kappa_1$, $\psi_v^{-1}\circ\psi_w$ is an isomorphism between $(G\setminus S)^p[(B_w\setminus S)\cup\{w\}]$ and $(G\setminus S)^p[(B_v\setminus S)\cup\{v\}]$ assigning $w$ to $v$.
    Thus, $\psi_v^{-1}\circ\psi_z(B_z^*)=B_v^*$.
    
    For the first step, we will use the following three claims.
    The proofs of Claims~\ref{clm:pck new0-1} and~\ref{clm:pck new0-2} are similar to that of Claims~\ref{clm:new2-1} and~\ref{clm:new2}, respectively.
    
    \begin{CLM}\label{clm:pck new0-1}
        For vertices $v,w\in\kappa_1$, $\psi_w^{-1}\circ\psi_v$ is an isomorphism between $G^p[(B_v\setminus S)\cup\{v\}]$ and $G^p[(B_w\setminus S)\cup\{w\}]$.
    \end{CLM}
    \begin{subproof}
        It suffices to show that for $i,j\in[t]$, $\psi_v^{-1}(i)$ is adjacent to $\psi_v^{-1}(j)$ in $G^p$ if and only if $\psi_w^{-1}(i)$ is adjacent to $\psi_w^{-1}(j)$ in $G^p$.
        Suppose that $\psi_v^{-1}(i)$ is adjacent to $\psi_v^{-1}(j)$ in $G^p$.
        Since $\psi_w^{-1}\circ\psi_v$ is an isomorphism between $(G\setminus S)^p[(B_v\setminus S)\cup\{v\}]$ and $(G\setminus S)^p[(B_w\setminus S)\cup\{w\}]$, we may assume that $\psi_v^{-1}(i)$ and $\psi_v^{-1}(j)$ are nonadjacent in $(G\setminus S)^p[(B_v\setminus S)\cup\{v\}]$.
        Thus, if $G$ has a path of length at most $p$ between $\psi_v^{-1}(i)$ and $\psi_v^{-1}(j)$, then the path must have a vertex in $S$.
        
        We take an arbitrary path $Q$ of $G$ between $\psi_v^{-1}(i)$ and $\psi_v^{-1}(j)$ having length at most $p$.
        Let $q_i$ and $q_j$ be the vertices in $V(Q)\cap S$ such that each of $\dist_Q(\psi_v^{-1}(i),q_i)$ and $\dist_Q(\psi_v^{-1}(j),q_j)$ is minimum.
        Such $q_i$ and $q_j$ exist, because $Q$ has a vertex in $S$.
        Let $Q_i$ be the subpath of $Q$ between $\psi_v^{-1}(i)$ and $q_i$, and $Q_j$ be the subpath of $Q$ between $\psi_v^{-1}(j)$ and $q_j$.
        Note that both $Q_i$ and $Q_j$ are $S$-avoiding paths of length at most $p\leq r+1$.
        
        Since $\{v,w\}\subseteq\kappa_1$, $\rho_{r+1}^G[\psi_v^{-1}(i),S]=\rho_{r+1}^G[\psi_w^{-1}(i),S]$.
        Thus, $G$ has an $S$-avoiding path $Q'_i$ between $\psi_w^{-1}(i)$ and $q_i$ whose length is at most that of $Q_i$.
        Similarly, $G$ also has an $S$-avoiding path $Q'_j$ between $\psi_w^{-1}(j)$ and $q_j$ whose length is at most that of $Q_j$.
        By substituting $Q_i$ and $Q_j$ with $Q'_i$ and $Q'_j$ from $Q$, respectively, we obtain a walk of $G$ between $\psi_w^{-1}(i)$ and $\psi_w^{-1}(j)$ whose length is at most $p$.
        Therefore, $\psi_w^{-1}(i)$ is adjacent to $\psi_w^{-1}(j)$ in $G^p$, and this proves the claim.
    \end{subproof}
    
    \begin{CLM}\label{clm:pck new0-2}
        $\kappa_1$ contains at most $d^2/4$ vertices $v$ such that $G^p[f(v)]$ is not isomorphic to $G^p[B_z]$.
    \end{CLM}
    \begin{subproof}
        For vertices $u\in\kappa_1\setminus\{z\}$ and $i\in\psi_z(B_z^*)$, since $\{u,z\}\subseteq\kappa_1$, $\rho_{r+1}^G[\psi_u^{-1}(i),S]$ and $\rho_{r+1}^G[\psi_z^{-1}(i),S]$ are same.
        Therefore, for each $w\in S$, $\psi_u^{-1}(i)$ is adjacent to $w$ in $G^p$ if and only if $\psi_z^{-1}(i)$ is adjacent to $w$ in $G^p$.
        By Claim~\ref{clm:pck new0-1}, the restriction of $\psi_u^{-1}\circ\psi_z$ on $B_z^*$ is an isomorphism between $G^p[B_z^*]$ and $G^p[B_u^*]$.
        
        We first show that for all vertices $v\in\kappa_1$, $i\in\psi_z(B_z^*)$, and $w\in B_z\setminus(B_z^*\cup S)$, if $\psi_z^{-1}(i)$ is adjacent to $w$ in $G^p$, then $\psi_v^{-1}(i)$ is adjacent to $w$ in $G^p$.
        Suppose that $\psi_z^{-1}(i)$ is adjacent to $w$ in $G^p$.
        We take an arbitrary path $Q'$ of $G$ between $\psi_z^{-1}(i)$ and $w$ having length at most $p$.
        Since $(G\setminus S)^p[B_z^*]$ is a component of $(G\setminus S)^p[(B_z\setminus S)\cup\{z\}]$ having $z$ and $w\notin B_z^*$, $Q'$ must have a vertex in $S$.
        
        Let $q$ be the the vertex in $V(Q')\cap S$ such that $\dist_{Q'}(\psi_z^{-1}(i),q)$ is minimum.
        Let $Q'_1$ be the subpath of $Q'$ between $\psi_z^{-1}(i)$ and $q$.
        Note that $Q'_1$ is an $S$-avoiding path of length at most $p\leq r+1$.
        Since $\rho_{r+1}^G[\psi_v^{-1}(i),S]=\rho_{r+1}^G[\psi_z^{-1}(i),S]$, $G$ has an $S$-avoiding path $Q'_2$ between $\psi_v^{-1}(i)$ and $q$ having length at most that of $Q'_1$.
        By substituting $Q'_1$ with $Q'_2$ from $Q'$, we obtain a walk of $G$ between $\psi_v^{-1}(i)$ and $w$ having length at most $p$.
        Hence, $\psi_v^{-1}(i)$ is adjacent to $w$ in $G^p$.
        
        Thus, there is no pair of vertices $i\in\psi_z(B_z^*)$ and $w\in B_z\setminus(B_z^*\cup S)$ such that $\psi_z^{-1}(i)$ is adjacent to $w$ in $G^p$ and $\psi_u^{-1}(i)$ is nonadjacent to $w$ in $G^p$.
        
        We now show that if there exist vertices $i\in\psi_z(B_z^*)$ and $w\in B_z\setminus(B_z^*\cup S)$ such that $\psi_z^{-1}(i)$ is nonadjacent to $w$ in $G^p$, then $\kappa_1$ contains at most one vertex $x$ such that $\psi_x^{-1}(i)$ is adjacent to $w$ in $G^p$.
        To prove the claim, it suffices to show this statement, because $\abs{B_z^*}\cdot\abs{B_z\setminus(B_z^*\cup S)}\leq d^2/4$.
        
        Suppose for contradiction that there exist $i\in\psi_z(B_z^*)$, $w\in B_z\setminus(B_z^*\cup S)$, and distinct $x,x'\in\kappa_1$ such that $\psi_z^{-1}(i)$ is nonadjacent to $w$ in $G^p$ and both $\psi_x^{-1}(i)$ and $\psi_{x'}^{-1}(i)$ are adjacent to $w$ in $G^p$.
        Then $G$ has paths $R$ and $R'$ of length at most $p$ from $w$ to $\psi_x^{-1}(i)$ and $\psi_{x'}^{-1}(i)$, respectively.
        
        We first verify that $R$ or $R'$ has a vertex in $S$.
        Suppose not.
        Since $\abs{B_x^*}\leq d$, $G\setminus S$ has a path $R_1$ of length at most $p(d-1)$ between $x$ and $\psi_x^{-1}(i)$.
        Similarly, $G\setminus S$ has a path $R'_1$ of length at most $p(d-1)$ between $x'$ and $\psi_{x'}^{-1}(i)$.
        Since neither $R$ nor $R'$ has a vertex in $S$, by concatenating $R_1$, $R$, $R'$, and $R'_1$, we obtain a walk of $G\setminus S$ of length at most $2pd\leq r'$ between $x$ and $x'$, contradicting the assumption that $L$ is distance-$r'$ independent in $G\setminus S$.
        Hence, $R$ or $R'$ has a vertex in $S$.
        By symmetry, we may assume that $R$ has a vertex in $S$.
        
        Let $t$ be the vertex in $V(R)\cap S$ such that $\dist_R(\psi_x^{-1}(i),t)$ is minimum.
        Let $R_0$ be the subpath of $R$ between $\psi_x^{-1}(i)$ and $t$.
        Note that $R_0$ is an $S$-avoiding path of length at most $p\leq r+1$.
        Since $\rho_{r+1}^G[\psi_x^{-1}(i),S]=\rho_{r+1}^G[\psi_z^{-1}(i),S]$, $G$ has an $S$-avoiding path $R'_0$ between $\psi_z^{-1}(i)$ and $t$ having length at most that of $R_0$.
        By substituting $R_0$ with $R'_0$ from $R$, we obtain a walk of $G$ between $\psi_z^{-1}(i)$ and $w$ having length at most $p$, contradicting the assumption that $\psi_z^{-1}(i)$ is nonadjacent to $w$ in $G^p$, and this proves the claim.
    \end{subproof}
    
    Since $\abs{\kappa_1}\geq d\cdot(f_{\mathrm{cl}}(r',\varepsilon)\cdot f_{\mathrm{dual}}(\lfloor r/2\rfloor,d,\varepsilon)^\varepsilon\cdot k^{2\varepsilon}+s+d^2/4+1)+1$, by Claim~\ref{clm:pck new0-2}, $\kappa_1\setminus\{z\}$ has a subset $\kappa_2$ of size at least $d\cdot(f_{\mathrm{cl}}(r',\varepsilon)\cdot f_{\mathrm{dual}}(\lfloor r/2\rfloor,d,\varepsilon)^\varepsilon\cdot k^{2\varepsilon}+s+1)$ such that for each vertex $u\in\kappa_2$, $G^p[f(u)]$ is isomorphic to $G^p[B_z]$, which is isomorphic to a graph in $\mathcal{F}$.

    For each $u\in\kappa_2$, since $f(u)\subseteq A\setminus\{z\}$, by assumption, $I\setminus\{B_z\}$ contains an element $C_u$ with $\dist_G(f(u),C_u)\leq r$.
    We take an arbitrary path $P_u$ of $G$ between $b_u\in f(u)$ and $c_u\in C_u$ having length at most $r$.
    Since $\{B_z,C_u\}\subseteq I$ which is a $(p,r,\mathcal{F})$-packing of $A$ in $G$, $\dist_G(B_z\setminus B_z^*,C_u)\geq\dist_G(B_z,C_u)>r$.
    Thus, $b_u\in f(u)\setminus(B_z\setminus B_z^*)=B_u^*$.
    
    We will use the following claim.
    
    \begin{CLM}\label{clm:pck new1}
        For each $u\in\kappa_2$, $V(P_u)\cap S=\emptyset$.
    \end{CLM}
    \begin{subproof}
        Suppose for contradiction that for some $u\in\kappa_2$, $V(P_u)\cap S\neq\emptyset$.
        Let $q$ be the vertex in $V(P_u)\cap S$ such that $\dist_{P_u}(b_u,q)$ is minimum.
        Let $P_1$ be the subpath of $P_u$ between $b_u$ and $q$.
        Note that $P_1$ is an $S$-avoiding path of length at most~$r$.
        Since $\{u,z\}\subseteq\kappa_1$, $G$ has an $S$-avoiding path $P_2$ between $\psi_z^{-1}\circ\psi_u(b_u)$ and $q$ having length at most that of $P_1$.
        By substituting $P_1$ with $P_2$ from $P_u$, we obtain a walk of $G$ between $\psi_z^{-1}\circ\psi_u(b_u)\in B_z$ and $c_u$ having length at most $r$, contradicting the assumption that $\dist_G(B_z,C_u)>r$.
    \end{subproof}
    
    Since $L$ is distance-$r'$ independent in $G\setminus S$ and $2r\leq r'$, by Claim~\ref{clm:pck new1}, $c_u\neq c_{u'}$ for distinct $u,u'\in\kappa_2$.
    Since $\abs{\kappa_2}\geq d\cdot(f_{\mathrm{cl}}(r',\varepsilon)\cdot f_{\mathrm{dual}}(\lfloor r/2\rfloor,d,\varepsilon)^\varepsilon\cdot k^{2\varepsilon}+s+1)$ and every element in $I$ contains at most $d$ vertices, there is a set $\kappa_3\subseteq\kappa_2$ of size at least $f_{\mathrm{cl}}(r',\varepsilon)\cdot f_{\mathrm{dual}}(\lfloor r/2\rfloor,d,\varepsilon)^\varepsilon\cdot k^{2\varepsilon}+s+1$ such that $C_u\neq C_{u'}$ for all distinct $u,u'\in\kappa_3$.
    
    Let $\kappa'_3$ be the set of vertices $u\in\kappa_3$ with $\dist_G(C_u,S)\leq\lfloor r/2\rfloor$.
    Since $I$ is a $(p,r,\mathcal{F})$-packing of $A$ in $G$, for all distinct $u,u'\in\kappa_3$, $\dist_G(C_u,C_{u'})>r$.
    Thus, we deduce that $\abs{\kappa'_3}\leq\abs{S}\leq s$.
    Let $\kappa_4:=\kappa_3\setminus\kappa'_3$.
    Note that $\abs{\kappa_4}\geq f_{\mathrm{cl}}(r',\varepsilon)\cdot f_{\mathrm{dual}}(\lfloor r/2\rfloor,d,\varepsilon)^\varepsilon\cdot k^{2\varepsilon}+1$.
    
    We now show that $\kappa_4$ contains distinct vertices $v$ and $v'$ with $\dist_{G\setminus S}(v,v')\leq r'$.
    For each $u\in\kappa_4$, since $G^p[C_u]$ is isomorphic to a graph in $\mathcal{F}$ and $X_{\mathrm{cl}}$ is a $(p,\lfloor r/2\rfloor,\mathcal{F})$-cover of $A$ in $G$, $G$ has a path $R_u$ of length at most $\lfloor r/2\rfloor$ between some $y_u\in C_u$ and $x_u\in X_{\mathrm{cl}}$.
    Since $u\notin\kappa'_3$, $V(R_u)\cap S=\emptyset$.
    Since $G^p[C_u]$ is isomorphic to a connected graph in $\mathcal{F}$, $G$ has a path $Q_u$ of length at most $p(d-1)$ between $c_u$ and $y_u$.
    More specifically, $Q_u$ is a concatenation of $Q_u^1,\ldots,Q_u^{t'}$ for $t'\leq d-1$ such that for each $i\in[t']$, the length of $Q_u^i$ is at most $p$ and the ends of $Q_u^i$ are in $C_u$.
    Since $p\leq2\lfloor r/2\rfloor+1$, for some $j\in[t']$, if $V(Q_u^j)\cap S\neq\emptyset$, then $\dist_G(C_u,S)\leq\lfloor r/2\rfloor$, contradicting that $u\notin\kappa'_3$.
    Thus, $V(Q_u)\cap S=\emptyset$.
    By Claim~\ref{clm:pck new1}, $V(P_u)\cap S=\emptyset$.
    Since $(G\setminus S)^p[B_u^*]$ is connected and $\abs{B_u^*}\leq d$, $G\setminus S$ has a path $O_u$ of length at most $p(d-1)$ between $u$ and $b_u$.
    By concatenating $O_u$, $P_u$, $Q_u$, and $R_u$, we obtain a walk of $G\setminus S$ between $u$ and $x_u$ having length at most
    \begin{align*}
        &\abs{E(O_u)}+\abs{E(P_u)}+\abs{E(Q_u)}+\abs{E(R_u)}\\
        &\leq p(d-1)+r+p(d-1)+\lfloor r/2\rfloor=2p(d-1)+r+\lfloor r/2\rfloor\leq\lfloor r'/2\rfloor.
    \end{align*}
    Let $W_u$ be a path of $G\setminus S$ between $u$ and $x_u$ consisting of edges of the walk.
    Let $w_u$ be the vertex in $V(W_u)\cap X_{\mathrm{cl}}$ such that $\dist_{W_u}(u,w_u)$ is minimum.
    Such $w_u$ exists, because $x_u\in X_{\mathrm{cl}}$.
    Note that the subpath of $W_u$ between $u$ and $w_u$ is an $X_{\mathrm{cl}}$-avoiding path of length at most $\lfloor r'/2\rfloor$.
    Thus, $w_u$ is contained in $M_{r'}^G(u,X_{\mathrm{cl}})$.
    Since $\{u,z\}\subseteq\kappa_1\subseteq\lambda$ where $\lambda$ is an equivalence class of $\sim$, $M_{r'}^G(u,X_{\mathrm{cl}})$ and $M_{r'}^G(z,X_{\mathrm{cl}})$ are same.
    Therefore, $w_u\in M_{r'}^G(z,X_{\mathrm{cl}})$.
    
    Since $\abs{\kappa_4}\geq f_{\mathrm{cl}}(r',\varepsilon)\cdot f_{\mathrm{dual}}(\lfloor r/2\rfloor,d,\varepsilon)^\varepsilon\cdot k^{2\varepsilon}+1\geq\abs{M_{r'}^G(z,X_{\mathrm{cl}})}+1$, by the pigeonhole principle, there are distinct $v,v'\in\kappa_4$ with $w_v=w_{v'}$.
    By concatenating $W_v$ and $W_{v'}$, we obtain a walk of $G\setminus S$ between $v$ and $v'$ having length at most $r'$, contradicting the assumption that $L$ is distance-$r'$ independent in $G\setminus S$.
    
    Therefore, there are a vertex $z'\in\kappa_1\setminus\{z\}$ and a set $B_{z'}\subseteq A\setminus\{z\}$ such that $z'\in B_{z'}$ and $(I\setminus\{B_z\})\cup\{B_{z'}\}$ is a $(p,r,\mathcal{F})$-packing of $A\setminus\{z\}$ in $G$ having the same size as $I$.
    We conclude the proof by scaling $\varepsilon$ to $\varepsilon/C$ throughout the reasoning.
\end{proof}

We present a linear kernel for the \textsc{Annotated $(p,r,\mathcal{F})$-Packing} problem on every class of graphs with bounded expansion, which generalizes the linear kernel of~\cite{PS2021}.

\begin{THM}\label{thm:ker2'}
    For every class $\mathcal{C}$ of graphs with bounded expansion, there exists a function $f_{\mathrm{pck}}:\mathbb{N}\times\mathbb{N}\to\mathbb{N}$ satisfying the following.
    For every nonempty family $\mathcal{F}$ of connected graphs with at most~$d$ vertices and $p,r\in\mathbb{N}$ with $p\leq2\lfloor r/2\rfloor+1$, there exists a polynomial-time algorithm that given a graph $G\in\mathcal{C}$, $A\subseteq V(G)$, and $k\in\mathbb{N}$, either
    \begin{enumerate}
        \item[$\bullet$]    correctly decides that $\alpha_{p,r}^\mathcal{F}(G,A)>k$, or
        \item[$\bullet$]    outputs sets $Y\subseteq V(G)$ of size at most $f_{\mathrm{pck}}(r,d)\cdot k$ and $Z\subseteq A\cap Y$ such that $\alpha_{p,r}^\mathcal{F}(G,A)\geq k$ if and only if $\alpha_{p,r}^\mathcal{F}(G[Y],Z)\geq k$.
    \end{enumerate}
\end{THM}

To prove Theorem~\ref{thm:ker2'}, we will use the following lemma.

\begin{LEM}\label{lem:rd'}
    For every class $\mathcal{C}$ of graphs with bounded expansion, there exists a function $f_{\mathrm{rd}}:\mathbb{N}\times\mathbb{N}\to\mathbb{N}$ satisfying the following.
    For every nonempty family $\mathcal{F}$ of connected graphs with at most~$d$ vertices and $p,r\in\mathbb{N}$ with $p\leq2\lfloor r/2\rfloor+1$, there exists a polynomial-time algorithm that given a graph $G\in\mathcal{C}$, $A\subseteq V(G)$, and $k\in\mathbb{N}$, either
    \begin{enumerate}
        \item[$\bullet$]    correctly decides that $\alpha_{p,r}^\mathcal{F}(G,A)>k$, or
        \item[$\bullet$]    outputs a set $Z\subseteq A$ of size at most $f_{\mathrm{rd}}(r,d)\cdot k$ such that $\alpha_{p,r}^\mathcal{F}(G,A)\geq k$ if and only if $\alpha_{p,r}^\mathcal{F}(G,A\setminus\{z\})\geq k$.
    \end{enumerate}
\end{LEM}
\begin{proof}
    It easily follows from the proofs of Lemmas~\ref{lem:rd} and~\ref{lem:rd recursion} by setting $\xi:=d\cdot(f_{\mathrm{cl}}(r')+s+1+d/4)$ and replacing Lemmas~\ref{lem:proj} and~\ref{lem:cl} and Proposition~\ref{prop:approx2} with Lemma~\ref{lem:proj'} and~\ref{lem:cl'} and Proposition~\ref{prop:approx1}, respectively.
\end{proof}

We can now easily prove Theorem~\ref{thm:ker2'} from the proof of Theorem~\ref{thm:ker2} by replacing Lemmas~\ref{lem:rd} and~\ref{lem:pth} with Lemmas~\ref{lem:rd'} and~\ref{lem:pth'}, respectively.

\subsection{Proof of the second main theorem}

The kernel of Theorem~\ref{thm:ker2} always returns an instance $(G[Y],Z,k)$ for the \textsc{Annotated $(p,r,\mathcal{F})$-Packing} problem, even if the kernel gets an input $(G,V(G),k)$, which is basically an input for the \textsc{$(p,r,\mathcal{F})$-Packing}.
From the resulting instance $(G[Y],Z,k)$, however, we can prove Theorem~\ref{thm:ker2} by constructing an instance $(G',V(G'),k+1)$ of size $O(k^{1+\varepsilon})$ which is equivalent to $(G,V(G),k)$ in polynomial time.

\pcknowhere*

\begin{proof}
    Let $d$ be the maximum order of a graph in $\mathcal{F}$.
    We apply Theorem~\ref{thm:ker2} for $(G,V(G),k)$.
    We may assume that the first outcome of Theorem~\ref{thm:ker2} does not arise.
    Let $(G[Y],Z,k)$ be the resulting instance, and $F$ be a graph in $\mathcal{F}$ such that $\abs{V(F)}$ is minimum.
    Note that $\alpha_{p,r}^\mathcal{F}(G)\geq k$ if and only if $\alpha_{p,r}^\mathcal{F}(G[Y],Z)\geq k$.
    
    Suppose that $r\leq1$.
    Note that $p\leq1$.
    Then for every set $B\subseteq Z$, $G[Y]^p[B]$ is isomorphic to $G[Z]^p[B]$.
    Since $r\leq 1$, every $(p,r,\mathcal{F})$-packing of $Z$ in $G[Y]$ is also a $(p,r,\mathcal{F})$-packing of $G[Z]$, and vice versa.
    Hence, $\alpha_{p,r}^\mathcal{F}(G[Y],Z)=\alpha_{p,r}^\mathcal{F}(G[Z])$.
    Let $G'$ be the disjoint union of $G[Z]$ and $F$.
    Note that $G'$ has at most $\abs{Z}+d\leq\abs{Y}+d$ vertices.
    Since $\abs{Y}=O(k^{1+\varepsilon})$, one can choose the function $g_{\mathrm{pck}}(r,d,\varepsilon)$ such that $\abs{V(G')}\leq g_{\mathrm{pck}}(r,d,\varepsilon)\cdot k^{1+\varepsilon}$.
    
    If $p=0$ and $\abs{V(F)}\geq2$, then $\alpha_{p,r}^\mathcal{F}(G)=0$, so we report that $\alpha_{p,r}^\mathcal{F}(G)=0$.
    Otherwise, either $p=0$ and $\abs{V(F)}=1$, or $p=1$.
    In both cases, the following hold.
    \begin{enumerate}
        \item[$\bullet$]    For every $(p,r,\mathcal{F})$-packing $I$ of $Z$ in $G[Y]$, $I\cup\{V(F)\}$ is a $(p,r,\mathcal{F})$-packing of $G'$.
        \item[$\bullet$]    For every maximal $(p,r,\mathcal{F})$-packing $I$ of $G'$, by the minimality of $F$, $I$ should contain $\{V(F)\}$, and $I\setminus\{V(F)\}$ is a $(p,r,\mathcal{F})$-packing of $Z$ in $G[Y]$.
    \end{enumerate}
    Therefore, $\alpha_{p,r}^\mathcal{F}(G)\geq k$ if and only if $\alpha_{p,r}^\mathcal{F}(G[Y],Z)\geq k$ if and only if $\alpha_{p,r}^\mathcal{F}(G')\geq k+1$.
    Hence, the statement holds for $r\leq1$.
    
    Thus, we may assume that $r\geq2$.
    Suppose that $p=0$.
    Since $G^0$ is edgeless, if $\abs{V(F)}\geq2$, then $\alpha_{0,r}^\mathcal{F}(G)=0$, so we report that $\alpha_{0,r}^\mathcal{F}(G)=0$.
    Thus, we may assume that $\abs{V(F)}=1$.
    We construct a graph $G'$ from $G[Y]$ as follows.
    \begin{enumerate}
        \item[$\bullet$]    Add two new vertices $h$ and $h'$, and connect them by a path of length $\lceil r/2\rceil$.
        \item[$\bullet$]    For each vertex $v\in Y\setminus Z$, connect $h$ and $v$ by a path of length $\lfloor r/2\rfloor$.
    \end{enumerate}
    Let $G'$ be the resulting graph.
    Since $r\geq2$, $G'$ is well defined.
    Note that $G'$ can be constructed in polynomial time, and
    \begin{equation*}
        \abs{V(G')}\leq1+\abs{Z}+\lfloor r/2\rfloor\cdot\abs{Y\setminus Z}+\lceil r/2\rceil\leq1+\abs{Y}+\lceil r/2\rceil.
    \end{equation*}
    Since $\abs{Y}=O(k^{1+\varepsilon})$, one can choose the function $g_{\mathrm{pck}}(r,d,\varepsilon)$ such that $\abs{V(G')}\leq g_{\mathrm{pck}}(r,d,\varepsilon)\cdot k^{1+\varepsilon}$.
    
    Since $\alpha_{0,r}^\mathcal{F}(G)\geq k$ if and only if $\alpha_{0,r}^\mathcal{F}(G[Y],Z)\geq k$, to show that $\alpha_{0,r}^\mathcal{F}(G)\geq k$ if and only if $\alpha_{0,r}^\mathcal{F}(G')\geq k+1$, it suffices to show that $\alpha_{0,r}^\mathcal{F}(G[Y],Z)\geq k$ if and only if $\alpha_{0,r}^\mathcal{F}(G')\geq k+1$.
    
    Firstly, suppose that $\alpha_{0,r}^\mathcal{F}(G[Y],Z)\geq k$.
    Then $G[Y]$ has a $(0,r,\mathcal{F})$-packing $I$ of $Z$ having size at least $k$.
    Note that $\dist_{G'}(h',Z)>r$.
    Thus, $I\cup\{h'\}$ is a $(0,r,\mathcal{F})$-packing of $G'$ having size at least $k+1$.
    Therefore, $\alpha_{0,r}^\mathcal{F}(G')\geq k+1$.
    
    Conversely, suppose that $\alpha_{0,r}^\mathcal{F}(G')\geq k+1$.
    Then $G'$ has a $(0,r,\mathcal{F})$-packing $I$ of $G'$ having size at least $k+1$.
    We verify that $I$ contains at most one element having a vertex in $V(G')\setminus Z$.
    Note that $V(G')\setminus Z=N^{\lfloor r/2\rfloor}_{G'}[h]\cup\{h'\}$.
    Thus, the distance in $G'$ between two vertices in $V(G')\setminus Z$ is at most~$r$.
    Since $I$ is a $(p,r,\mathcal{F})$-packing, $I$ contains at most one element having a vertex in $V(G')\setminus Z$.
    
    Thus, $I$ has a subset $I'$ of size at least $k$ such that each element in $I'$ is a subset of $Z$.
    Since $p=0$ and $G[Y]$ is an induced subgraph of $G'$, $I'$ is a $(0,r,\mathcal{F})$-packing of $Z$ in $G[Y]$.
    Therefore, $\alpha_{0,r}^\mathcal{F}(G[Y],Z)\geq k$.
    Hence, the statement holds for $r\geq2$ and $p=0$.
    
    Thus, we may further assume that $p>0$.
    Let $p':=\lfloor p/2\rfloor$.
    By Lemma~\ref{lem:critical}, one can find in polynomial time a $(p,\mathcal{F})$-critical graph $H$ having at most $d(dp+1)/2$ vertices.
    Let $x$ be a vertex of $H$.
    We construct a graph $G'$ as follows.
    \begin{enumerate}
        \item[$\bullet$]    Take the disjoint union of $G[Y]$ and $H$, and add a new vertex $h$.
        \item[$\bullet$]    For each vertex $v\in Y\setminus Z$, connect $h$ and $v$ by a path $P_v$ of length $\lfloor r/2\rfloor$.
        \item[$\bullet$]    For each vertex $v\in N_H^{p'}[x]$, connect $h$ and $v$ by a path $P_v$ of length $\lceil r/2\rceil$.
    \end{enumerate}
    Let $G'$ be the resulting graph.
    Note that $G'$ can be constructed in polynomial time, and
    \begin{equation*}
        \abs{V(G')}\leq1+\abs{Z}+\lfloor r/2\rfloor\cdot\abs{Y\setminus Z}+\lceil r/2\rceil\cdot\abs{N_H^{p'}[x]}\leq1+\abs{Y}+\lceil r/2\rceil\cdot d(dp+1)/2.
    \end{equation*}
    Since $\abs{Y}=O(k^{1+\varepsilon})$, one can choose the function $g_{\mathrm{pck}}(r,d,\varepsilon)$, such that $\abs{V(G')}\leq g_{\mathrm{pck}}(r,d,\varepsilon)\cdot k^{1+\varepsilon}$.
    
    We verify that for every set $B\subseteq Z$, $G'^p[B]$ is isomorphic to $G[Y]^p[B]$.
    If $G'$ has a path $P$ between two vertices in $Z$ such that $V(P)\setminus Y$ is nonempty, then $P$ should have $h$, and therefore the length of $P$ is at least $2\lfloor r/2\rfloor+2>p$.
    Thus, every path of $G'$ between two vertices in $Z$ having length at most $p$ is also a path of $G[Y]$.
    Therefore, for every set $B\subseteq Z$, $G'^p[B]$ is isomorphic to $G[Y]^p[B]$.
    
    We remark that $\lfloor p/2\rfloor\leq\lfloor r/2\rfloor$.
    Note that $p\leq2\lfloor r/2\rfloor\leq r+1$.
    If $p\leq r$, then $\lfloor p/2\rfloor\leq\lfloor r/2\rfloor$ clearly.
    Otherwise, $r$ should be even, so that $\lfloor p/2\rfloor=\lfloor (r+1)/2\rfloor=\lfloor r/2\rfloor$.
    Thus, $\lfloor p/2\rfloor\leq\lfloor r/2\rfloor$.
    
    We will use the following claim.
    
    \begin{CLM}\label{clm:preserve new}
        For vertices $v,w\in V(H)$, $v$ and $w$ are adjacent in $H^p$ if and only if $v$ and $w$ are adjacent in $G'^p$.
    \end{CLM}
    \begin{subproof}
        Since $H$ is an induced subgraph of $G'$, the forward direction is obvious.
        We show the backward direction.
        Suppose for contradiction that $v$ and $w$ are nonadjacent in $H^p$ and $G'$ has a path $P$ of length at most $p$ between $v$ and $w$.
        Since $v$ and $w$ are nonadjacent in $H^p$, $P$ has a vertex not in $V(H)$.
        Thus, $P$ should have $h$.
        Since $h$ is a cut-vertex of $G'$ and the ends of $P$ are in $V(H)$, $P$ has vertices $v',w'\in N^{p'}_H[x]$ such that the subpath $Q_1$ of $P$ between $v'$ and $w'$ has $h$ as an internal vertex.
        Note that every vertex of $P$ which is not an internal vertex of $Q_1$ is contained in $V(H)$.
        Since both $v'$ and $w'$ are in $N^{p'}_H[x]$, $H$ has a path $Q_2$ of length at most $2p'=2\lfloor p/2\rfloor$ between $v'$ and~$w'$.
        
        Since $2\lfloor p/2\rfloor\leq2\lceil r/2\rceil$ and $Q_2$ has length $2\lceil r/2\rceil$, by substituting $Q_1$ with $Q_2$ from $P$, we obtain a path of $H$ between $v$ and $w$ having length at most $p$, contradicting the assumption that $v$ and $w$ are nonadjacent in $H^p$.
    \end{subproof}
    
    Since $\alpha_{p,r}^\mathcal{F}(G)\geq k$ if and only if $\alpha_{p,r}^\mathcal{F}(G[Y],Z)\geq k$, to show that $\alpha_{p,r}^\mathcal{F}(G)\geq k$ if and only if $\alpha_{p,r}^\mathcal{F}(G')\geq k+1$, it suffices to show that $\alpha_{p,r}^\mathcal{F}(G[Y],Z)\geq k$ if and only if $\alpha_{p,r}^\mathcal{F}(G')\geq k+1$.
    
    Firstly, suppose that $\alpha_{p,r}^\mathcal{F}(G[Y],Z)\geq k$.
    Then $G[Y]$ has a $(p,r,\mathcal{F})$-packing $I$ of $Z$ having size at least $k$.
    Since $H$ is $(p,\mathcal{F})$-critical, $V(H)$ has a subset $X$ such that $H^p[X]$ is isomorphic to a graph in $\mathcal{F}$.
    By Claim~\ref{clm:preserve new}, $H^p[X]$ is isomorphic to $G'^p[X]$.
    Note that $\dist_{G'}(X,Z)>r$.
    Thus, $I\cup\{X\}$ is a $(p,r,\mathcal{F})$ is a $(p,r,\mathcal{F})$-packing of $G'$ having size at least $k+1$.
    Therefore, $\alpha_{p,r}^\mathcal{F}(G')\geq k+1$.
    
    Conversely, suppose that $\alpha_{p,r}^\mathcal{F}(G')\geq k+1$.
    Then $G'$ has a $(p,r,\mathcal{F})$-packing $I$ having size at least $k+1$.
    We verify that $I$ contains at most one element having a vertex in $V(G')\setminus Z$.
    Since $p'=\lfloor p/2\rfloor\leq\lfloor r/2\rfloor$, the distance in $G'$ between two vertices in $N_{G'}^{\lfloor r/2\rfloor}[h]\cup N_{G'}^{p'}[x]$ is at most $r$.
    By Lemma~\ref{lem:guard}, $I$ has no element which is a subset of $V(H)\setminus N_{G'}^{p'}[x]$.
    Therefore, $I$ contains at most one element having a vertex in $V(G')\setminus Z$.
    
    Thus, $I$ has a subset $I'$ of size at least $k$ such that each element in $I'$ is a subset of $Z$.
    Note that for every $B\in I'$, $G'^p[B]$ is isomorphic to $G[Y]^p[B]$.
    Since $G[Y]$ is an induced subgraph of $G'$, $I'$ is a $(p,r,\mathcal{F})$-packing of $Z$ in $G[Y]$.
    Therefore, $\alpha_{p,r}^\mathcal{F}(G[Y],Z)\geq k$.
    Hence, the statement holds for $r\geq2$ and $p>0$, and this completes the proof.
\end{proof}

Similarly, we can show the following.

\begin{THM}\label{thm:ker2' origin}
    For every class $\mathcal{C}$ of graphs with bounded expansion, there exists a function $g_{\mathrm{pck}}:\mathbb{N}\times\mathbb{N}\to\mathbb{N}$ satisfying the following.
    For every nonempty family $\mathcal{F}$ of connected graphs with at most~$d$ vertices and $p,r\in\mathbb{N}$ with $p\leq2\lfloor r/2\rfloor+1$, there exists a polynomial-time algorithm that given a graph $G\in\mathcal{C}$ and $k\in\mathbb{N}$, either
    \begin{enumerate}
        \item[$\bullet$]    correctly decides that $\alpha_{p,r}^\mathcal{F}(G)=0$,
        \item[$\bullet$]    correctly decides that $\alpha_{p,r}^\mathcal{F}(G)>k$, or
        \item[$\bullet$]    constructs a graph $G'$ such that $\abs{V(G')}=O(k)$, and $\alpha_{p,r}^\mathcal{F}(G)\geq k$ if and only if $\alpha_{p,r}^\mathcal{F}(G')\geq k+1$.
    \end{enumerate}
\end{THM}

\providecommand{\bysame}{\leavevmode\hbox to3em{\hrulefill}\thinspace}
\providecommand{\MR}{\relax\ifhmode\unskip\space\fi MR }
\providecommand{\MRhref}[2]{%
  \href{http://www.ams.org/mathscinet-getitem?mr=#1}{#2}
}
\providecommand{\href}[2]{#2}

\end{document}